\PassOptionsToPackage{dvipsnames,svgnames}{xcolor}
\documentclass[a4paper,onecolumn,11pt]{quantumarticle}
\pdfoutput=1

\usepackage[T1]{fontenc}

\usepackage{amsmath,amssymb,amsfonts}
\usepackage{amsthm}
\usepackage{mathtools}
\usepackage{mathrsfs}
\usepackage{bm}
\usepackage{bbm}
\usepackage[bb=boondox]{mathalfa}

\usepackage{thmtools}
\usepackage{thm-restate}
\usepackage{enumitem}

\usepackage{graphicx}
\usepackage{booktabs}
\usepackage{multirow}
\usepackage{arydshln}     
\usepackage{ytableau}

\usepackage{algorithm}
\usepackage{algorithmicx}
\usepackage{algpseudocode}
\usepackage{listings}

\usepackage{textcomp}
\usepackage{manyfoot}
\usepackage[title]{appendix}
\usepackage[makeroom]{cancel}
\usepackage{soul}
\usepackage{url}
\usepackage{xspace}
\usepackage{xparse}
\usepackage{tabto}
\usepackage[kerning=true]{microtype}
\usepackage{multicol}

\usepackage[dvipsnames,svgnames]{xcolor}
\usepackage{xcolor}

\colorlet{MyRed}{FireBrick!50!Crimson}
\colorlet{MyBlue}{DodgerBlue!75!black}
\colorlet{MyGreen}{DarkGreen!85!black}
\colorlet{MyViolet}{DarkMagenta}

\colorlet{MyLightBlue}{DodgerBlue!20}
\colorlet{MyLightGreen}{MyGreen!20}

\colorlet{PrimalColor}{MyBlue}
\colorlet{PrimalFill}{MyLightBlue}
\colorlet{DualColor}{MyRed}

\colorlet{RevColor}{MyRed}
\colorlet{LinkColor}{MediumBlue}

\usepackage[showdeletions]{color-edits} 
\usepackage[normalem]{ulem}            
\usepackage[review]{revision-markup}   

\newcommand{\debug}[1]{#1} 

\usepackage{acronym}

\usepackage[numbers]{natbib}

\usepackage{style}

\usepackage{hyperref}
\hypersetup{
  colorlinks=true,
  linkcolor=LinkColor,
  citecolor=LinkColor,
  urlcolor=LinkColor
}
\usepackage[noabbrev]{cleveref}

\theoremstyle{plain}
\newtheorem{theorem}{Theorem}
\newtheorem{proposition}{Proposition}
\newtheorem{lemma}{Lemma}
\newtheorem{corollary}{Corollary}
\newtheorem{claim}{Claim}
\newtheorem{conjecture}{Conjecture}

\theoremstyle{definition}
\newtheorem{definition}{Definition}

\newtheorem{remark}{Remark}

\newtheorem{condition}{Condition}

\newcommand{\newmacro}[2]{\newcommand{#1}{\debug{#2}}}
\newcommand{\newop}[2]{\DeclareMathOperator{#1}{\debug{#2}}}

\DeclarePairedDelimiter{\braces}{\{}{\}}
\DeclarePairedDelimiter{\bracks}{[}{]}
\DeclarePairedDelimiter{\parens}{(}{)}

\DeclarePairedDelimiter{\abs}{\lvert}{\rvert}
\DeclarePairedDelimiter{\ceil}{\lceil}{\rceil}

\DeclarePairedDelimiter{\clip}{[}{]}

\DeclarePairedDelimiter{\inner}{\langle}{\rangle}
\DeclarePairedDelimiter{\norm}{\lVert}{\rVert}
\DeclarePairedDelimiterXPP{\dnorm}[1]{}{\lVert}{\rVert}{_{\ast}}{#1}
\DeclarePairedDelimiterXPP{\fnorm}[1]{}{\lVert}{\rVert}{_F}{#1}
\DeclarePairedDelimiterXPP{\onenorm}[1]{}{\lVert}{\rVert}{_{1}}{#1}
\DeclarePairedDelimiterXPP{\twonorm}[1]{}{\lVert}{\rVert}{_{2}}{#1}
\DeclarePairedDelimiterXPP{\supnorm}[1]{}{\lVert}{\rVert}{_{\infty}}{#1}
\DeclarePairedDelimiterXPP{\opnorm}[1]{}{\lVert}{\rVert}{_{\operatorname{op}}}{#1}

\DeclareMathOperator*{\argmax}{arg\,max}
\DeclareMathOperator*{\argmin}{arg\,min}

\DeclareMathOperator{\rank}{rank}

\DeclareMathOperator{\range}{range}
\DeclareMathOperator{\bigoh}{\mathcal{O}}		

\newop{\val}{val}

\newmacro{\hmat}{\mathcal{H}}
\DeclareMathOperator{\dist}{dist}		
\DeclarePairedDelimiterX{\setdef}[2]{\{}{\}}{#1:#2}	
\newmacro{\eye}{\mathbb{I}}	
\newmacro{\payoff}{f}
\newmacro{\gradient}{F}
\DeclareMathOperator{\grad}{\nabla}
\newcommand{\defeq}{\coloneqq}
\renewcommand{\a}{\debug{\alpha}}
\renewcommand{\b}{\debug{\beta}}
\newop{\Opt}{Opt}		
\newmacro{\aopt}{\astate^{\ast}}
\newmacro{\bopt}{\bstate^{\ast}}
\newmacro{\jointstar}{\cond{\joint}}
\newmacro{\aalt}{\astate{'}}
\newmacro{\balt}{\bstate{'}}
\newop{\reg}{h}
\newop{\maximizer}{\tau}
\newmacro{\smooth}{\mu}
\newmacro{\spmsparam}{s}
\newmacro{\gradpayoff}{F}
\newmacro{\optim}{\nu}

\newmacro{\constraintset}{\mathcal C}
\newmacro{\constraintf}{\constraintset_{f}}
\newmacro{\constrainto}{\constraintset_{o}}
\newmacro{\regmax}{D}
\newmacro{\smoothradius}{D_h}

\newmacro{\dd}{\:d}

\newcommand{\wilde}{\widetilde}

\newmacro{\aspace}{\mathcal A}
\newmacro{\bspace}{\mathcal B}

\newop{\superop}{\Xi}
\newop{\superopalt}{\Gamma}
\newmacro{\astate}{\alpha}
\newmacro{\bstate}{\beta}
\newmacro{\adim}{n}
\newmacro{\bdim}{m}

\newmacro{\joint}{\Psi}
\newmacro{\jointalt}{\joint'}
\newmacro{\jointbar}{\overline{\joint}}
\newmacro{\jointtilde}{\widetilde{\joint}}
\newmacro{\jointsp}{\mathcal Z}

\newcommand{\tr}[1]{\operatorname{Tr}\bracks*{#1}}
\newcommand{\adj}[1]{#1^{\dagger}}

\newcommand{\sol}[1]{#1^{\ast}}
\newcommand{\cond}[1]{#1^{\ast}}

\newop{\proj}{\Pi}
\newcommand{\projz}{\Pi_{\cond{\jointsp}}}

\newmacro{\diam}{D}
\newmacro{\gap}{\mathcal{G}}

\newop{\primal}{U}
\newop{\dual}{L}

\newcommand{\qsmoothing}{\texttt{q-Smoothing}\xspace}
\newcommand{\itersmooth}{\texttt{IterSmooth}\xspace}
\newcommand{\ogda}{\texttt{OGDA}\xspace}
\newcommand{\ommwu}{\texttt{OMMWU}\xspace}
\newcommand{\oftrl}{\texttt{OFTRL}\xspace}
\newcommand{\mmwu}{\texttt{MMWU}\xspace}

\newcommand{\gameval}[0]{\val(\superop)}
\DeclarePairedDelimiter{\ket}{\lvert}{\rangle}		
\DeclarePairedDelimiter{\bra}{\langle}{\rvert}		
     
\DeclarePairedDelimiterXPP{\exclude}[1]{\mathopen{}\setminus}{\{}{\}}{}{#1}
      
\newcommand{\from}{\colon}

\newmacro{\acone}{h} 
\newmacro{\abar}{\bar{\a}}
\newmacro{\bbar}{\bar{\b}}

\newcommand{\jointopt}[0]{\sol{\joint}}
\newmacro{\constant}{C'}
\newcommand{\aspaceopt}[0]{\cond{\aspace}}
\newcommand{\bspaceopt}[0]{\cond{\bspace}}
\newmacro{\feasi}{\mu}
\newmacro{\feasialt}{\mu'}
\newmacro{\ncone}{\mathcal{N}_{\aopt}}
\newmacro{\mcone}{\mathcal{M}_{\aopt}}
\newmacro{\pcone}{\mathcal{P}}
\newmacro{\feasicone}{\mathcal{K}_{f}}
\newop{\cone}{cone}
\newmacro{\conespace}{\mathcal{S}}
\newmacro{\measure}{p}

\newop{\coefmap}{\varphi}
\newop{\extpt}{ext}
\newmacro{\const}{c}        
\newmacro{\Const}{\rho}     
\newmacro{\coefalt}{\mu}        
\newmacro{\param}{\theta}       
\newmacro{\params}{\Theta}      
\newmacro{\pexp}{p}     
\newmacro{\qexp}{q}     
\newmacro{\rexp}{r}     
\newmacro{\radius}{r}

\newmacro{\mat}{A}		
\newmacro{\matalt}{B}		

\newmacro{\uvec}{u}
\DeclareMathOperator{\diag}{diag}	

\newcommand{\mg}{\succ}		
\newcommand{\mgeq}{\succcurlyeq}		
\newcommand{\ml}{\prec}		
\newcommand{\mleq}{\preccurlyeq}		
\newop{\opt}{opt}

\newmacro{\NE}{\mathrm{NE}}

\newmacro{\BR}{\mathrm{BR}}

\newacro{LHS}{left-hand side}
\newacro{RHS}{right-hand side}
\newacro{iid}[i.i.d.]{independent and identically distributed}

\newacro{NE}{Nash equilibrium}
\newacroplural{NE}[NE]{Nash equilibria}

\newacro{DGF}{distance-generating function}
\newacro{KKT}{Karush\textendash Kuhn\textendash Tucker}
\newacro{SFO}{stochastic first-order oracle}
\newacro{PMD}{Preconditioned Mirror Descent}
\newacro{POVM}{positive operator-valued measurement}
\newacro{SP-MS}{Saddle-Point Metric Subregularity}

\usepackage{tikz}
\usetikzlibrary{arrows.meta}
\usepackage{transparent}

\begin{document}

\title{Breaking $1/\varepsilon$ Barrier in Quantum Zero-Sum Games: Generalizing Metric Subregularity for Spectraplexes}

\author{Yiheng Su}

\email{su228@wisc.edu}
 \author{Emmanouil-Vasileios Vlatakis-Gkaragkounis}
 \email{vlatakis@wisc.edu}
 \author{Pucheng Xiong}
 \email{pxiong79@wisc.edu}
 \affiliation{University of Wisconsin-Madison}
\maketitle

\begin{abstract}
In recent years, quantum zero-sum games have evolved from toy models into a canonical framework for non-local games, quantum interactive proofs, and quantum machine learning. In this setting, two players iteratively mix quantum states to optimize a bilinear payoff induced by a joint measurement. While classical bilinear games over polyhedral domains admit gradient methods with linear last-iterate convergence rates of $\bigoh(\log(1/\varepsilon))$, analogous guarantees in the quantum setting have remained elusive. Recent work by \cite{Vasconcelos2025quadraticspeedupin} established an $\bigoh(1/\varepsilon)$ average-iterate rate, {raising the question of which geometric conditions on the spectraplex support faster last-iterate convergence}.

We establish a H\"older error bound with a finite, game-dependent exponent for every finite-dimensional quantum zero-sum game. If the game admits an equilibrium whose two states are positive definite, this bound is linear. We prove that quantum zero-sum games {under this full-rank hypothesis} admit algorithms with linear last-iterate convergence to {the set of} Nash {equilibria}, matching the asymptotic rate of the classical polyhedral case. In particular, matrix variants of Nesterov’s iterative smoothing and Optimistic Gradient Descent--Ascent (\ogda) achieve $\bigoh(\log(1/\varepsilon))$ {convergence}. The main technical ingredient is a new error-bound theory for semidefinite geometry: we establish metric subregularity of the underlying monotone operator over {spectraplexes} despite the absence of polyhedral structure.

Of independent interest, we develop a global geometric characterization of Nash equilibria in semidefinite games via slack operators, classifying strategic directions as essential, neutral, or non-essential. {If every Nash equilibrium is strictly complementary, as under the stated rank-nondegeneracy condition}, this trichotomy collapses to a sharp classical-style dichotomy.

Finally, we revisit Optimistic Matrix Multiplicative Weights Update
(\ommwu), the state of the art average-iterate method. We show that {it also} admits an $\tilde{\mathcal{O}}(1/\varepsilon)$ last-iterate guarantee when the game has a unique Nash equilibrium whose two states have full rank, by extending the Quantal Response Equilibria (QRE) framework to spectraplex games. On the other hand, we prove that {any  exponential speedup for \ommwu --{$\bigoh(\log(1/\varepsilon))$} last-iterate rate-- must depend on a natural condition number of the game, and hence on the ambient dimensions, revealing a sharp tradeoff between acceleration and conditioning}. Empirically, \ommwu is outperformed by \ogda in the regimes we study, and we leave open the problem of characterizing the spectraplex games in which \ommwu may still exhibit similar provable speed-ups.
\end{abstract}

\tableofcontents
\newpage

\section{Introduction}
Min--max optimization lies at the intersection of optimization, game theory, and machine learning. In its most general form, given an objective function $f : \mathbb{R}^n \times \mathbb{R}^m \to \mathbb{R}$, one seeks to solve
\[
(\chi^*, \psi^*) \;=\; \arg\min_{\chi \in \mathcal{X}} \;\arg\max_{\psi \in \Psi} f(\chi,\psi).
\]
The origins of this problem trace back to young John von Neumann, who in December 1926 presented at the Göttingen mathematics seminar his seminal work on \emph{Gesellschaftsspiele}~\citep{v1928theorie}, introducing what are now called \emph{min–max points}.
When the strategy sets $(\mathcal{X},\Psi)$ are simplices $(\Delta^n, \Delta^m)$, these points correspond to stable outcomes where no player can profit from unilateral deviation---formalized later as \emph{Nash equilibria}.

Nearly a century later, attention has shifted to quantum zero-sum games, where strategies are density matrices and payoffs arise from joint measurements. Originating from the study of non-locality and the EPR paradox~\citep{cshs69, bell64}, this framework has evolved nowadays from a foundational tool for non-local games~\citep{greenberger2007going, mermin90, hardy93,  Aravind:01, reichardt_classical_2013,mipre} and interactive proofs~\citep{Gutoski05,mochon2007quantum} into a critical template for quantum robust machine learning, such as training QGANs~\citep{KilloranQGANs}.

While computing Nash equilibria of \emph{general} quantum games is PPAD-complete~\citep{Bostanci2022quantumgametheory}, the situation is more favorable for the \emph{zero-sum} case. Seminal work by \cite{jain2009parallel}, refined by \cite{Vasconcelos2025quadraticspeedupin}, established that min--max points can be approximated to $\varepsilon$-accuracy in $\bigoh(1/\varepsilon)$ iterations. These average-iterate guarantees leave open which quantum games admit faster last-iterate convergence. Error-bound analyses of classical games~\citep{wei2021linear} suggest examining the relation between duality gap and distance to equilibrium over the spectraplex. Motivated by this, our work asks
\begin{center}
    \begin{minipage}[c]{0.9\linewidth}
        \centering
        \textbf{\textit{{Under what conditions can} gradient-based methods achieve linear convergence,\\[1pt]
i.e., $\bigoh(\log(1/\varepsilon))$ iteration complexity, in quantum zero-sum games?}}
    \end{minipage}%
    \hspace{0.5em} 
    \begin{minipage}[c]{0.05\linewidth}
        \phantomsection\label{main-question}$(\star)$
    \end{minipage}
\end{center}

\subsection{Prior Work and The Geometry of Convergence}
Our work builds on the results of \cite{jain2009parallel} regarding the \emph{Matrix Multiplicative Weights Update} (\texttt{MMWU}). Vasconcelos et al.~\cite{Vasconcelos2025quadraticspeedupin} extended this framework, showing that a class of quantum channel algorithms correspond to a hierarchy of first-order methods  adapted to quantum games, with \emph{Optimistic  \texttt{MMWU}} (\texttt{\ommwu}) achieving the state-of-the-art $\bigoh(1/\varepsilon)$ rate.

To understand the barrier to sharper rates, we must look to the geometry of the problem. A zero-sum game $\min_{x \in \mathcal{X}} \max_{y \in \mathcal{Y}} \langle x, Uy\rangle$ is equivalent to a variational inequality (VI) for a monotone operator $F(z) := (Uy, -U^\top x)$ over the domain $\mathcal{Z} = \mathcal{X} \times \mathcal{Y}$:
\begin{align*}
    \langle F(z^*),\, z - z^* \rangle \;\geq\; 0 \quad \forall\, z \in \mathcal{Z}.
\end{align*}
In the classical case, $\mathcal{Z}$ is a polytope.

By linear programming duality, this theoretically allows for $\bigoh(\log(1/\varepsilon))$ rates via interior-point methods. However, their implementation difficulties and the prohibitive per-iteration cost of such schemes in high dimensions motivates the use of first-order algorithms. While generic gradient methods (e.g., Mirror Prox) typically requires strong convexity for acceleration better than the sublinear $\bigoh(1/\varepsilon)$ rates—recent advances reveal a surprising phenomenon: polyhedral constraints act as implicit regularizers. This geometry suffices to drive optimistic algorithms (e.g., \ogda) to linear last-iterate convergence~\citep{wei2021linear}.

\paragraph{Key Ingredient: Error bounds.}
The structural property enabling linear convergence is the existence of an \emph{error bound}. To formalize this, we distinguish two notions of $\varepsilon$-Nash approximation: \emph{weak} (incentive-based), measured by the duality gap $\gap(z)$---the maximum unilateral payoff improvement---and \emph{strong} (metric), measured by $\dist_{\mathcal{N}\!ash}(z)$---the distance to the Nash equilibrium set.
A zero-sum game satisfies a linear \emph{error bound} with modulus $\delta(\superop) > 0$ if {distance is bounded by a constant multiple of the gap}:
{\setlength{\abovedisplayskip}{6pt}
\setlength{\belowdisplayskip}{6pt}
\begin{equation}
    \gap(z) \;\ge\; \delta(\superop) \, \mathrm{dist}_{\mathcal{N}\!ash}(z)
    \quad \text{for all } z \in \mathcal{Z}.\tag{Error Bound}
    \label{eq:error-bound}
\end{equation}}This condition is the game-theoretic analogue of the classical \emph{Hoffman bound} for linear systems, which bounds the distance to the feasible set by the magnitude of constraint violations. {On a compact strategy space, continuity of payoff already implies that vanishing Nash gap forces the distance to the equilibrium set to vanish. The question is quantitative: whether this implication has a linear rate or a slower rate.}
It has been known since Tseng~\citep{tseng1995} that affine VIs over \emph{polyhedral} sets satisfy such bounds. The central question
\hyperref[main-question]{$(\star)$} thus reduces to a geometric one:
\vspace{-0.25em}
\begin{center}
    \begin{minipage}[c]{0.9\linewidth}
        \centering
        \textbf{\small\textit{Which error bounds hold when the polyhedral simplex is replaced by the spectraplex?}}
    \end{minipage}%
    \hspace{0.5em} 
    \begin{minipage}[c]{0.05\linewidth}
        \phantomsection\label{second-question}$(\#)$
    \end{minipage}
    \vspace{-1em}
\end{center}
{\subsection{Technical Challenges: Spectraplex vs. Simplex.}
\label{sec:challenges}}

\paragraph{The Geometry of Error Bounds.}
The central argument for linear convergence in classical games relies on the geometry of the simplex. As shown in \cite{Gilpin2012FirstOrderAlgorithm}, when feasible sets are polytopes, the epigraph of the duality gap function $z \mapsto \gap(z)$ is polyhedral. {This implies a \emph{finite vertex separation}: In high level, if the support of a Nash equilibrium lies in the relative interior of a face {$F$} of the simplex, then every other vertex {$v\notin F$} is \emph{uniformly} suboptimal: there exists a strictly positive ``vertical'' margin {$t_{\min}>0$} such that {$\gap(v)\ge t_{\min}$} for all such vertices. This \emph{finite vertex separation} yields a Hoffman-type error bound, with modulus {$\kappa$} scaling proportionally to {$t_{\min}$}}.
This intuition does not directly extend over the spectraplex $\mathcal{S} = \{\rho \succeq 0 : \tr{\rho}=1\}$.
First, the set of extreme points (rank-1 projectors) is \emph{uncountable}, allowing sequences of extreme points  $\{z_k\}$ of spectraplex --analogue to vertices-- with $\gap(z_k) \downarrow 0$ without ever reaching the equilibrium set. This destroys the finite vertex separation property (as $t_{\min} \to 0$).
Second, the boundary is \emph{curved} (semi-algebraic). Hence, unlike polytopes, the spectraplex admits infinitely many supporting hyperplanes, and
standard Hoffman-type constants for linear systems do not apply.
From convex geometry perspective, local normal cones can vary continuously along rank-deficient strata--analogue to faces and vertices--, which allows
approach directions along which the gap may decay faster than distance.\footnote{A simple intuition already appears in convex minimization without a linear growth condition. For instance, $f(x)=x^4$ is convex with unique minimizer $x^\star=0$, but
\(
f(x)-f(x^\star)=|x|^4
=\operatorname{dist}(x,\{x^\star\})^4.
\)
Hence
\[
\frac{f(x)-f(x^\star)}
{\operatorname{dist}(x,\{x^\star\})}
=|x|^3\to0
\qquad\text{as }x\to x^\star.
\]
Thus the objective gap may vanish much faster than the solution distance, and no uniform linear error bound of the form
$f(x)-f^\star\ge \kappa\,\operatorname{dist}(x,X^\star)$
need hold. More generally, $f(x)=x^{2m}$ exhibits arbitrarily high-order separation between function gap and solution distance.}
Compactness supplies a positive minimum of the gap-to-distance ratio away from the equilibrium set, but does not control that ratio as one approaches equilibrium. A global linear error bound therefore requires a separate geometric argument.

\emph{Takeaway: $\delta(\superop)$ must be \emph{instance-dependent but $\varepsilon$-independent}: a smaller modulus gives worse conditioning. We use \(\kappa(\superop)=\opnorm{\gradient}/\delta(\superop)\) for the corresponding condition number.}
\paragraph{The Geometry of Nash Equilibria.}
A second obstruction is that the equilibrium set itself is governed by cone geometry rather than by finitely many coordinates. In a simplex, an inactive pure strategy is simply a zero coordinate; over the spectraplex, inactivity is encoded by kernels of positive semidefinite slack operators, so small perturbations can rotate tight eigenspaces or create \emph{neutral} directions that are payoff-tight but never receive mass at equilibrium. Thus support, tightness, and optimality may fail to coincide unless additional conditions such as strict complementarity or nondegeneracy hold. We make this precise in Appendix~\ref{appendix:strict complementarity}, where slack operators yield an essential/neutral/non-essential trichotomy and identify when it collapses to the classical dichotomy. {Crucially, our finite-exponent error bound applies without assumptions on this classification. To obtain a linear error bound and logarithmic Euclidean iteration counts, we instead use the existence of a full-rank equilibrium; the proof exploits its positive definiteness directly}.

{\subsection{Our Contributions}
\label{sec:contributions}}
\noindent Having identified the barriers, we present our main results. We give a sufficient condition for question~\hyperref[main-question]{$(\star)$}: we establish a H\"older error bound for every quantum zero-sum game, strengthen it to a linear bound when some equilibrium has full rank, and show that under this hypothesis \itersmooth and \ogda converge linearly in the last iterate. \\

\paragraph{1. H\"older and Linear Error Bounds on the Spectraplex.}
\vspace{-0.25em}
The engine behind Theorem~\ref{thm:main-euclidean} is a new geometric result. We prove that the monotone operator $F$ associated with a quantum game satisfies saddle-point metric subregularity with a finite exponent over the spectraplex, despite the curvature and infinite extreme points described in Section~\ref{sec:challenges}.

\vspace{-0.45em}
\begin{lemma}[Quantum Error Bound]
\vspace{-0.25em}
    Let $\mathcal{Z}^\star$ be the set of Nash equilibria. There exist a constant $K > 0$ and an integer $q\geq 1$, depending on the game but not on $z$ or on the target accuracy, such that for all $z \in \mathcal{Z}$:
    \[
        \gap(z) \;\geq\; \frac{1}{K} \cdot \mathrm{dist}(z, \mathcal{Z}^\star)^{q}.
\vspace{-0.25em}
    \]
    If there exists an equilibrium whose two states are positive definite, one may take $q=1$.
\end{lemma}
{\emergencystretch=3em
\noindent This establishes a quantitative gap-to-distance relation. Only the linear case $q=1$ is used for the convergence guarantees below. Without the full-rank hypothesis, the result guarantees a finite exponent but not $q=1$.\par}

\paragraph{2. Linear Convergence of Euclidean Methods.}
We then show that matrix variants of standard Euclidean algorithms achieve linear rates under the full-rank equilibrium hypothesis, matching the accuracy dependence in classical polyhedral theory.

\begin{theorem}[Informal]\label{thm:main-euclidean}
    Suppose a quantum zero-sum game admits an equilibrium whose two states are positive definite. Then the matrix adaptations of \emph{Nesterov’s Iterative Smoothing} and \emph{Optimistic Gradient Descent-Ascent (\ogda)} with the stated admissible parameters compute an $\varepsilon$-Nash equilibrium of the game in
    \[
        T = \widetilde \bigoh\left( \log(1/\varepsilon) \right)
    \]
    iterations, where the hidden constant depends on the condition number of the game and the fixed algorithmic parameters, but not on $\varepsilon$.
\vspace{-0.25em}
\end{theorem}

\vspace{-0.55em}

{

\paragraph{3. Entropic Methods: Regularized Guarantees and Uniform-Rate Obstructions.}
While Euclidean \ogda attains linear convergence under the full-rank hypothesis, its admissible step size is controlled by the Lipschitz constant in the Frobenius norm, which may depend on the matrix dimensions. Over the spectraplex, such dependence can be prohibitive. This motivates revisiting the \emph{Optimistic Matrix Multiplicative Weights Update (\ommwu)}, whose entropic geometry is more naturally adapted to the semidefinite domain. We distinguish this unregularized method from the regularized entropic mirror-prox scheme defined in Section~\ref{sec:entropic-methods}.

\begin{theorem}[Informal]\label{thm:main-ommwu}
For quantum zero-sum games with a unique matrix Nash equilibrium whose two states have full rank, regularized entropic mirror-prox with any fixed step size $0<\eta\leq 1$, independent of the matrix dimensions, returns an iterate within Frobenius distance $\varepsilon$ of that equilibrium in $\widetilde{\bigoh}(1/\varepsilon)$ iterations under von Neumann entropy regularization of the payoff, with the regularization strength chosen proportional to the target accuracy. The hidden constants may depend on the game.
\end{theorem}

This result extends to the quantum setting the Quantal Response Equilibrium (QRE) perspective previously developed for classical games~\cite{Sokota2023QRE}. Prior work of \cite{Vasconcelos2025quadraticspeedupin} established an $\bigoh(1/\varepsilon)$ duality-gap guarantee for the \emph{average iterate} of unregularized \ommwu. Our guarantee instead concerns the \emph{last iterate} of the separately defined regularized method and matches this order up to a logarithmic factor, measured in Frobenius distance to the equilibrium. For fixed regularization the method converges linearly to a QRE, while approximation of Nash requires the regularization strength to decrease with the target accuracy. The QRE contraction itself requires neither uniqueness nor full rank; the full-rank hypothesis converts the regularization bias into distance to the equilibrium set, and uniqueness identifies this distance with the distance to a single equilibrium. For background on equilibrium structure over the spectraplex, including slack operators, essential/\allowbreak neutral/\allowbreak non-essential directions, and strict complementarity, see Appendix~\ref{appendix:strict complementarity}; see also \cite{ickstadt2025nash}.


\paragraph{4. Uniform-rate obstruction for unregularized \ommwu.}
The preceding theorem naturally raises the question of whether the entropic geometry of \ommwu can deliver linear convergence without extra regularizations or paying an explicit price for the spectral structure of the game. We ask whether a last-iterate rate can hold uniformly over games, even at fixed dimension. In the classical simplex setting, Cai et al.~\cite{cai2025fastlastiterateconvergencelearning} showed that instance-independent linear convergence is impossible. We show that a comparable obstruction persists in the quantum setting: a diagonal embedding, together with an initialization-specific argument, gives constant-gap excursions at arbitrarily late times across a family of one-qubit games.

\begin{theorem}[Informal Lower Bound for \ommwu]
{Fix \(0<\eta\leq1/2\). For every iteration threshold \(T\geq1\), there exists a quantum zero-sum game with one qubit per player and \(0\preceq U\preceq I\) such that unregularized \ommwu, initialized at the maximally mixed state and run without the stopping test, has an iterate satisfying}
\[
\gap(\joint_t)\geq c_{\mathrm{gap}}>0 \qquad \text{for some } t\geq T.
\]
{The constant \(c_{\mathrm{gap}}\) is absolute.}
\end{theorem}
Thus, unregularized \ommwu cannot have a last-iterate rate tending to zero uniformly over this class. The result permits rates with instance-dependent constants and does not bound the first time the gap becomes small. Since the games in this family have multiple quantum equilibria, it also leaves open Conjecture~\ref{conj:ommwu-linear}, which posits linear last-iterate convergence under a unique quantum equilibrium.

\paragraph{{Equilibrium structure and SDP applications.}}
{Appendix~\ref{appendix:strict complementarity} characterizes equilibrium supports through slack kernels and proves the global essential/neutral/non-essential classification. The neutral class disappears if every equilibrium is strictly complementary. Appendix~\ref{appendix:approximation} establishes the reciprocal SDP--game value identity and reconstructs feasible primal and dual solutions from approximate equilibria. Both this identity and the logarithmic relative-accuracy iteration bound require the associated game to satisfy the full-rank hypothesis in addition to the stated SDP positivity conditions.}

\textbf{Organization. }
Section~\ref{sec:preliminaries} introduces the model and preliminaries. Section~\ref{subsec:spms} proves the error bounds for quantum zero-sum games. Section~\ref{sec:algorithms} uses them to establish linear last-iterate convergence of the Euclidean methods, \itersmooth and \ogda. Section~\ref{sec:entropic-methods} studies the entropic methods: the regularized scheme and the lower bound for unregularized \ommwu. Section~\ref{appendix:convergence experiment} reports experiments, and Section~\ref{sec:epilogue} concludes with a discussion. Omitted proofs and the application to semidefinite programs are deferred to the appendices.

\subsection{AI Disclosure}
We used generative AI models (ChatGPT 6-Astra) to polish the writing and audit the proofs. Most of the mathematical content comes from the authors, while the proof of Claim \ref{claim:spms-semialgebraic} of Theorem \ref{thm:sp-ms-for-spectraplex-appendix} is developed with the assistance of ChatGPT 6-Astra. The authors verified and take responsibility for every
mathematical statement, proof, and citation.
\section{Model and Preliminaries}
\label{sec:preliminaries}
We first fix notation and define the quantum zero-sum games studied in the paper. Let $\hmat^d=\{A\in\mathbb C^{2^d\times 2^d}:A=A^\dagger\}$ denote the Hermitian matrices, equipped with the Hilbert--Schmidt inner product $\inner{A,B}=\tr{A^\dagger B}$ and Frobenius norm $\fnorm{A}=\sqrt{\inner{A,A}}$. The positive semidefinite cone is $\hmat^d_+=\{A\in\hmat^d:A\succeq0\}$. Additional notation and the smoothness proof for the smoothed duality gap are collected in Appendix~\ref{appendix:OmittedPreliminaries}.

\subsection{Quantum States, Measurements, and Payoffs}
\label{subsec:model-payoffs}
The players' strategies are density matrices living in their respective spectraplexes. Alice chooses a state $\astate \in \aspace$ and Bob chooses $\bstate \in \bspace$, where:
$\aspace \coloneqq \setdef*{\astate \in \hmat^{\adim}_{+}}{\tr{\astate} = 1}$ and $\bspace \coloneqq \setdef*{\bstate \in \hmat^{\bdim}_{+}}{\tr{\bstate} = 1}$.
The joint action space is the Cartesian product $\jointsp \coloneqq \aspace \times \bspace$, with joint state $\joint = (\astate, \bstate)$. The spectraplexes are convex, compact sets, and the diameter of $\jointsp$ is bounded by 2 (For details, see \citet[Appendix C.1 and Proposition 5]{Vasconcelos2025quadraticspeedupin}).

The referee measures the tensor product of the joint state $\joint^{\otimes} = \astate \otimes \bstate$ using a Positive Operator-Valued Measure (POVM) $\{P_{\omega}\}_{\omega \in \Omega}$, satisfying $\sum_{\omega} P_{\omega} = \eye$. Each outcome $\omega$ is associated with a utility $\mathcal{U}(\omega) \in [-1, 1]$. We define the \emph{payoff observable} $U$ as $U \coloneqq \sum_{\omega \in \Omega} \mathcal{U}(\omega) P_\omega$.
By the Born rule, the probability of outcome $\omega$ is $p_{\omega}(\joint) = \tr{P_\omega (\astate \otimes \bstate)}$. Consequently, Alice's expected payoff $\payoff: \jointsp \to \mathbb{R}$ is:
\begin{align} \label{eqn:payoff-bilinear}
    \payoff(\joint)
    = \sum_{\omega \in \Omega} \mathcal{U}(\omega)\, p_\omega(\joint)
    = \tr{U (\astate \otimes \bstate)}.
\end{align}
This confirms that the game is bilinear in the players' density matrices.

Following \citet{Vasconcelos2025quadraticspeedupin}, we express the payoff gradients via linear maps between matrix spaces (``superoperators''). Leveraging Choi-Jamiołkowski isomorphism, a linear map $\superop:\hmat^{\bdim}\to \hmat^{\adim}$ admits a unique adjoint $\adj{\superop}:\hmat^{\adim}\to \hmat^{\bdim}$ defined by the duality relation
\begin{equation}\label{eqn:adjoint-def}
\inner{\astate,\superop(\bstate)} \;=\; \inner{\adj{\superop}(\astate),\bstate},
\qquad \forall\,\astate\in\hmat^{\adim},\ \bstate\in\hmat^{\bdim}.
\end{equation}
In our game, these maps arise from partial traces. Specifically, define the interaction superoperator $\superop:\hmat^{\bdim}\to \hmat^{\adim}$ by
\begin{equation}\label{eqn:superoperator}
\superop(\bstate)\;\coloneqq\;\operatorname{Tr}_{\bspace}\!\big[\,U^\dagger(\eye_{\aspace}\otimes\bstate)\,\big],
\qquad
\adj{\superop}(\astate)\;\coloneqq\;\operatorname{Tr}_{\aspace}\!\big[\,U^\dagger(\astate\otimes\eye_{\bspace})\,\big],
\end{equation}
where $\operatorname{Tr}_{\aspace}$ and $\operatorname{Tr}_{\bspace}$ denote the partial traces over $\aspace$ and $\bspace$, respectively.
By cyclicity of the trace, the expected payoff admits the bilinear representation
\[
\payoff(\astate,\bstate)\;=\;\inner{\astate,\superop(\bstate)}
\;=\;\inner{\adj{\superop}(\astate),\bstate}.
\]
Consequently, identifying the individual payoff gradients as $\gradient_\a(\bstate) \defeq \grad_{\astate}\,\payoff(\astate,\bstate) = \superop(\bstate)$ and $\gradient_\b(\astate) \defeq \grad_{\bstate}\,(-\payoff(\astate,\bstate))  = -\adj{\superop}(\astate)$, we define the \emph{joint gradient vector field} $F:\jointsp\to \hmat^{\adim}\times\hmat^{\bdim}$ simply as:
\begin{equation}\label{eqn:payoff-gradient}
    F(\astate,\bstate)
    \;\coloneqq\; \big( \gradient_\a(\bstate),\, \gradient_\b(\astate) \big)
    \;=\; \big( \superop(\bstate),\, -\adj{\superop}(\astate) \big).
\end{equation}
\vspace{-2em}
\subsection{Idealized von Neumann Measurements }
We use the term idealized von Neumann measurement in the sense of a
projective measurement of a Hermitian observable, allowing degenerate
eigenvalues~\cite{BarberenaFriedman2024}. Let $A=A^\dagger$ act on a
finite-dimensional Hilbert space $\mathcal H$, with spectral decomposition:
\begin{equation*}
A=\sum_{\lambda\in\operatorname{spec}(A)}\lambda P_\lambda,
\qquad
P_\lambda^\dagger=P_\lambda,
\qquad
P_\lambda P_\mu=\delta_{\lambda\mu}P_\lambda,
\qquad
\sum_{\lambda\in\operatorname{spec}(A)}P_\lambda=\eye.
\end{equation*}
The sum runs over the distinct eigenvalues of $A$, and $P_\lambda$ is
the orthogonal projector onto the eigenspace $\ker(A-\lambda I)$.
If this eigenspace has dimension $r_\lambda$ and orthonormal basis
$\{|\lambda,j\rangle\}_{j=1}^{r_\lambda}$, then
\begin{equation*}
P_\lambda=\sum_{j=1}^{r_\lambda}
|\lambda,j\rangle\langle\lambda,j|.
\end{equation*}
In particular, $P_\lambda$ is rank one precisely when $\lambda$ is
nondegenerate. For a density matrix $\rho$, the measurement returns outcome $\lambda$
with probability $p_\lambda=\tr{\rho P_\lambda}$. For a pure state
$\rho=|\psi\rangle\langle\psi|$, this becomes
\begin{equation*}
p_\lambda=\|P_\lambda|\psi\rangle\|^2
=\sum_{j=1}^{r_\lambda}|\langle\lambda,j|\psi\rangle|^2.
\end{equation*}
When a post-measurement state is required, we adopt the
L\"uders rule: conditional on an outcome $\lambda$ with $p_\lambda>0$,
the state becomes
$\rho_\lambda=P_\lambda\rho P_\lambda/p_\lambda$~\cite{Luders2006}.
Our payoff analysis uses only the outcome probabilities.

In the quantum game, the referee receives
$\rho=\alpha\otimes\beta$. If outcome $\lambda$ is assigned utility
$\mathcal U(\lambda)$, define the payoff observable
$U=\sum_\lambda\mathcal U(\lambda)P_\lambda$. The expected payoff is then
\begin{equation*}
f(\alpha,\beta)
=\sum_\lambda\mathcal U(\lambda)
  \tr{P_\lambda(\alpha\otimes\beta)}
=\tr{U(\alpha\otimes\beta)}.
\end{equation*}
Thus projective measurements yield the bilinear payoff representation
used throughout the paper.\medskip
\begin{figure}[htbp]
\begin{center}
\par\medskip
\resizebox{\textwidth}{!}{
\begingroup%
\definecolor{GameBlue}{HTML}{326FA8}%
\definecolor{GameGreen}{HTML}{7FA16A}%
\definecolor{GamePurple}{HTML}{8664A6}%
\definecolor{GameInk}{HTML}{202B35}%
\begin{tikzpicture}[
    x=1cm,y=1cm,
    font=\normalfont\fontsize{8}{10}\selectfont,
    text=GameInk,
    panel/.style={rounded corners=3pt,line width=0.65pt},
    title/.style={font=\sffamily\bfseries\fontsize{7.5}{9}\selectfont},
    labeltext/.style={font=\rmfamily\fontsize{7}{9}\selectfont},
    player/.style={font=\sffamily\bfseries\fontsize{7.5}{9}\selectfont},
    flow/.style={-{Stealth[length=3.5pt,width=2.8pt]},draw=GameInk,line width=0.6pt},
    formula/.style={inner sep=0pt,align=center},
    smallformula/.style={formula,font=\normalfont\fontsize{8}{10}\selectfont}
]
\path[use as bounding box] (0,0) rectangle (18.6,5.6);

\foreach \x/\tone/\headtone in {0/GameBlue/GameBlue,6.25/GameGreen/GameGreen!85,12.5/GamePurple/GamePurple!15}{
    \draw[panel,draw=\tone!65,fill=\tone!4] (\x,0) rectangle +(6.1,5.6);
    \begin{scope}
        \clip[rounded corners=3pt] (\x,0) rectangle +(6.1,5.6);
        \fill[\headtone] (\x,5.18) rectangle +(6.1,0.42);
    \end{scope}
}
\node[title,text=white] at (3.05,5.39) {1. PLAYERS' STRATEGIES};
\node[title,text=white] at (9.30,5.39) {2. MEASUREMENT AND PAYOFF};
\node[title,text=GamePurple!85!black] at (15.55,5.39) {3. GRADIENTS VIA PARTIAL TRACES};

\node[player,text=GameBlue] at (1.55,4.97) {Alice};
\node[player,text=GamePurple] at (4.55,4.97) {Bob};
\node[smallformula] at (1.55,4.57)
    {\resizebox{2.78cm}{!}{$\aspace=\{\astate\in\hmat^{\adim}_{+}:\tr{\astate}=1\}$}};
\node[smallformula] at (4.55,4.57)
    {\resizebox{2.78cm}{!}{$\bspace=\{\bstate\in\hmat^{\bdim}_{+}:\tr{\bstate}=1\}$}};
\foreach \x/\tone/\state in {1.55/GameBlue/\astate,4.55/GamePurple/\bstate}{
    \begin{scope}[shift={(\x,3.49)}]
        \shade[top color=\tone!6,bottom color=\tone!42,shading angle=-25,
            draw=\tone!85,line width=0.6pt] (0,0) circle (0.82);
        \draw[\tone!65,densely dashed,line width=0.45pt]
            (-0.82,0) arc[start angle=180,end angle=0,x radius=0.82,y radius=0.25];
        \draw[\tone!80,line width=0.45pt]
            (-0.82,0) arc[start angle=180,end angle=360,x radius=0.82,y radius=0.25];
        \draw[\tone!40,line width=0.35pt] (0,0) ellipse[x radius=0.30,y radius=0.82];
        \fill[\tone!90!black] (-0.17,0.08) circle (1.5pt);
        \node[anchor=west,inner sep=2pt,text=\tone!70!black] at (-0.11,0.08) {$\state$};
    \end{scope}
}
\node[labeltext] at (1.55,2.40) {$\adim$ qubits};
\node[labeltext] at (4.55,2.40) {$\bdim$ qubits};
\node[smallformula] at (1.55,2.04) {$\astate\in\mathbb C^{2^{\adim}\times2^{\adim}}$};
\node[smallformula] at (4.55,2.04) {$\bstate\in\mathbb C^{2^{\bdim}\times2^{\bdim}}$};
\node[labeltext,text=GameBlue!85!black] at (3.05,1.68) {Spectraplexes (schematic)};
\draw[panel,draw=GameBlue!65,densely dashed,fill=GameBlue!4]
    (0.53,0.27) rectangle (5.57,1.40);
\node[labeltext,text=GameBlue!85!black] at (3.05,1.17) {Joint strategy space and state};
\node[formula] at (3.05,0.83) {$\jointsp=\aspace\times\bspace$};
\node[formula] at (3.05,0.48) {$\joint=(\astate,\bstate)\in\jointsp$};

\node[smallformula] at (9.30,4.90)
    {Referee measures $\joint^{\otimes}=\astate\otimes\bstate$};
\node[smallformula] at (9.30,4.54)
    {POVM $\{P_{\omega}\}_{\omega\in\Omega}$,\quad $\sum\nolimits_{\omega}P_{\omega}=\eye$};
\draw[panel,draw=GameBlue!75,fill=GameBlue!15] (6.68,3.62) rectangle (7.20,4.19);
\node[formula] at (6.94,3.905) {$\astate$};
\node[formula] at (7.44,3.905) {$\otimes$};
\draw[panel,draw=GamePurple!75,fill=GamePurple!15] (7.68,3.62) rectangle (8.20,4.19);
\node[formula] at (7.94,3.905) {$\bstate$};
\draw[flow] (8.35,3.905) -- (8.80,3.905);
\draw[panel,draw=GameGreen!90!black,fill=GameGreen!18]
    (8.96,3.60) rectangle (9.76,4.21);
\draw[GameInk,line width=0.6pt] (9.10,3.76)
    arc[start angle=180,end angle=0,radius=0.26];
\draw[flow] (9.36,3.76) -- (9.51,4.05);
\draw[flow] (9.91,3.905) -- (10.35,3.905);
\node[labeltext] at (11.17,4.14) {Outcome $\omega$};
\node[labeltext,font=\rmfamily\fontsize{7}{8}\selectfont] at (11.17,3.86) {Utility};
\node[smallformula] at (11.17,3.59) {$\mathcal U(\omega)\in[-1,1]$};

\draw[panel,draw=GameGreen!85,fill=GameGreen!10]
    (6.64,2.80) rectangle (11.96,3.32);
\node[smallformula] at (9.30,3.06)
    {Payoff observable: $U=\sum\nolimits_{\omega\in\Omega}\mathcal U(\omega)P_{\omega}$};
\draw[panel,draw=GameGreen!75,densely dashed,fill=white]
    (6.91,1.78) rectangle (11.69,2.57);
\node[labeltext,anchor=west] at (7.08,2.36) {Born rule};
\node[formula] at (9.30,2.03)
    {$p_{\omega}(\joint)=\tr{P_{\omega}(\astate\otimes\bstate)}$};
\draw[panel,draw=GameGreen!85,fill=GameGreen!10]
    (6.57,0.27) rectangle (12.03,1.57);
\node[labeltext] at (9.30,1.34) {Expected payoff (bilinear)};
\node[formula,anchor=base west] at (7.48,0.90)
    {$\payoff(\joint)=\sum\nolimits_{\omega\in\Omega}\mathcal U(\omega)p_{\omega}(\joint)$};
\node[formula,anchor=base west] at (7.48,0.46)
    {$\phantom{\payoff(\joint)}=\tr{U(\astate\otimes\bstate)}.$};

\node[labeltext,anchor=west] at (12.83,4.93) {Interaction map and its adjoint};
\draw[panel,draw=GamePurple!75,fill=GamePurple!8]
    (12.86,3.54) rectangle (18.24,4.70);
\node[smallformula] at (15.55,4.47)
    {$\superop:\hmat^{\bdim}\to\hmat^{\adim},\qquad
      \adj{\superop}:\hmat^{\adim}\to\hmat^{\bdim}$};
\node[formula] at (15.55,4.10)
    {$\superop(\bstate)=\operatorname{Tr}_{\bspace}\!\big[U^{\dagger}(\eye_{\aspace}\otimes\bstate)\big]$};
\node[formula] at (15.55,3.76)
    {$\adj{\superop}(\astate)=\operatorname{Tr}_{\aspace}\!\big[U^{\dagger}(\astate\otimes\eye_{\bspace})\big]$};

\node[labeltext,anchor=west] at (12.83,3.28) {Adjoint identity};
\node[smallformula] at (15.55,2.94)
    {$\payoff(\astate,\bstate)=\inner{\astate,\superop(\bstate)}
        =\inner{\adj{\superop}(\astate),\bstate}.$};
\node[labeltext,anchor=west] at (12.83,2.52) {Individual gradients};
\draw[panel,draw=GamePurple!75,fill=GamePurple!8]
    (12.86,1.25) rectangle (18.24,2.29);
\node[formula] at (15.55,1.98)
    {$\gradient_{\astate}(\bstate)=\nabla_{\astate}\payoff(\astate,\bstate)=\superop(\bstate)$};
\node[formula] at (15.55,1.58)
    {$\gradient_{\bstate}(\astate)=\nabla_{\bstate}(-\payoff(\astate,\bstate))=-\adj{\superop}(\astate)$};
\draw[panel,draw=GamePurple!85,fill=GamePurple!15]
    (12.86,0.27) rectangle (18.24,1.03);
\node[labeltext] at (15.55,0.82) {Joint gradient vector field};
\node[formula] at (15.55,0.48)
    {$\gradient(\joint)=\big(\superop(\bstate),-\adj{\superop}(\astate)\big)$};
\end{tikzpicture}%
\endgroup%
}
\end{center}
\caption{Quantum zero-sum game model: strategy spectraplexes, measurement and payoff, and the associated gradient operator. The state-space drawings are schematic; the formulas apply to arbitrary \(\adim\) and \(\bdim\).}
\label{fig:quantum-game-model}
\end{figure}

\subsection{Equilibrium and Error Measures}
\label{subsec:equilibrium-measures}
We analyze the game through the lens of min--max optimization. Since the payoff $\payoff(\joint)$ is bilinear and the spectraplexes $\aspace, \bspace$ are compact convex sets, von Neumann's Min--Max Theorem guarantees the existence of a value $\val$ such that:
\begin{align} \label{eqn:minmax-thm}
    \val \;\coloneqq\; \min_{\a \in \aspace} \max_{\b \in \bspace} \payoff(\a, \b) \;=\; \max_{\b \in \bspace} \min_{\a \in \aspace} \payoff(\a, \b).
\end{align}
We denote the set of optimal strategies (Nash equilibria) by $\Opt \coloneqq \Opt_\a \times \Opt_\b$, where:
\begin{align}\label{eqn:star-argminmax}
    \Opt_\a \coloneqq \argmin_{\a \in \aspace} \max_{\b \in \bspace} \payoff(\a, \b), \quad
    \Opt_\b \coloneqq \argmax_{\b \in \bspace} \min_{\a \in \aspace} \payoff(\a, \b).
\end{align}
A joint state $\jointstar = (\aopt, \bopt) \in \Opt$ constitutes a \emph{Nash Equilibrium}, satisfying the saddle-point inequality $\payoff(\aopt, \b) \le \payoff(\aopt, \bopt) \le \payoff(\a, \bopt)$ for all $(\a, \b) \in \jointsp$.

We assess the convergence of our algorithms using two distinct metrics, corresponding to the standard notions of weak and strong approximation in optimization.

\noindent\textbf{$\blacktriangleright$ Weak Approximation (Duality Gap).}
To quantify convergence in terms of objective value (exploitability), we employ the \emph{duality gap function} $\gap: \jointsp \to \mathbb{R}_{\ge 0}$. For any joint state $\joint = (\a, \b)$, the gap measures how much each player can gain by deviating to their best response:
\begin{align} \label{eqn:duality-gap}
    \gap(\joint)
    \;\defeq\; \max_{\balt \in \bspace} \payoff(\a, \balt) - \min_{\aalt \in \aspace} \payoff(\aalt, \b)
    \;=\; \max_{(\aalt, \balt) \in \jointsp} \braces*{ \inner{\a}{\superop(\balt)} - \inner{\aalt}{\superop(\b)} }. \tag{Duality Gap}
\end{align}
By construction, $\gap(\joint) \ge 0$, with equality holding if and only if $\joint$ is a Nash equilibrium. A state $\joint$ is termed an \emph{$\varepsilon$-approximate Nash equilibrium} if $\gap(\joint) \le \varepsilon$. This metric controls the ``economic'' error of the game. An error bound makes its relation to distance from the optimal set quantitative.

\noindent\textbf{$\blacktriangleright$ Strong Approximation (Distance to Optimum).}
To quantify the convergence of the iterates themselves, we measure the distance to the set of equilibria $\Opt$. This corresponds to strong approximation and is defined via the Frobenius norm:
\begin{align} \label{eqn:dist-measure}
    \dist(\joint) \coloneqq \min_{\jointstar \in \Opt} \fnorm{\joint - \jointstar}.
\end{align}
These two error measures quantify different aspects of approximation, and the two directions relating them behave differently. Distance to the equilibrium set controls the gap linearly for every game, as a direct consequence of bilinearity. Indeed, the gap is the maximum over \(\jointalt\in\jointsp\) of the linear functions \(\joint\mapsto\langle\joint,\gradient(\jointalt)\rangle\), whose slopes satisfy \(\fnorm{\gradient(\jointalt)}\leq\sqrt2\,\opnorm{\gradient}\) because \(\fnorm{\jointalt}\leq\sqrt2\). Hence, for all \(\joint,\joint_0\in\jointsp\),
\[
|\gap(\joint)-\gap(\joint_0)|\leq\sqrt2\,\opnorm{\gradient}\fnorm{\joint-\joint_0}.
\]
Taking \(\joint_0\) to be a nearest equilibrium, so that \(\gap(\joint_0)=0\), gives \(\gap(\joint)\leq\sqrt2\,\opnorm{\gradient}\dist(\joint)\), which converts metric convergence into a duality-gap guarantee.

Conversely, on our compact strategy space, $\gap(\joint) \to 0$ implies $\dist(\joint) \to 0$, since the continuous gap vanishes exactly on the equilibrium set. By itself, however, this implication is only qualitative. Our work makes it quantitative: a finite Hölder exponent exists for every game, and a linear modulus is available under the full-rank equilibrium hypothesis.

\subsection{Convex-Analytic Conventions}
\label{subsec:convex-analytic-conventions}
We use the standard smoothness and strong-convexity conventions for differentiable functions on normed matrix spaces.
\begin{definition}[$\mu$-strong convexity]
A differentiable function $f:\mathcal X\to\mathbb R$ is $\mu$-strongly convex with respect to a norm $\norm{\cdot}$ if, for all $X,Y\in\mathcal X$,
\[
f(X)\ge f(Y)+\inner{\nabla f(Y),X-Y}+\frac{\mu}{2}\norm{X-Y}^2.
\]
\end{definition}

\begin{definition}[$\beta$-smoothness]
A differentiable function $f:\mathcal X\to\mathbb R$ is $\beta$-smooth with respect to a norm $\norm{\cdot}$ if, for all $X,Y\in\mathcal X$,
\[
f(X)\le f(Y)+\inner{\nabla f(Y),X-Y}+\frac{\beta}{2}\norm{X-Y}^2.
\]
\end{definition}

\subsection{Smoothing and Basic Convexity}
\label{subsec:smoothing-prelim}

For the smoothing algorithm, we use Nesterov's smoothing of the duality gap.
We fix the smoothing center at the maximally mixed joint state
\begin{equation}
\jointtilde
\defeq
\Big(\frac{1}{2^n}\eye_{\aspace},\frac{1}{2^m}\eye_{\bspace}\Big),
\end{equation}
which will also serve as the initialization \(\joint_0\) below.
Given a strongly convex regularizer $\reg(\jointalt)=\frac12\fnorm{\jointalt-\jointtilde}^2$, define
\begin{equation}
\label{eqn:smoothed-gap}
\gap_\smooth(\joint)
=\max_{\jointalt\in\jointsp}\mathcal L_\smooth(\joint,\jointalt),
\qquad
\mathcal L_\smooth(\joint,\jointalt)
=\inner{\a,\superop(\balt)}-\inner{\aalt,\superop(\b)}-\smooth\reg(\jointalt).
\end{equation}
The associated smoothing prox-radius is
\begin{equation}
\label{eqn:smoothing-prox-radius}
\smoothradius
\defeq
\max_{\jointalt\in\jointsp}\reg(\jointalt)
=
\max_{\jointalt\in\jointsp}\frac12\fnorm{\jointalt-\jointtilde}^2.
\end{equation}
Since \(\jointtilde\in\jointsp\) and \(\operatorname{diam}(\jointsp)\leq 2\), we have \(\smoothradius\leq 2\). The smoothing center \(\jointtilde\) remains fixed across outer stages and is distinct from the stage-dependent prox anchor \(\joint_i\) used by the accelerated subroutine.
The maximizer $\maximizer(\joint)$ is unique, $\nabla\gap_\smooth(\joint)=\gradient(\maximizer(\joint))$, and the gradient is Lipschitz with constant $L_\smooth=\opnorm{\gradient}^2/\smooth$ for the quadratic regularizer. The proof appears in Appendix~\ref{appendix:prop:smoothness}.


\section{Error Bounds for Quantum Games}
\label{subsec:spms}

\noindent
\begin{minipage}[t]{0.74\textwidth}
    \vspace{0pt}
    \paragraph{Geometric Foundation: SP-MS.}
    Our convergence analysis rests on the relation between the
    duality gap and the distance to the Nash equilibrium set.
    We establish Saddle-Point Metric Subregularity (SP-MS) with a
    finite exponent for every quantum zero-sum game. When the game
    admits an equilibrium whose two states have full rank, we obtain
    the stronger conclusion with exponent zero. This distinction
    determines when the geometric argument supports linear
    convergence over the non-polyhedral spectraplex.
\end{minipage}%
\hfill
\begin{minipage}[t]{0.22\textwidth}
    \vspace{0pt}
    \centering
    \includegraphics[width=0.8\linewidth]{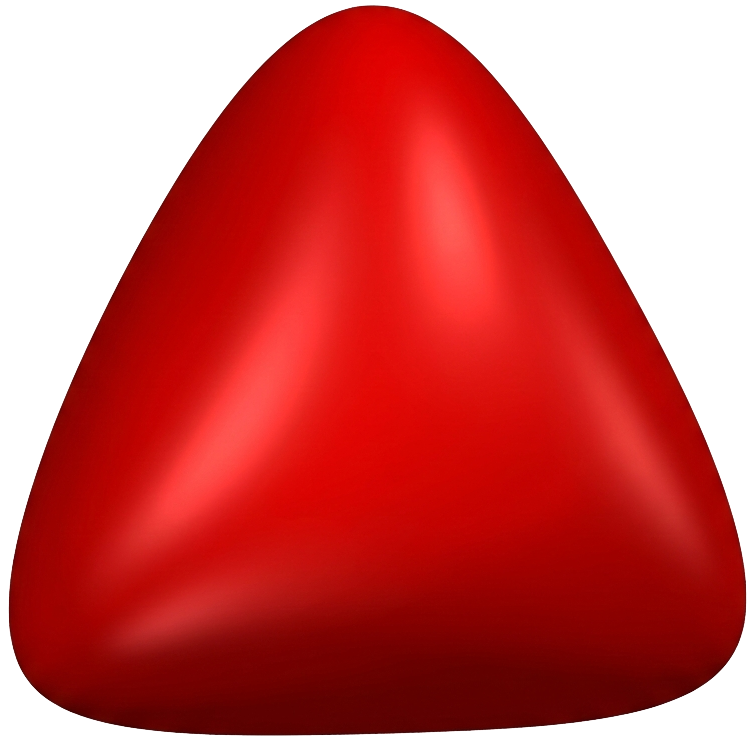}
    \label{fig:spectraplex_mini}
\end{minipage}\vspace{1em}

Let $\sol{\joint}$ denote the Frobenius projection of $\joint$
onto the equilibrium set $\cond{\jointsp}=\Opt$, so that
$\dist(\joint)=\fnorm{\joint-\sol{\joint}}$.
For $\joint\in\jointsp\setminus\cond{\jointsp}$, define the
normalized saddle-point residual
\begin{equation}
    R(\joint)
    :=\sup_{\substack{\jointalt\in\jointsp\\\jointalt\neq\joint}}
      \frac{\tr{\gradpayoff(\joint)(\joint-\jointalt)}}
           {\fnorm{\joint-\jointalt}}.
    \label{eq:main-spms-residual}
\end{equation}
Here the trace pairing and the Frobenius norm are taken on the
product space. SP-MS with parameter $\spmsparam\geq0$ means that
$R(\joint)\geq C\dist(\joint)^{1+\spmsparam}$ for some $C>0$
and every $\joint$ outside the equilibrium set.

\begin{restatable}{theorem}{thmspmsforspectraplex}
\label{thm:sp-ms-for-spectraplex}
For every finite-dimensional quantum zero-sum game with bilinear
payoff, there exist an integer $q\geq1$ and a constant $K>0$ such
that
\begin{equation}
    \dist(\joint)^q\leq K\gap(\joint)
    \qquad\text{for every }\joint\in\jointsp.
    \label{eq:main-spms-holder-bound}
\end{equation}
Consequently, the game satisfies SP-MS with
$\spmsparam=q-1$ and $C=1/(2K)$:
\begin{equation}
    R(\joint)\geq C\dist(\joint)^{1+\spmsparam}
    \qquad
    \text{for every }\joint\in\jointsp\setminus\cond{\jointsp}.
    \label{eq:spms_def}
\end{equation}
If there exists an equilibrium
$(\bar{\astate},\bar{\bstate})\in\cond{\jointsp}$ with
$\bar{\astate}\succ0$ and $\bar{\bstate}\succ0$, then one may
take $q=1$, and hence $\spmsparam=0$.
The exponent and constants may depend on the game, but are
independent of $\joint$ and of the desired accuracy.
\end{restatable}

\begin{remark}
Only one equilibrium with both components positive definite is required for the linear bound. The equilibrium need not be unique, and other equilibria, including the one nearest to a given $\joint$, may be rank deficient. The bounds are global in the sense that they hold uniformly over the strategy space of a fixed game. The constants depend on the game and are not uniform across games.
\end{remark}

\subsubsection*{Proof Strategy and Technical Departures}
The full proof appears in
Appendix~\ref{appendix:thm:sp-ms-for-spectraplex-appendix}.
As in the simplex analysis of \citet[Theorem~5]{wei2021linear},
the main task is to bound distance to equilibrium in terms of the
duality gap. Indeed, bilinearity gives
\begin{equation*}
    \gap(\joint)
    =\max_{\jointalt\in\jointsp}
       \tr{\gradpayoff(\joint)(\joint-\jointalt)}.
\end{equation*}
Outside the equilibrium set, the maximum is positive and is
attained at a point different from $\joint$. Since the diameter
of $\jointsp$ is at most $2$, evaluating
\eqref{eq:main-spms-residual} at such a maximizer yields
$R(\joint)\geq\gap(\joint)/2$.
Thus \eqref{eq:main-spms-holder-bound} implies SP-MS with
$\spmsparam=q-1$.
We establish the gap-to-distance bound through five claims.

\paragraph{Claim~\ref{claim:compact}: Compactness of Equilibria.}
The optimal strategy sets $\aspaceopt=\Opt_\a$ and
$\bspaceopt=\Opt_\b$ are nonempty, compact, and convex.
Their product is the equilibrium set
$\cond{\jointsp}=\aspaceopt\times\bspaceopt$.
Consequently, the nearest equilibrium $\sol{\joint}$ is
well-defined and unique. Moreover, $\gap$ and $\dist$ are
nonnegative and vanish on exactly the same set,
$\cond{\jointsp}$. These facts hold without a rank assumption.

\paragraph{Claim~\ref{claim:spms-affine-optimality}: Affine Equations for Optimality.}
For the linear error bound, fix an equilibrium
$(\bar{\astate},\bar{\bstate})$ whose two states are positive
definite. Complementarity gives
\begin{equation*}
    \superop(\bar{\bstate})=\val \eye_{\aspace},
    \qquad
    \adj{\superop}(\bar{\astate})=\val \eye_{\bspace},
\end{equation*}
where $\eye_{\aspace}$ and $\eye_{\bspace}$ are the identity matrices on the players'
respective spaces. To see why, the positive semidefinite slack
$\superop(\bar{\bstate})-\val \eye_{\aspace}$ has zero trace pairing with
$\bar{\astate}\succ0$, so it must vanish; the other identity
follows in the same way.
Pairing any optimal strategy with the fixed opponent's equilibrium
state then shows that the entire optimal strategy sets satisfy
\begin{align*}
    \aspaceopt
    &=\{\astate\in\aspace:
           \adj{\superop}(\astate)=\val \eye_{\bspace}\},\\
    \bspaceopt
    &=\{\bstate\in\bspace:
           \superop(\bstate)=\val \eye_{\aspace}\}.
\end{align*}
\emph{Technical departure.}
The full-rank equilibrium turns the semidefinite optimality
inequalities into affine equalities. Only this fixed equilibrium
needs to have full rank; the nearest equilibrium to a given
iterate may be rank deficient.

\paragraph{Claim~\ref{claim:spms-residual-gap}: The Duality Gap Controls Equation Residuals.}
Define the nonnegative one-sided gaps
$g_A(\astate)=\lambda_{\max}(\adj{\superop}(\astate))-\val$
and $g_B(\bstate)=\val-\lambda_{\min}(\superop(\bstate))$,
so that $\gap(\joint)=g_A(\astate)+g_B(\bstate)$.
The residual matrices
$H_A=\adj{\superop}(\astate)-\val \eye_{\bspace}$ and
$H_B=\val \eye_{\aspace}-\superop(\bstate)$ satisfy
\begin{equation*}
    \tr{\bar{\bstate}H_A}=0,
    \qquad
    \tr{\bar{\astate}H_B}=0.
\end{equation*}
These identities follow from the adjoint relation and the
identities in Claim~\ref{claim:spms-affine-optimality}.
If a $k\times k$ density matrix $W\succeq\gamma I$, with
$\gamma>0$, satisfies $\tr{WH}=0$ for a Hermitian matrix $H$,
then an elementary spectral estimate gives
$\fnorm{H}\leq(\sqrt{k}/\gamma)\lambda_{\max}(H)$.
Thus positive definiteness lets the largest eigenvalue control
the full Frobenius norm, including the negative eigenvalues.
Applying this estimate with the fixed equilibrium states yields
constants $L_A,L_B>0$ such that
\begin{equation*}
    \fnorm{H_A}\leq L_Ag_A(\astate),
    \qquad
    \fnorm{H_B}\leq L_Bg_B(\bstate).
\end{equation*}

\paragraph{Claim~\ref{claim:spms-affine-repair}: Distance to the Optimal Strategy Sets.}
Consider affine equations $T(\rho)=B$ on a spectraplex, where
$T$ is real-linear and a positive definite density matrix
$\bar\rho$ satisfies the equations.
For a density matrix $\rho$, project it onto
$\{Y=Y^\dagger:\tr{Y}=1,\ T(Y)=B\}$, and call the projection
$Y$. Linear algebra bounds
$\delta=\fnorm{\rho-Y}$ by a constant multiple of
$\fnorm{T(\rho)-B}$.
Although $Y$ may have negative eigenvalues, it satisfies
$Y\succeq-\delta I$.
With $\gamma=\lambda_{\min}(\bar\rho)>0$, the matrix
\begin{equation*}
    \widetilde\rho
    =\frac{\gamma Y+\delta\bar\rho}{\gamma+\delta}
\end{equation*}
is positive semidefinite, has trace one, and still satisfies
$T(\widetilde\rho)=B$.
Moreover,
$\fnorm{\rho-\widetilde\rho}
\leq(1+\sqrt2/\gamma)\delta$.
Applying this construction to the affine descriptions of
$\aspaceopt$ and $\bspaceopt$, and using
Claim~\ref{claim:spms-residual-gap}, gives
\begin{align*}
    \fnorm{\astate-\proj_{\aspaceopt}(\astate)}
    &\leq K_Ag_A(\astate),\\
    \fnorm{\bstate-\proj_{\bspaceopt}(\bstate)}
    &\leq K_Bg_B(\bstate).
\end{align*}
The product structure of the equilibrium set therefore implies
$\dist(\joint)\leq K\gap(\joint)$ with
$K=\max\{K_A,K_B\}$.

\noindent\emph{Technical departure.}
Mixing with a fixed positive definite feasible state repairs the
semidefinite constraint while preserving the affine equations.
This yields constants that are uniform over all strategy pairs
of the fixed game.

\paragraph{Claim~\ref{claim:spms-semialgebraic}: An Unconditional H\"older Error Bound.}
Without a full-rank equilibrium, we use semialgebraic geometry.
A set is \emph{semialgebraic} if it is a finite Boolean combination
of polynomial equalities and inequalities with real coefficients;
a function is semialgebraic if its graph is semialgebraic.
In real coordinates on the Hermitian matrix spaces, positivity
is described by nonnegative principal minors. Thus the joint
spectraplex is compact and semialgebraic.
The extreme-eigenvalue functions are continuous and semialgebraic,
so the formula
\begin{equation*}
    \gap(\joint)
    =\lambda_{\max}(\adj{\superop}(\astate))
     -\lambda_{\min}(\superop(\bstate))
\end{equation*}
has the same properties. Its zero set is the equilibrium set,
which is therefore semialgebraic as well.
The Tarski--Seidenberg theorem also makes distance to this
nonempty compact set semialgebraic
\cite[Chapter~2]{BochnakCosteRoy1998}; distance is continuous
because it is $1$-Lipschitz.
Since $\gap$ and $\dist$ have the same zero set, the compact
semialgebraic \L{}ojasiewicz inequality
\cite[Theorem~1.1 \& Theorem 2.9]{BasuMohammadNezhad2024} gives an integer
$q\geq1$ and a constant $K>0$ for which
$\dist(\joint)^q\leq K\gap(\joint)$ throughout $\jointsp$.
Together with $R(\joint)\geq\gap(\joint)/2$, this proves
SP-MS with $\spmsparam=q-1$.
This argument supplies a finite exponent, but does not establish
that $q=1$ always suffices.

The proof establishes the gap-to-distance bound directly, as well
as SP-MS. Under the full-rank assumption, it gives the global
linear error bound used in the convergence analysis.

\begin{corollary}[Global Error Bound]
\label{prop:main:condition-number}
Suppose the quantum game admits an equilibrium
$(\bar{\astate},\bar{\bstate})$ with
$\bar{\astate}\succ0$ and $\bar{\bstate}\succ0$.
Then there exists a constant $\delta(\superop)>0$ such that
\begin{equation}
    \dist(\joint)\leq\frac{\gap(\joint)}{\delta(\superop)}
    \qquad\text{for every }\joint\in\jointsp.
    \label{eq:main-spms-linear-bound}
\end{equation}
\end{corollary}

For a general game, the corresponding conclusion is the H\"older
bound $\dist(\joint)\leq (K\gap(\joint))^{1/q}$.
The existence of a full-rank equilibrium is sufficient for the
stronger linear bound. Accordingly, convergence guarantees that
use $\spmsparam=0$ require this hypothesis or a separate
verification of SP-MS with exponent zero.

\section{Euclidean Methods and Convergence}
\label{sec:algorithms}
Our goal is to efficiently compute an $\varepsilon$-approximate Nash equilibrium $\joint=(\a,\b)\in\jointsp$, i.e., a joint state satisfying $\gap(\joint) \le \varepsilon$. In this section, we present two Euclidean matrix-adapted first-order methods. The first approach relies on \emph{smoothing} the non-smooth duality gap, while \ogda employs \emph{optimism} to stabilize the saddle-point dynamics directly. Both logarithmic iteration guarantees rely on the linear error bound of Corollary~\ref{prop:main:condition-number} and hence on the full-rank equilibrium hypothesis.

\subsection{Iterative Smoothing: Updates and Convergence}
Directly minimizing the duality gap $\gap(\joint)$ is difficult due to its
non-smooth nature. To overcome this, we adopt a continuation strategy based on
Nesterov's smoothing technique~\citep{Nesterov2005}. The core idea is to replace
$\gap$ by a smooth surrogate $\gap_{\smooth}$, solve the smoothed problem using
Nesterov's accelerated method~\citep{Nesterov1983}, and progressively decrease
the smoothing parameter.

\begin{algorithm}[H]
\caption{\textsc{IterSmooth} (with inlined \textsc{Q-Smoothing} subroutine)}
\label{alg:iterative smoothing}
\begin{algorithmic}[1]
\Require Accuracy $\varepsilon>0$, reduction $\gamma>1$
\State \textbf{Projection:} $\proj_{\jointsp}(X):=\arg\min_{\joint\in\jointsp}\frac12\fnorm{\joint-X}^2$
\State $\joint_0 \gets \jointtilde$
\State $\varepsilon_0 \gets \gap(\joint_0)$
\State \textbf{If} $\varepsilon_0\leq\varepsilon$, \textbf{return} $\joint_0$.

\Statex
\State \textit{Inner \textsc{Q-Smoothing} at outer stage {\(i\)}}
\Function{\textsc{Q-Smoothing}}{$\joint_i,\ \smooth_{i+1},\ \varepsilon_{i+1}$}
    \State \textbf{Lipschitz constant:} $L_{i+1} \gets \opnorm{\gradient}^2/\smooth_{i+1}$
    \State $\joint_i^{(0)} \gets \joint_i$, \quad $\Phi_i^{(0)} \gets \joint_i$
    \For{$k = 0,1,2,\dots$}
        \State $\jointbar_i^{(k)} \gets \frac{2}{k+2}\,\Phi_i^{(k)} + \frac{k}{k+2}\,\joint_i^{(k)}$
        \State $G_i^{(k)} \gets \nabla \gap_{\smooth_{i+1}}(\jointbar_i^{(k)})$
        \State $\joint_i^{(k+1)} \gets \arg\min_{\joint\in\jointsp}$
        \Statex $\quad
        \braces*{
            \inner{G_i^{(k)},\joint-\jointbar_i^{(k)}} +
            \frac{L_{i+1}}{2}\fnorm{\joint-\jointbar_i^{(k)}}^2
        }$
        \Statex $\quad
        =\proj_{\jointsp}\!\Big(\jointbar_i^{(k)} - \frac{G_i^{(k)}}{L_{i+1}}\Big)$
        \If{$\gap(\joint_i^{(k+1)}) \le \varepsilon_{i+1}$}
            \State \Return $\joint_i^{(k+1)}$ \Comment{Success at target {\(\varepsilon_{i+1}\)}}
        \EndIf
        \State $\Phi_i^{(k+1)} \gets \arg\min_{\joint\in\jointsp}$
        \Statex $\quad
        \braces*{
            \sum_{t=0}^{k}\frac{t+1}{2}
            \inner{G_i^{(t)},\joint-\jointbar_i^{(t)}}
            + \frac{L_{i+1}}{2}\fnorm{\joint-\joint_i}^2
        }$
        \Statex $\quad$
    \EndFor
\EndFunction

\Statex
\State \textit{Outer loop}
\For{$i = 0,1,2,\dots$}
    \State {\(\varepsilon_{i+1} \gets \varepsilon_i/\gamma\)} \Comment{Set the next target accuracy}
    \State \textbf{Smoothing parameter:} $\smooth_{i+1} \gets \varepsilon_{i+1}/(2\smoothradius)$
    \State $\joint_{i+1} \gets \textsc{Q-Smoothing}(\joint_i,\ \smooth_{i+1},\ \varepsilon_{i+1})$
    \If{$\gap(\joint_{i+1}) \le \varepsilon$}
        \State \Return $\joint_{i+1}$
    \EndIf
\EndFor
\end{algorithmic}
\end{algorithm}

At smoothing level $\smooth$, let $L_{\smooth}$ denote the Lipschitz constant of
$\nabla \gap_{\smooth}$, and let
\[
    \proj_{\jointsp}(X)
    \defeq
    \argmin_{\joint\in\jointsp}
    \frac12\fnorm{\joint-X}^2
\]
be the Frobenius projection onto the spectraplex. The inner accelerated loop
uses the update
\begin{align} \label{eqn:smoothing-update}
    \jointbar^{(t)}
    &=
    \frac{2}{t+2}\Phi^{(t)}
    +
    \frac{t}{t+2}\joint^{(t)},
    \qquad
    \joint^{(t+1)}
    =
    \proj_{\jointsp}\!\left(
        \jointbar^{(t)}
        -
        \frac{1}{L_{\smooth}}
        \nabla \gap_{\smooth}(\jointbar^{(t)})
    \right).
\end{align}
The auxiliary aggregate sequence \(\Phi^{(t)}\) is updated by the weighted linearized prox subproblem in Algorithm~\ref{alg:iterative smoothing}; together with \eqref{eqn:smoothing-update}, this gives the complete three-sequence accelerated recurrence.

To eliminate the smoothing bias, we use a geometric accuracy schedule. At the
beginning of stage {\(i\)}, the current state satisfies
{\(\gap(\joint_i)\leq\varepsilon_i\)}. The method sets
{\(\varepsilon_{i+1}=\varepsilon_i/\gamma\)} and
{\(\smooth_{i+1}=\varepsilon_{i+1}/(2\smoothradius)\)}, and then calls
\qsmoothing from {\(\joint_i\)} to produce {\(\joint_{i+1}\)} with
{\(\gap(\joint_{i+1})\leq\varepsilon_{i+1}\)}. Thus the same invariant holds at
the beginning of the next stage, and the method terminates once the original
duality gap is at most {\(\varepsilon\)}. The full procedure is given in
Algorithm~\ref{alg:iterative smoothing}. For polyhedral feasible sets, this
type of decaying-regularization scheme was analyzed by
\citet{Gilpin2012FirstOrderAlgorithm}; in
Section~\ref{subsec:q-iterated-convergence}, we establish the corresponding
guarantee for the spectraplex domain.

\subsubsection*{{Iteration bound}} \label{subsec:q-iterated-convergence}

We first analyze the \itersmooth\ algorithm (Algorithm \ref{alg:iterative smoothing}). This method operates by solving a sequence of smoothed games. We begin by bounding the complexity of the inner loop, \qsmoothing.

\begin{restatable}[Inner Loop Complexity]{proposition}{qsmoothingprop} \label{prop:qsmoothing}
    At any outer stage \(i\), the \qsmoothing\ subroutine initialized at {\(\joint_i\)}, with target accuracy {\(\varepsilon_{i+1}\)} and smoothing parameter {\(\smooth_{i+1} = \varepsilon_{i+1}/(2\smoothradius)\)}, returns a state \(\joint_{i+1}\) with $\gap(\joint_{i+1}) \le \varepsilon_{i+1}$ after at most
    \begin{align*}
        \ceil*{
        \frac{2\sqrt{2}\cdot \opnorm{\gradient}\cdot \sqrt{\smoothradius}\cdot \dist(\joint_i, \cond{\jointsp})}{\varepsilon_{i+1}}
        }
    \end{align*}
    iterations.
\end{restatable}

The proof (Appendix~\ref{appendix:prop:qsmoothing}) relies on adapting Nesterov's acceleration analysis to the smoothed duality gap $\gap_\smooth$, which uniformly approximates the true gap via $\gap_\smooth(\joint) \le \gap(\joint) \le \gap_\smooth(\joint) + \smooth \smoothradius$. By chaining these inner loops with a geometrically decaying smoothing schedule $\smooth_{i+1} = \smooth_i / \gamma$, we obtain the total complexity. The error bound (Corollary~\ref{prop:main:condition-number}) is crucial here: it ensures that as the gap decreases, the iterate remains close to the optimal set, preventing the required number of inner iterations from exploding. 
Indeed, $\gap(\joint_i)\le\varepsilon_i$ implies $\dist(\joint_i)\le\varepsilon_i/\delta(\Xi)$, whereas the next target is $\varepsilon_{i+1}=\varepsilon_i/\gamma$. Hence $\dist(\joint_i)/\varepsilon_{i+1}\le\gamma/\delta(\Xi)$, so the bound of Proposition~\ref{prop:qsmoothing} on the stage cost is independent of $i$.


\begin{restatable}[Convergence of \itersmooth]{theorem}{qiteratedconvergenceratefin}
\label{thm:q-iterated-convergence-rate-fin}
    Suppose the quantum game admits an equilibrium
$(\bar{\astate},\bar{\bstate})$ with
$\bar{\astate}\succ0$ and $\bar{\bstate}\succ0$. The \itersmooth\ algorithm computes an $\varepsilon$-equilibrium in at most
    \[
        T=\bigoh\left( \kappa(\superop) \, \ln\left( \frac{\opnorm{\gradient}}{\varepsilon} \right) \right)
    \]
    first-order iterations, where $\kappa(\superop) = \opnorm{\gradient}/\delta(\superop)$ is the condition number of the game.
\end{restatable}
%
\subsection{{Optimistic Gradient Descent--Ascent}}
An alternative to smoothing is to directly stabilize the saddle-point dynamics.
Standard Gradient Descent--Ascent may cycle or diverge even in
bilinear zero-sum games. Optimistic methods address this instability by using a
predictive correction based on past gradient information, yielding
extragradient-like stability while requiring only one fresh evaluation of the
saddle-point operator per iteration
\citep{rakhlin2013online,wei2021linear,applegate2023faster,fercoq2023quadratic}.

We present two variants of this philosophy. \ogda\ uses Euclidean geometry and
updates by Frobenius projection, whereas \ommwu\ uses the matrix-entropy
geometry and updates by matrix exponentiation. The two methods both maintain an auxiliary point
$\hat{\joint}$ and use the current saddle-point operator value to form the next
iterate. However, they differ in the underlying geometry and algorithmic
formalism: \ogda\ is a projected optimistic mirror-descent
method, while \ommwu\ is an optimistic FTRL method with the negative von Neumann
entropy regularizer.

\vspace{1em}
\noindent
\begin{minipage}[t]{0.48\textwidth}
    \centering
    \textbf{1. Optimistic GDA (Euclidean)} \\
    \vspace{0.5em}

    \ogda performs a projected update using the Euclidean projection $\proj_{\jointsp}$. The dynamics evolve as:
    \begin{align*}
        \joint_{t+1} &= \proj_{\jointsp}\!\parens*{\hat{\joint}_t - \eta \gradient(\joint_t)}, \\
        \hat{\joint}_{t+1} &= \proj_{\jointsp}\!\parens*{\hat{\joint}_t - \eta \gradient(\joint_{t+1})}.
    \end{align*}
    It requires one gradient oracle call per step (reusing $\gradient(\joint_t)$ from the previous iteration).
    \vspace{0.5em}
    \hrule
    \vspace{0.2em}
    {\footnotesize See Algorithm~\ref{alg:ogda} for details.}
\end{minipage}
\hfill
\vline
\hfill
\begin{minipage}[t]{0.48\textwidth}
    \centering
    \textbf{2. Optimistic  \texttt{MMWU} (Entropic)} \\
    \vspace{0.5em}

    \ommwu replaces projection with the Matrix Softmax $\Lambda(\cdot)$, respecting the information geometry. The updates are:
    \begin{align*}
        \joint_{t+1} &= \Lambda\!\parens*{\log \hat{\joint}_t - \eta \gradient(\joint_t)}, \\
        \hat{\joint}_{t+1} &= \Lambda\!\parens*{\log \hat{\joint}_t - \eta \gradient(\joint_{t+1})}.
    \end{align*}
    This multiplicative update naturally preserves positive semi-definiteness and unit trace.
    \vspace{0.5em}
    \hrule
    \vspace{0.2em}
    {\footnotesize See Algorithm~\ref{alg:ommwu} for details.}
\end{minipage}
\vspace{1.5em}

Both routines are initialized at the maximally mixed joint state and terminate
once $\gap(\joint_t)\le \eps$.
{The step size $\eta$ is fixed according to the
conditions in {Section}~\ref{subsec:ogda-convergence}
{for \ogda; the unregularized \ommwu dynamics are discussed separately in Section~\ref{subsec:ommwu-slowdown}}. In particular, under the SP--MS condition{ with exponent zero and the admissible step size},
the matrix-adapted \ogda\ enjoys linear last-iterate convergence and computes an
$\eps$-approximate equilibrium in ${\bigoh_d}(\log(1/\eps))$ iterations.
The \ommwu\ update is the matrix-entropy analogue of optimistic
Follow-the-Regularized-Leader~\citep{rakhlin2013online,syrgkanis2015fastconvergenceregularizedlearning}
and follows the matrix multiplicative-weights line of work
\citep{tsuda2005matrix,NY83,BECK2003167,Vasconcelos2025quadraticspeedupin}.

{In summary, Iterative Smoothing reduces the problem to a sequence of smooth optimizations, while the Optimistic family stabilizes the oscillation of the original game directly. We analyze {the Euclidean convergence guarantees} in this section and further in Appendix \ref{appendix:OmittedConvergence}{, and distinguish the regularized QRE method from unregularized \ommwu in Section~\ref{sec:entropic-methods}}.}
\begin{algorithm}[h]
\caption{\textsc{Optimistic Gradient Descent--Ascent (\ogda)}}
\label{alg:ogda}
\begin{algorithmic}[1]
\Require Accuracy $\varepsilon > 0$, step size $\eta > 0$
\State \textbf{Projection:} $\proj_{\jointsp}(X) := \argmin_{\joint \in \jointsp} \frac{1}{2}\fnorm{\joint - X}^2$
\State \textbf{Gradient:} $\gradient(\joint) = \big(\superop(\b), -\adj{\superop}(\a)\big)$
\State $\joint_0 \gets \Big(\frac{1}{2^n}\eye_{\aspace},\, \frac{1}{2^m}\eye_{\bspace}\Big)$
\State $\hat{\joint}_0 \gets \joint_0$
\For{$t = 0,1,2,\ldots$}
    \State $\joint_{t+1} \gets \proj_{\jointsp}\!\Big(\hat{\joint}_t - \eta\,\gradient(\joint_t)\Big)$
    \If{$\gap(\joint_{t+1}) \le \varepsilon$}
        \State \Return $\joint_{t+1}$
    \EndIf
    \State $\hat{\joint}_{t+1} \gets \proj_{\jointsp}\!\Big(\hat{\joint}_t - \eta\,\gradient(\joint_{t+1})\Big)$
\EndFor
\end{algorithmic}
\end{algorithm}

\paragraph{{Linear convergence of Optimistic Gradient Descent-Ascent.}} \label{subsec:ogda-convergence}
The power of the SP--MS framework lies in its modularity: it separates the geometric difficulty from the algorithmic dynamics. Having verified that quantum zero-sum games {with a full-rank equilibrium} satisfy the requisite error-bound condition \eqref{eq:spms_def} {with exponent zero} despite their non-polyhedral nature, we can directly invoke standard optimistic convergence guarantees without reinventing the algorithmic Lyapunov analysis~\cite{hsieh2019convergence,wei2021linear,malickrate}.
\begin{theorem}[\ogda Linear Convergence] \label{thm:ogda-convergence-rate}
    Suppose the quantum game admits an equilibrium
$(\bar{\astate},\bar{\bstate})$ with
$\bar{\astate}\succ0$ and $\bar{\bstate}\succ0$. {Let \(L>0\) be a Lipschitz constant for \(F\).} Consider the \ogda algorithm with step size ${\eta = \Theta(1/L)}$. Since the quantum payoff geometry satisfies the SP-MS condition with parameter $\spmsparam=0$ (as proven in Section~\ref{subsec:spms}), the iterates $\joint_t$ converge linearly to the set of equilibria $\cond{\jointsp}$
    \begin{equation*}
        \dist^2(\joint_t, \cond{\jointsp}) \;\leq\; C_{\text{init}} \cdot \dist^2(\joint_0, \cond{\jointsp}) \cdot (1 + \mu)^{-t},
    \end{equation*}
    where $\mu > 0$ depends on the error-bound constant $C$ and the step size. Consequently, \ogda computes an $\varepsilon$-approximate equilibrium in $\bigoh(\log(1/\varepsilon))$ iterations.
\end{theorem}
A proof is provided in Appendix~\ref{appendix_subsec:ogda-convergence}.

{The step-size restriction is the one in \citet[Theorem~8]{wei2021linear}. To pass from its distance estimate to our accuracy criterion, using the gap bound from Section~\ref{subsec:equilibrium-measures}}, we get
\[
{\gap(\joint_t)\leq\sqrt{2C_{\text{init}}}\,\opnorm{\gradient}\dist(\joint_0)(1+\mu)^{-t/2}.}
\]
{For fixed game-dependent constants, this upper bound is at most \(\varepsilon\) after \(O(\log(1/\varepsilon))\) iterations. The algorithm may stop earlier if its gap test is already satisfied.}

\section{{Entropic Methods: Regularization and Lower Bound}}
\label{sec:entropic-methods}

{Algorithm~\ref{alg:ommwu} specifies the unregularized entropic dynamics. We first analyze the convergence regularized \ommwu to the quantal response equilibrium, then give a lower bound for convergence of the unregularized algorithm.}

\begin{algorithm}[h]
\caption{\textsc{Optimistic Matrix Multiplicative Weights Update (\ommwu)}}
\label{alg:ommwu}
\begin{algorithmic}[1]
\Require Accuracy $\varepsilon > 0$, step size $\eta > 0$
\State \textbf{Matrix softmax:} $\Lambda(\joint) := \big(\exp(\a)/\tr{\exp(\a)},\, \exp(\b)/\tr{\exp(\b)}\big)$
\State \textbf{Matrix log:} $\log(\joint) := \big(\log(\a),\, \log(\b)\big)$
\State \textbf{Gradient:} $\gradient(\joint) = \big(\superop(\b), -\adj{\superop}(\a)\big)$
\State $\joint_0 \gets \Big(\frac{1}{2^n}\eye_{\aspace},\, \frac{1}{2^m}\eye_{\bspace}\Big)$
\State $\hat{\joint}_0 \gets \joint_0$
\For{$t = 0,1,2,\ldots$}
    \State $\joint_{t+1} \gets \Lambda\!\Big(\log \hat{\joint}_t - \eta\,\gradient(\joint_t)\Big)$
    \If{$\gap(\joint_{t+1}) \le \varepsilon$}
        \State \Return $\joint_{t+1}$
    \EndIf
    \State $\hat{\joint}_{t+1} \gets \Lambda\!\Big(\log \hat{\joint}_t - \eta\,\gradient(\joint_{t+1})\Big)$
\EndFor
\end{algorithmic}
\end{algorithm}

\subsection{{Regularized \ommwu and Quantal Response Equilibrium}}
\label{subsec:regularized-entropic}
We next revisit Optimistic Matrix Multiplicative Weights Update (\ommwu), whose geometry is naturally adapted to the spectraplex. In contrast to Euclidean projection methods, \ommwu evolves by matrix exponentiation and therefore preserves positive semidefiniteness and unit trace by construction. Its natural Lyapunov function is the quantum relative entropy, or equivalently, the Bregman divergence induced by the negative von Neumann entropy. This makes \ommwu the matrix analog of optimistic entropic mirror descent.

Our analysis proceeds through an entropy-regularized game using the framework of \cite{Sokota2023QRE}. For a regularization parameter \(\delta>0\), let \(\joint^{(\delta)}\) denote the corresponding quantum quantal-response equilibrium (QRE), i.e., the saddle point of the game obtained by adding the negative von Neumann entropy regularizer to the players' objectives. The regularized {\ommwu} (i.e., {mirror-prox with an entropic regularizer}) defined below satisfies the following convergence guarantee for QRE.

\paragraph{{Updates, geometry, and output.}}
{Let \(\psi(\joint)=\tr{\a\log\a}+\tr{\b\log\b}\), with the continuous convention \(0\log0=0\), and let \(D_\psi\) be its Bregman divergence, equal to the sum of the players' quantum relative entropies. Work in the product Frobenius norm, take \(L>0\) with \(L\geq\opnorm{\gradient}\), and fix \(0<\eta\leq1/L\). Initialize \(\hat\joint_1=\jointtilde\). For \(t=1,2,\ldots\), define}
\begin{align*}
{\joint_t}
&={\argmin_{\joint\in\jointsp}\Bigl\{
\eta\inner{\gradient(\hat\joint_t),\joint}
+\eta\delta\psi(\joint)+D_\psi(\joint,\hat\joint_t)\Bigr\},}\\
{\hat\joint_{t+1}}
&{=}{\argmin_{\joint\in\jointsp}\Bigl\{
\eta\inner{\gradient(\joint_t),\joint}
+\eta\delta\psi(\joint)+D_\psi(\joint,\hat\joint_t)\Bigr\}.}
\end{align*}
{After \(t\) such iterations, return \(\hat\joint_{t+1}\). Each iteration uses two evaluations of \(F\); componentwise, a prox update with feedback \(g\) is \(\Lambda((\log\hat\joint_t-\eta g)/(1+\eta\delta))\), so all iterates remain positive definite. These updates use the feedback and the explicit payoff regularization analyzed in Appendix~\ref{subsec:ommwu-convergence}. Algorithm~\ref{alg:ommwu} instead uses entropy only as its mirror geometry and is the unregularized method studied in the next subsection.}

{Quantum Pinsker's inequality gives \(D_\psi(X,Y)\geq\frac12\fnorm{X-Y}^2\) on the product space, and \(\fnorm{F(X)-F(Y)}\leq L\fnorm{X-Y}\). These compatible estimates justify the step-size condition; \(L\) may depend on the dimensions.}

\begin{restatable}[Convergence for {\ommwu}]
{theorem}{ommwuconvergence}
\label{thm:ommwu-convergence}
Consider a quantum zero-sum game with unique Nash equilibrium $\jointopt$, where $\jointopt=(\aopt,\bopt)$ with $\aopt\succ0$ and $\bopt\succ0$. For smoothness parameter $L$ and step size ${0} <\eta \le \frac{1}{L}$, the {\ommwu (with an additional entropic regularization term)} guarantees the following last-iterate convergence:
\begin{align*}
    S(\joint^{(\delta)}\|\hat\joint_{t+1}) \le (1+\delta \eta)^{-t} S(\joint^{(\delta)}\|\hat\joint_{1}),\quad \text{ where $\|\joint^{(\delta)}-\jointopt\|_F \le \mathcal{O}_d(\delta)$.}
\end{align*}
The constant in the distance bound may depend on the game's error-bound modulus $\delta(\superop)$, but is independent of the regularization parameters.
\end{restatable}

{The convergence to the QRE uses monotonicity of \(F\), entropy strong convexity, and \(\eta L\leq1\). The full-rank Nash hypothesis enters only when converting regularization bias into distance to Nash. In fact, with \(d_A=2^n\) and \(d_B=2^m\), Lemma~\ref{lemm:qre-eps-nash-gap-vn} gives for every game and a particular constant $\delta'$}
\[
{\gap(\joint^{(\delta)})\leq\delta'(\log d_A+\log d_B).}
\]
{Under the full-rank hypothesis, Corollary~\ref{prop:main:condition-number} converts this into}
\[
{\dist(\joint^{(\delta)},\Opt)
\leq\frac{\delta'(\log d_A+\log d_B)}{\delta(\superop)}.}
\]
{Uniqueness identifies this distance with \(\fnorm{\joint^{(\delta)}-\jointopt}\). It is not needed for the gap bias bound itself.}

The theorem gives a linear last-iterate rate toward the regularized equilibrium for every fixed \(\delta>0\). Since the QRE is an \(\bigoh(\delta)\)-approximate Nash equilibrium in duality gap, choosing \(\delta\) proportional to the target accuracy $\varepsilon$ with a game-dependent proportionality constant yields an \(\widetilde \bigoh(1/\varepsilon)\) last-iterate guarantee for approximating the original Nash equilibrium (in Frobenius distance) with output $\hat\joint_{t+1}$, by Corollary~\ref{prop:main:condition-number} and the relative-entropy contraction above. Thus, the QRE analysis {provides a last-iterate guarantee for the regularized scheme} while also clarifying the role of the regularization parameter: linear convergence holds to the regularized target, but the passage from QRE to Nash introduces the familiar accuracy-dependent tradeoff. The relevant details and proof are {deferred} to Appendix \ref{subsec:ommwu-convergence}.

\subsection{{Lower Bound for Convergence of} \ommwu}
\label{subsec:ommwu-slowdown}

The preceding convergence analysis shows that {\ommwu} enjoys {linear} last-iterate convergence toward the (entropy-regularized) equilibrium. It is therefore natural to ask whether the same entropic geometry can yield an instance-independent linear rate directly for the unregularized quantum game. We show that {no last-iterate convergence rate tending to zero can hold uniformly over all games for the unregularized updates in Algorithm~\ref{alg:ommwu}}.

The obstruction already appears inside the diagonal subspace of the spectraplex. Given a
classical payoff matrix \(A\in\mathbb R^{2\times 2}\), consider the diagonal quantum payoff
observable
\[
    U=\operatorname{diag}(A_{11},A_{12},A_{21},A_{22}).
\]
If \ommwu is initialized at diagonal density matrices, then every iterate remains diagonal. Moreover, the induced diagonal trajectory coincides exactly with classical \ommwu on \(A\), and the quantum duality gap agrees with the classical duality gap at every time. Thus, a classical instance that admits slow last-iterate convergence embeds isometrically into quantum zero-sum games.

We use this reduction for the following hard instance proposed by \citet{cai2025fastlastiterateconvergencelearning}
\[
    A_\delta =
    \begin{pmatrix}
        \frac12+\delta & \frac12 \\
        0 & 1
    \end{pmatrix},
    \qquad
    U_\delta
    =
    \operatorname{diag}\!\left(\tfrac12+\delta,\tfrac12,0,1\right),
\]
where \(\delta>0\) indexes the family\footnote{It is distinct from the error-bound modulus \(\delta(\superop)\) and from the QRE regularization parameter}. The following theorem is the classical lower-bound statement, with the first played iterate and auxiliary initialization used in the cited source. Lemma~\ref{lem:ommwu-startup-lower-bound} separately adapts it to the initialization in Algorithm~\ref{alg:ommwu}.

\begin{restatable}[{\citet[Theorem~1]{cai2025fastlastiterateconvergencelearning}}]{theorem}{caiLastIterateTheorem}
\label{thm:name}
Let $\reg \from [0,1]\to\R$ be the negative entropy regularizer
\[
\reg(x)=x\log x+(1-x)\log(1-x).
\]
For any $\eta>0$, define the one-dimensional mirror map $F_{\eta,\reg}:\R\to[0,1]$ by
\[
F_{\eta,\reg}(e)
\;:=\;
\argmin_{x\in[0,1]}
\Bigl\{x\,e+\frac{1}{\eta}\reg(x)\Bigr\},
\qquad\text{so that }F_{\eta,\reg}(e)=F_{1,\reg}(\eta e).
\]
Then $F_{1,\reg}$ is $L$-Lipschitz with $L=\frac12$.
Let the constants
\[
c_1 = \frac12 - F_{1,\reg}\Bigl(\frac{1}{20L}\Bigr),
\qquad
c_2 = F_{1,\reg} \Bigl(-\frac{c_1^2}{480L}\Bigr)-\frac12 > 0,
\qquad
\delta' = \frac{c_1^2}{480L},
\]
Then there exists a constant $\hat{\delta} \in (0, \delta']$, depending only on $c_1$ and $\delta'$ such that for every $\delta \in (0, \hat{\delta})$, the \ommwu dynamics on $A_\delta$ with any step size $\eta \leq \frac{1}{4L}$ satisfies the following: there exists an iteration $t \geq \frac{c_1}{3 \eta L \delta}$ with $\gap_C(x_t, y_t) \geq c_2$.
\end{restatable}

Combining our reduction from classical to quantum with the initialization bound in Lemma~\ref{lem:ommwu-startup-lower-bound}, the theorem implies that \ommwu cannot admit an instance-independent convergence rate in the last iterate over all quantum games. Indeed, for the family \(U_\delta\), the dynamics have a constant-gap iterate at a time at least a constant multiple of \(1/(\eta\delta)\). Hence, even if \ommwu converges linearly on individual instances, the constants in the convergence rate cannot be chosen uniformly over this family. The detailed reduction and constructions are deferred to Appendix~\ref{appendix:slow-last-iter-ommwu}.

{
\subsubsection*{Remarks on Algorithmic Choice}

Before navigating the technical proofs, we pause to contextualize our analytical perspective and the rationale behind our algorithmic selection.

\paragraph{Novelty of Analysis.}
From a purely algorithmic standpoint, the matrix-adapted methods we study—variants of extragradient and smoothing—are standard fixtures in optimization. The primary contribution of this work is not their design, but their rigorous analysis within the class of \emph{semidefinite games} \cite{ickstadt2024semidefinite}. Specifically, we establish the critical link between \emph{error bounds} and \emph{linear last-iterate convergence} {under the full-rank equilibrium hypothesis} in the quantum setting.

Prior to this work, transferring such guarantees from polyhedral to spectraplex domains was an open challenge, as standard proof techniques relying on coordinate-wise commutativity break down in the non-commutative geometry of quantum states

\begin{figure}[h]
    \centering
    \includegraphics[width=\linewidth]{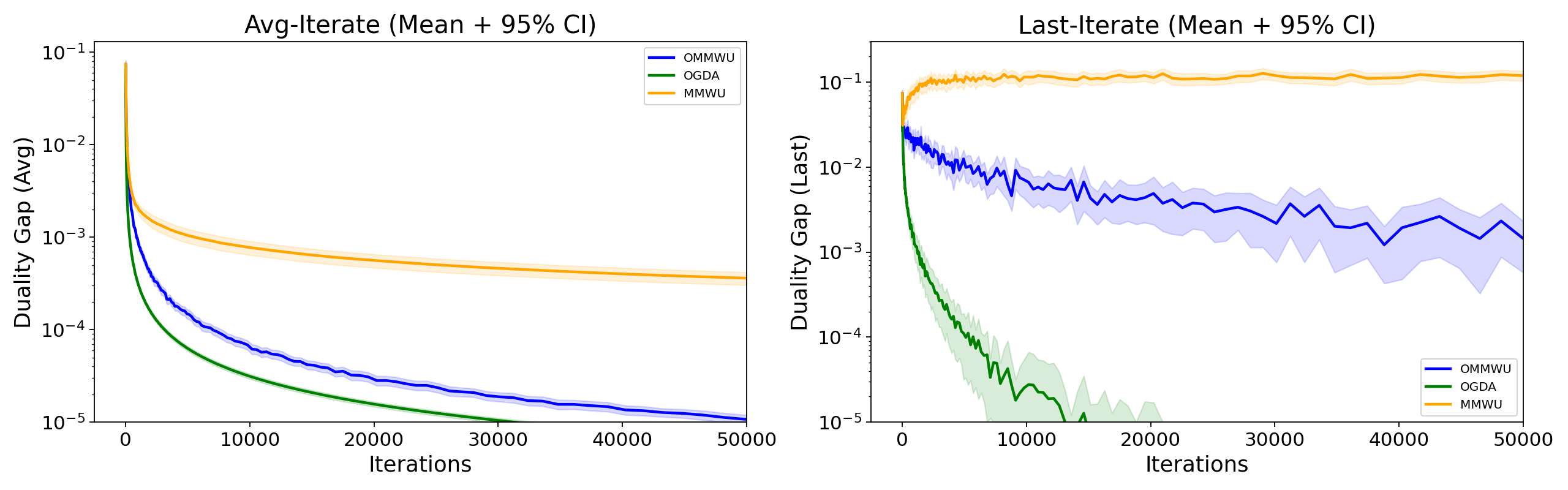}
     \vspace{-2em}
    \label{fig:performance}
\end{figure}

\paragraph{Practical Efficiency vs. Theoretical Geometry.}
Our study balances two complementary motivations:
\begin{itemize}
    \setlength{\emergencystretch}{2em}
    \item \textbf{The Euclidean Case (\ogda):} {In practice, direct projection methods like \ogda often yield the fastest convergence}.
    The figure above provides an initial illustration of this phenomenon in 2-qubit random quantum zero-sum games: over 50 independent games and \(T = 50{,}000\) iterations with \(\eta = 1\), \ogda exhibits the fastest decrease in both average-iterate and last-iterate duality gaps. This empirical superiority stems from aggressive step sizes—a phenomenon reminiscent of the "Edge of Stability" in minimization, where theoretical bounds are typically far more pessimistic than practical performance. Experimentally, we validate this advantage across both full-rank and low-rank settings, moving beyond the standard equilibrium assumptions or the strict step-size restrictions required by the theorem. Detailed results are provided in Section~\ref{sec:experiments_convergence}, Figure~\ref{fig:2qubit-qzsg-convergence-ci95}.

\item \textbf{The Entropic Case (\ommwu):} \ogda's speed notwithstanding, entropic methods deserve a separate analysis, for three reasons. First, unregularized \ommwu is the state of the art in the \emph{average} iterate: $\bigoh(\log d/\varepsilon)$ iterations, with only a logarithmic dependence on the dimension~\citep{Vasconcelos2025quadraticspeedupin}, a rate that is tight for the uniform-average output rule~\citep{su2026lowerbounds}. Second, the entropic geometry is the natural one for density matrices: the update is a matrix exponential, so positivity and unit trace are preserved without projection, and the iterates are Gibbs states, the objects on which the parallel algorithm of \citet{jain2009parallel} and quantum SDP solvers~\citep{brando2019} are built. Third, the regularized entropic mirror-prox scheme of Section~\ref{subsec:regularized-entropic} admits a relative-entropy contraction toward the quantal response equilibrium (Theorem~\ref{thm:ommwu-convergence}). Two caveats delimit what we claim. The step-size condition of that theorem is stated through the Frobenius Lipschitz constant, which may depend on the dimension, so no dimension-free step size is asserted; and the linear error bound enters only to convert the QRE-to-Nash bias into distance. The major technical hurdle here—which we overcome—is leveraging error bounds for Bregman-style updates in the absence of commutativity, a difficulty conceptually aligned with the challenges in \mbox{\cite{wei2021linear}}.

\end{itemize}}

\section{Experimental Evaluation}
\label{appendix:convergence experiment}
We complement the theory with numerical experiments on three unregularized methods: \mmwu, \ommwu, and \ogda. In every game, Alice and Bob control the same number of qubits, so a \(2k\)-qubit game gives each player a dimension \(d=2^k\). We measure accuracy by the duality gap and report it at two outputs: the \emph{last iterate} \((\alpha_t,\beta_t)\), and the \emph{average iterate}, the uniform average of all iterates up to time \(t\). The section has three parts. Section~\ref{sec:experiments_full_rank} tests the setting of Theorem~\ref{thm:ogda-convergence-rate}, using games with a prescribed full-rank equilibrium. Section~\ref{sec:experiments_convergence} removes this structure and uses random POVM games. Section~\ref{appendix:exp:ommwu-Udelta} studies the hard instances \(U_\delta\) behind the lower bound for \ommwu.

\subsection{Games with a Known Full-Rank Equilibrium}
\label{sec:experiments_full_rank}

We first test the three methods on randomly generated quantum zero-sum games with known full-rank equilibria. These games satisfy the full-rank hypothesis of Theorem~\ref{thm:ogda-convergence-rate}, so they are the natural place to check the predicted linear last-iterate convergence of \ogda. We consider 2-, 4-, and 6-qubit games, where Alice and Bob each control 1, 2, and 3 qubits, respectively.

\paragraph{Game construction.}
We build each game in two steps: we first fix the equilibrium, and then design a payoff observable around it. For each local dimension \(d\in\{2,4,8\}\), we sample independent traceless Hermitian matrices \(H_A\) and \(H_B\) from complex Gaussian matrices and normalize them to have operator norm one. We then prescribe the equilibrium states
\[
\alpha^*=\frac{\eye_d}{d}+\frac{H_A}{2d},
\qquad
\beta^*=\frac{\eye_d}{d}+\frac{H_B}{2d}.
\]
Both states have trace one and minimum eigenvalue at least \(1/(2d)\), so they are positive definite, that is, full rank, and differ from the maximally mixed state.

Next, we build the payoff observable. Let \(\{T_i\}_{i=1}^{d^2-1}\) be a Frobenius-orthonormal basis of traceless Hermitian matrices, and define
\[
E_i=T_i-\operatorname{Tr}[\alpha^* T_i]\,\eye_d,
\qquad
F_j=T_j-\tr{\beta^* T_j}\eye_d.
\]
We sample an invertible real matrix
\(C=Q_1\operatorname{diag}(s_1,\ldots,s_{d^2-1})Q_2^\top\),
where \(Q_1,Q_2\) are independent random orthogonal matrices and the singular values \(s_i\) are independently sampled uniformly from \([0.5,1.5]\). The payoff observable is
\[
U_0=\sum_{i,j=1}^{d^2-1}C_{ij}E_i\otimes F_j,
\qquad
U=\frac{0.8}{\|U_0\|_{\mathrm{op}}}U_0.
\]
Since \(\operatorname{Tr}[\alpha^* E_i]=
\operatorname{Tr}[\beta^* F_j]=0\), each player's payoff operator vanishes against the prescribed equilibrium strategy of the opponent. Hence \((\alpha^*,\beta^*)\) is an equilibrium with value zero. Its strict positivity, together with the invertibility of \(C\), makes this equilibrium unique. The payoff depends on both diagonal and off-diagonal components of the players' quantum states. For every generated game, we check that the prescribed equilibrium strategies are positive definite and satisfy the equilibrium conditions, and confirm them numerically using an independent SDP solver.

\paragraph{Setup.}
For each game size, we generate \(50\) independent instances and run each unregularized method for \(T=10{,}000\) iterations. All methods start from \(\alpha_0=\beta_0=I_d/d\), which differs from the prescribed equilibrium. We use a common step size within each game size: \(\eta=0.13\), \(\eta=0.14\), and \(\eta=0.19\) for the 2-, 4-, and 6-qubit games, respectively. We report the duality gap at both the last strategy pair and the uniformly averaged strategy pair. {Figures~\ref{fig:fullrank-2qubit}--\ref{fig:fullrank-6qubit} show} the mean across the \(50\) instances with pointwise \(95\%\) Student-\(t\) confidence intervals. The \(y\)-axis is logarithmic, with a displayed lower limit of \(10^{-5}\).

To relate these step sizes to Theorem~\ref{thm:ogda-convergence-rate}, we compute the induced Frobenius norm \(L=\opnorm{\superop}\) for each instance, which gives a Lipschitz constant for the joint gradient field. The sufficient condition used in its proof is \(\eta\leq1/(8L)\). Across all \(50\) instances, the corresponding uniform upper bounds on the step size are approximately \(0.1383\), \(0.1473\), and \(0.1962\) for the three game sizes. Hence all three step sizes satisfy this condition, so the experiments run \ogda within the step-size range covered by the theorem.

\paragraph{Results.}
Figures~\ref{fig:fullrank-2qubit}--\ref{fig:fullrank-6qubit} report the 2-, 4-, and 6-qubit experiments. In all three game sizes, the last-iterate duality gap of \ogda decreases along a nearly straight line on the logarithmic axis, that is, at a linear rate. This matches Theorem~\ref{thm:ogda-convergence-rate}: these games have a full-rank equilibrium, and the step sizes satisfy the condition \(\eta\leq1/(8L)\). The average iterate of \ogda also improves, but more slowly, so by the end of each run the last iterate of \ogda is more accurate than its average iterate.

The other two methods behave differently. For \ommwu, the average-iterate gap decreases steadily, and the last-iterate gap also decreases, but much more slowly than for \ogda; this decrease becomes slower as the number of qubits grows. For \mmwu, the average-iterate gap decreases, whereas the last-iterate gap does not converge: it stays bounded away from zero and slowly increases over the run.

\begin{figure}[H]
    \centering
    \includegraphics[width=0.9\linewidth]{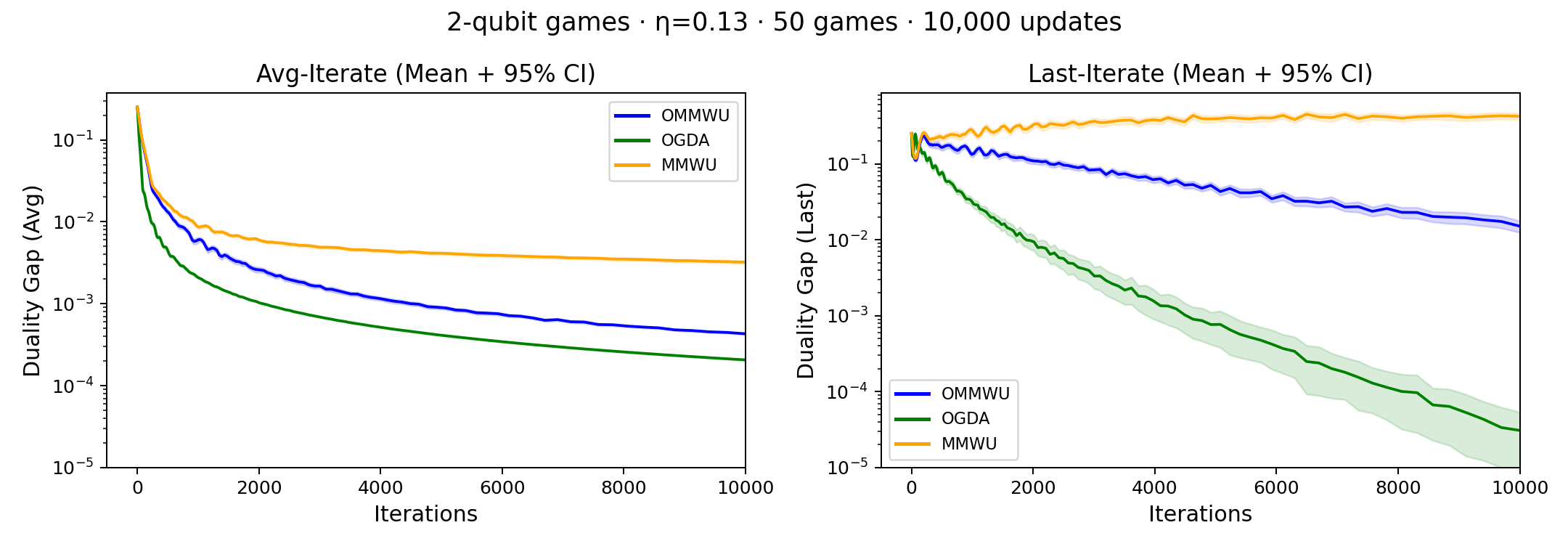}
    \caption{\textbf{Convergence on 2-qubit games with a known full-rank equilibrium.}
    Mean duality gap over \(50\) independent games for \(T=10{,}000\) iterations, using \(\eta=0.13\) for all methods. The left panel shows the average iterate and the right panel shows the last iterate, on a logarithmic \(y\)-axis. Shaded regions denote pointwise \(95\%\) confidence intervals over instances.}
    \label{fig:fullrank-2qubit}
\end{figure}

\begin{figure}[H]
    \centering
    \includegraphics[width=0.9\linewidth]{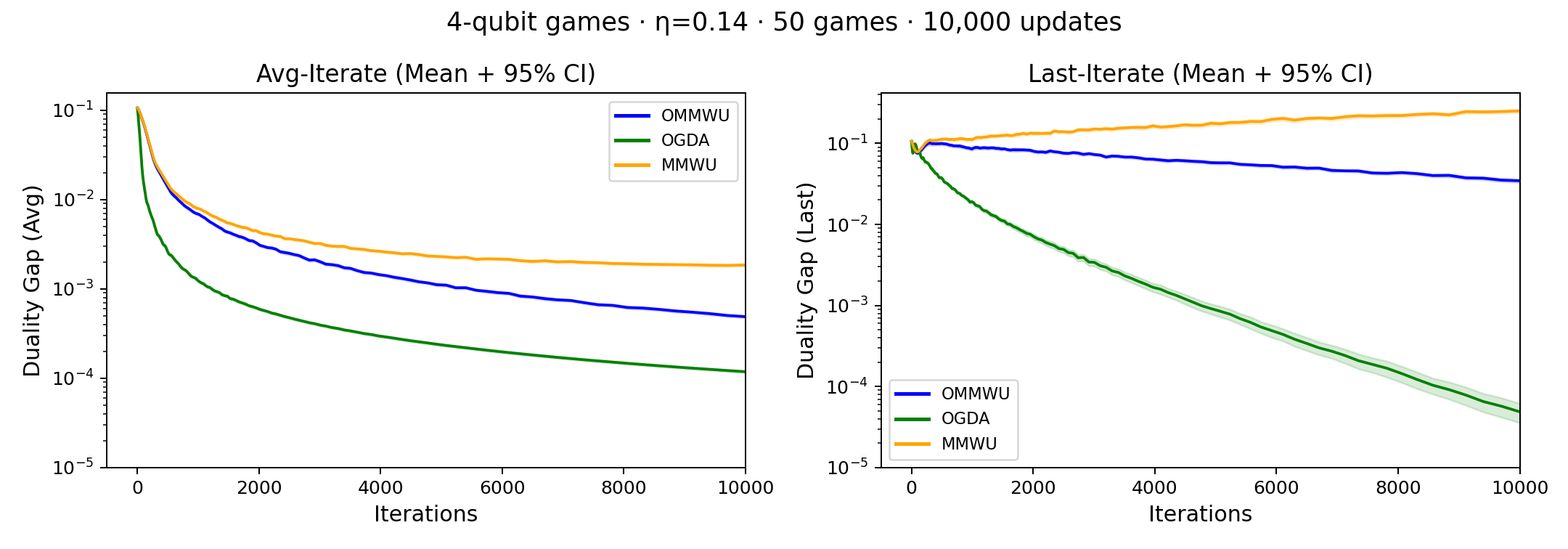}
    \caption{\textbf{Convergence on 4-qubit games with a known full-rank equilibrium.}
    Same setting as Figure~\ref{fig:fullrank-2qubit}, with \(\eta=0.14\) for all methods.}
    \label{fig:fullrank-4qubit}
\end{figure}

\begin{figure}[H]
    \centering
    \includegraphics[width=0.9\linewidth]{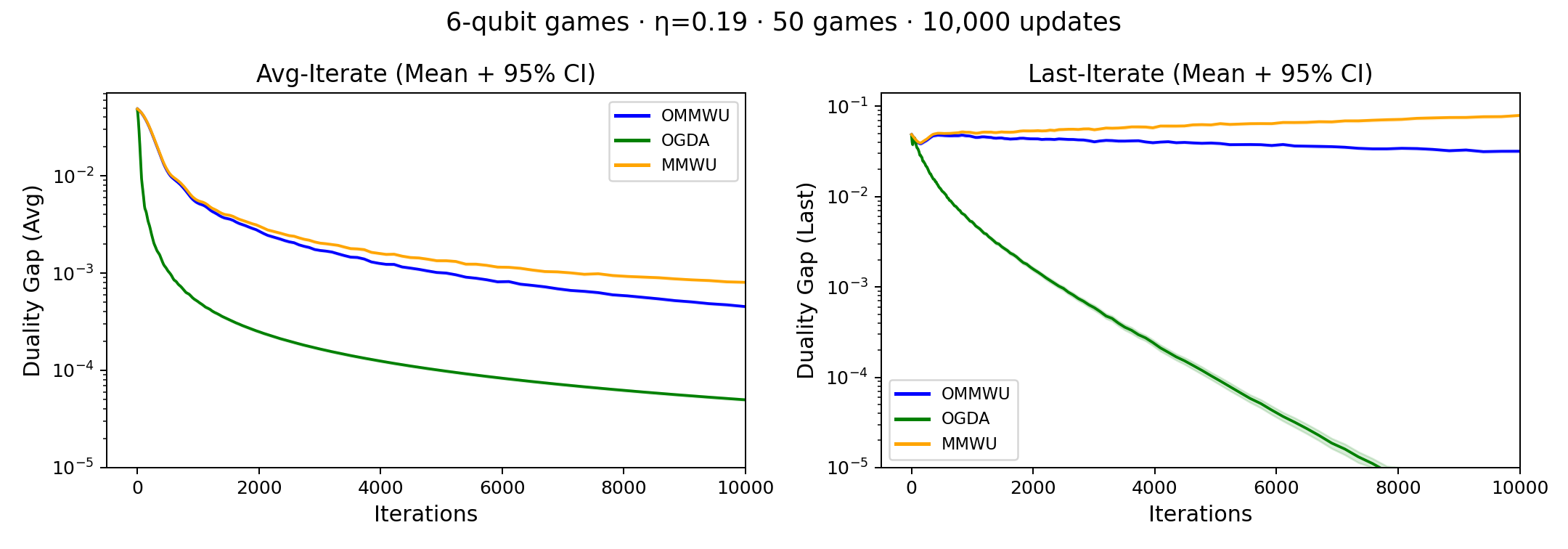}
    \caption{\textbf{Convergence on 6-qubit games with a known full-rank equilibrium.}
    Same setting as Figure~\ref{fig:fullrank-2qubit}, with \(\eta=0.19\) for all methods.}
    \label{fig:fullrank-6qubit}
\end{figure}

\subsection{Random POVM Games}
\label{sec:experiments_convergence}
We next test the methods on randomly generated 2-, 4-, and 6-qubit quantum zero-sum games, extending \citet[Experiment~1]{Vasconcelos2025quadraticspeedupin}. For each instance, we construct the payoff observable from randomly generated full-rank POVMs. These games are no longer constructed to have a full-rank equilibrium, and we use aggressive step sizes well beyond the condition of Theorem~\ref{thm:ogda-convergence-rate}. Our theory does not cover this regime, yet \ogda still converges quickly in the last iterate; analyzing such random semidefinite games is an interesting open problem. We generate $50$ independent games and run each method for $50{,}000$ iterations per game, using the step sizes stated below for each game size.  We measure performance by the duality gap along the trajectory and report both the last iterate and the average iterate. The plots report the mean across instances with $95\%$ confidence intervals and a log-scaled $y$-axis.

We first compare the three first-order methods on randomly generated 2-qubit quantum zero-sum games using a common step size \(\eta=1\). Figure~\ref{fig:2qubit-qzsg-convergence-ci95} reports both the average-iterate and last-iterate duality gaps. In the average-iterate plot (left), all three methods exhibit monotone decay on the log scale, with \ogda decreasing fastest and both \ommwu and \mmwu also converging steadily.

The distinction between the algorithms is more pronounced in the last-iterate plot (right). \ogda reduces the duality gap rapidly, while \ommwu also shows a clear decreasing trend, albeit with larger fluctuations across trials. In contrast, the last-iterate gap of standard \mmwu remains bounded away from zero and fluctuates around a non-vanishing level throughout the run. Thus, the experiment illustrates the theoretical separation discussed above: standard \mmwu achieves convergence primarily through averaging, whereas the optimistic variants reduce the oscillatory behavior and provide empirical evidence of last-iterate convergence in quantum zero-sum games.
\begin{figure}[H]
    \centering
    \includegraphics[width=0.9\linewidth]{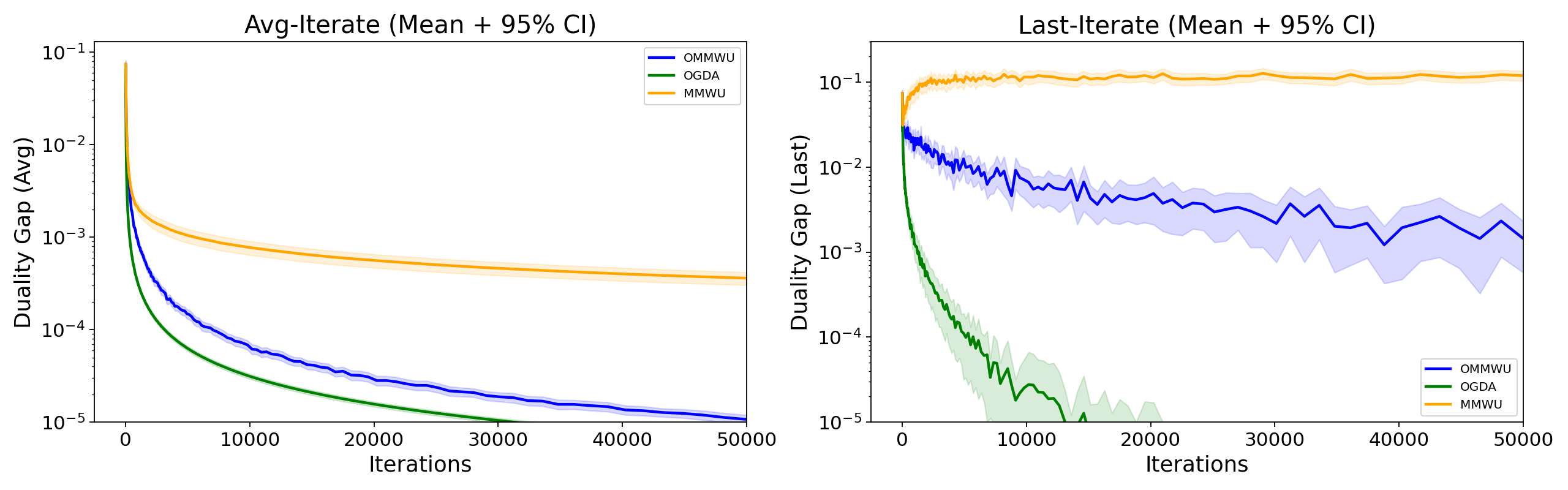}
   \caption{\textbf{Convergence on 2-qubit random quantum zero-sum games.}
    Mean duality gap over \(50\) independent games for \(T=50{,}000\) iterations, using \(\eta=1\) for all methods. The left panel shows the average-iterate duality gap and the right panel shows the last-iterate duality gap. The duality gap is plotted on a logarithmic \(y\)-axis. Shaded regions denote pointwise \(95\%\) confidence intervals over instances. All methods decrease in the average iterate. In the last iterate, \ogda\ and \ommwu\ decrease while standard \mmwu\ remains at a non-vanishing gap.}

    \label{fig:2qubit-qzsg-convergence-ci95}
\end{figure}

We next test whether the same behavior persists in larger randomly generated quantum zero-sum games. Figure~\ref{fig:4qubit-qzsg-convergence-ci95} reports the 4-qubit case, where Alice and Bob each control two qubits. We use step sizes \(\eta=10\) for \mmwu and \ommwu, and \(\eta=5\) for \ogda. The qualitative picture is consistent with the 2-qubit experiment: all three methods decrease the average-iterate duality gap, but the last-iterate behavior separates the optimistic methods from standard \mmwu. Both \ogda and \ommwu drive the last-iterate gap to the numerical floor, while the last-iterate gap of \mmwu remains bounded away from zero throughout the run.

Compared with the 2-qubit experiment, the confidence bands in the 4-qubit case are substantially narrower. This suggests that the slow last-iterate convergence behavior observed in some 2-qubit random instances is less prominent for these randomly generated 4-qubit games. One possible explanation is that random full-rank POVMs in the 2-qubit setting are more likely to produce instances with slow last-iterate convergence for \ommwu, as illustrated separately in Appendix~\ref{appendix:slow-last-iter-ommwu}. As the number of qubits increases, such slow-convergence instances appear harder to generate at random. We also use larger step sizes for \mmwu and \ommwu in the 4-qubit experiment to make the convergence behavior clearly visible over the finite horizon. Since the step sizes differ across game sizes, these comparisons are qualitative.

\begin{figure}[H]
    \centering
    \includegraphics[width=0.9\linewidth]{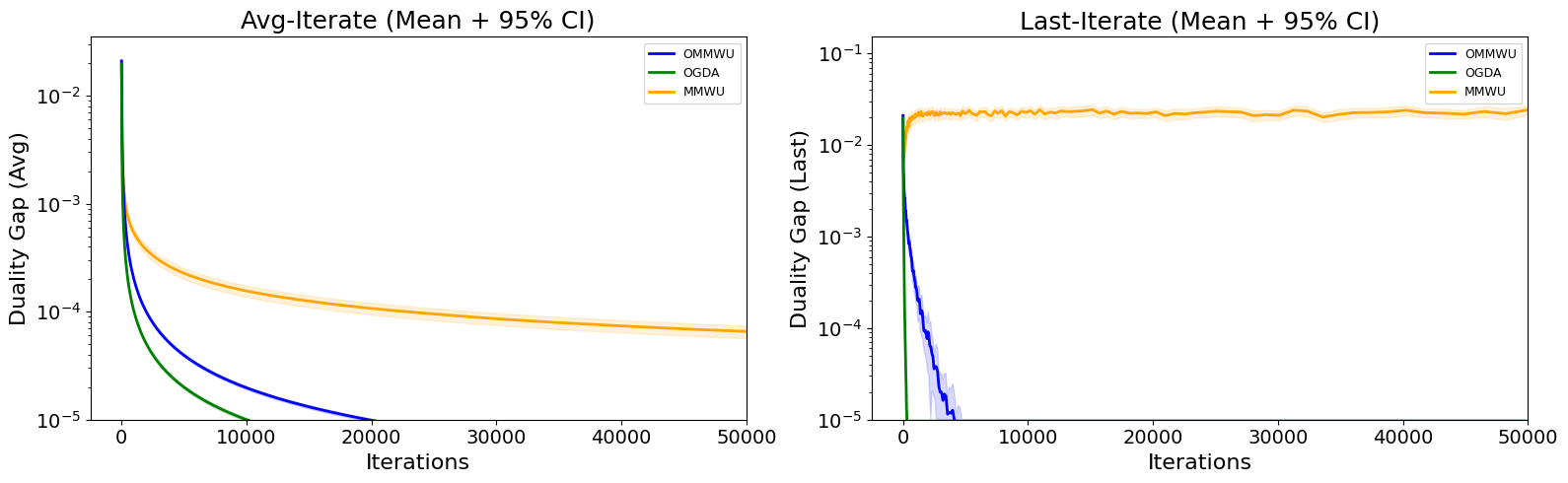}
    \caption{\textbf{Convergence on 4-qubit random quantum zero-sum games.}
    Mean duality gap over \(50\) independent games for \(T=50{,}000\) iterations; \mmwu\ and \ommwu\ use \(\eta=10\), while \ogda\ uses \(\eta=5\). The left panel shows the average-iterate duality gap and the right panel shows the last-iterate duality gap. The duality gap is plotted on a logarithmic \(y\)-axis. Shaded regions denote pointwise \(95\%\) confidence intervals over instances. The optimistic methods drive the last-iterate gap close to the numerical floor, whereas standard \mmwu\ remains bounded away from zero in the last iterate.}

    \label{fig:4qubit-qzsg-convergence-ci95}
\end{figure}

We finally consider the 6-qubit case, where Alice and Bob each control three qubits. We use the same step sizes as in the 4-qubit case, \(\eta=10\) for \mmwu and \ommwu, and \(\eta=5\) for \ogda. Figure~\ref{fig:6qubit-qzsg-convergence-ci95} shows the same qualitative pattern as the 4-qubit experiment, but with slower convergence for \ommwu. In the average iterate, all three methods decrease the duality gap. In the last iterate, \ogda reaches the numerical floor quickly, while \ommwu decreases steadily over the full horizon. By contrast, standard \mmwu remains bounded away from zero and gradually increases after the initial transient, again indicating persistent oscillatory behavior rather than last-iterate convergence. The confidence bands are narrower than in the 4-qubit experiment, further suggesting that slow last-iterate convergence for \ommwu is less common in these larger random instances.
\begin{figure}[H]
    \centering
    \includegraphics[width=0.9\linewidth]{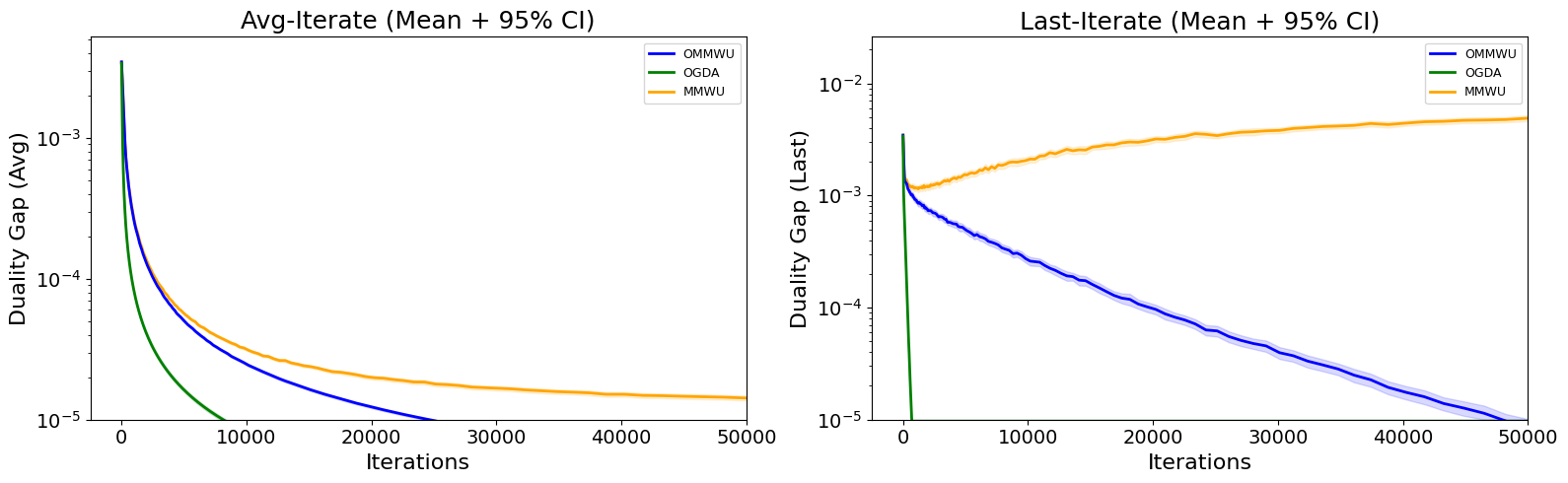}
    \caption{\textbf{Convergence on 6-qubit random quantum zero-sum games.}
    Mean duality gap over \(50\) independent games for \(T=50{,}000\) iterations; \mmwu\ and \ommwu\ use \(\eta=10\), while \ogda\ uses \(\eta=5\). The left panel shows the average-iterate duality gap and the right panel shows the last-iterate duality gap. The duality gap is plotted on a logarithmic \(y\)-axis. Shaded regions denote pointwise \(95\%\) confidence intervals over instances. \ogda\ reaches the numerical floor fastest, \ommwu\ decreases more gradually, and standard \mmwu\ stays bounded away from zero in the last iterate.}
    \label{fig:6qubit-qzsg-convergence-ci95}
\end{figure}

To complement the aggregate convergence plots, we also inspect the geometry of a single
2-qubit random instance. This diagnostic is not intended as a statistical comparison;
rather, it illustrates the trajectory-level behavior underlying the mean duality-gap
curves above. We use the same 2-qubit setting with \(\eta=1\) and run all methods for
\(T=50{,}000\) iterations on a fixed randomly generated full-rank POVM game.

Figure~\ref{fig:2qubit-single-instance-duality-gap} reports the average-iterate and
last-iterate duality gaps for this instance. The behavior is consistent with the
aggregate results. Standard \mmwu\ achieves a small average-iterate gap, but its
last-iterate gap remains large and oscillatory throughout the run. By contrast,
\ommwu\ has a decreasing last-iterate trend with large oscillations,
while \ogda\ rapidly drives the last-iterate gap to the numerical floor.

\begin{figure}[H]
    \centering
    \includegraphics[width=\linewidth]{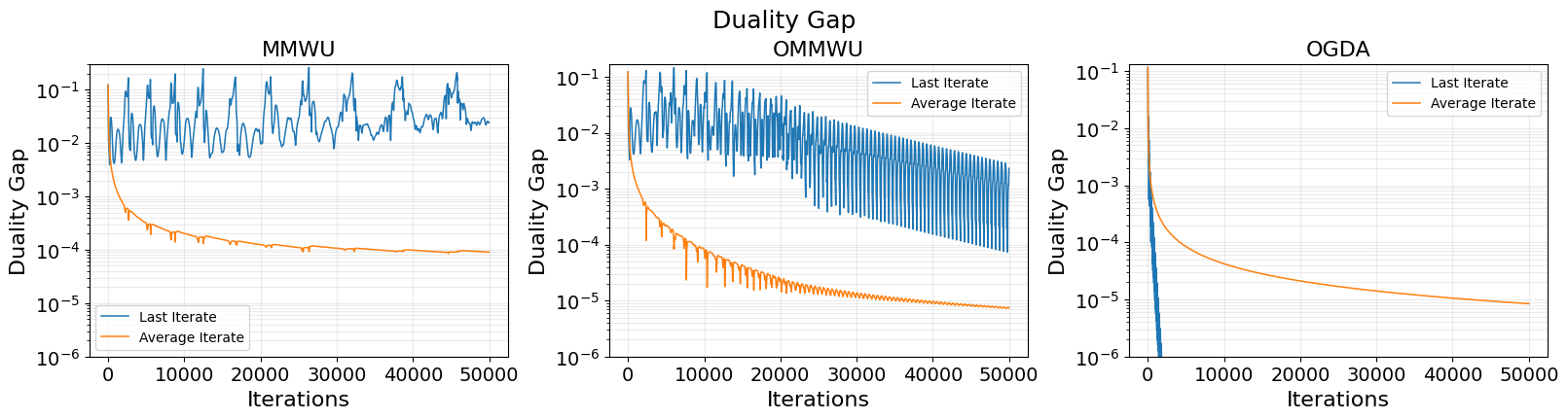}
    \caption{\textbf{Duality-gap trajectories for a single 2-qubit random quantum zero-sum game.}
    Each panel reports the last-iterate and average-iterate duality gaps for one method over
    \(T=50{,}000\) iterations. The duality gap is plotted on a logarithmic \(y\)-axis. 
    }
    \label{fig:2qubit-single-instance-duality-gap}
\end{figure}

Figure~\ref{fig:2qubit-single-instance-eigenvalues} tracks the eigenvalues of the
last-iterate product state \(\alpha_t \otimes \beta_t\). For standard \mmwu, the largest
eigenvalue repeatedly approaches one, while the other eigenvalues collapse toward zero
except for intermittent spikes. This indicates that the last iterates repeatedly move
near nearly pure product states, consistent with boundary cycling. For \ommwu, the
oscillations are gradually damped and the spectrum becomes more stable over time.
For \ogda, the spectrum stabilizes almost immediately after the initial transient.

\begin{figure}[H]
    \centering
    \includegraphics[width=\linewidth]{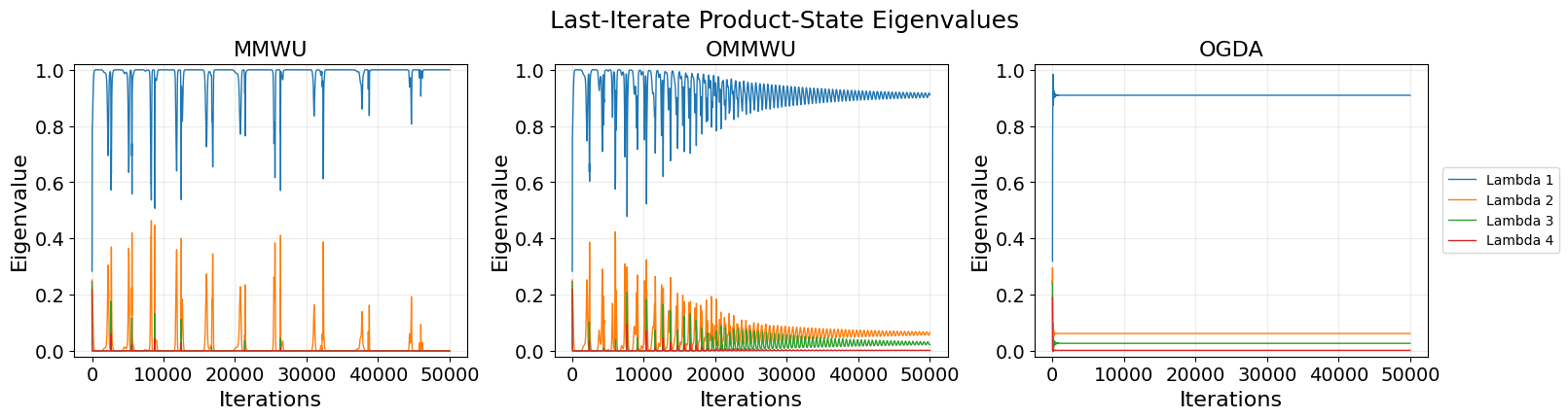}
    \caption{\textbf{Eigenvalues of the last-iterate product state for a single 2-qubit random quantum zero-sum game.}
    The plot tracks the spectrum of \(\alpha_t \otimes \beta_t\). 
    }
    \label{fig:2qubit-single-instance-eigenvalues}
\end{figure}

Finally, Figure~\ref{fig:2qubit-single-instance-bloch} visualizes the corresponding
Bloch-sphere trajectories of Alice's and Bob's last iterates. Standard \mmwu\ exhibits
broad recurrent motion for both players, especially for Bob, rather than settling near
a fixed point. The optimistic methods produce much more localized trajectories:
\ogda\ quickly concentrates near its terminal state, while \ommwu\ moves toward a
smaller region with damped oscillations. These trajectory-level diagnostics support
the same interpretation as the duality-gap plots: standard \mmwu\ converges mainly
through averaging, whereas optimism substantially suppresses cycling and improves
last-iterate behavior.

\begin{figure}[H]
    \centering
    \includegraphics[width=\linewidth]{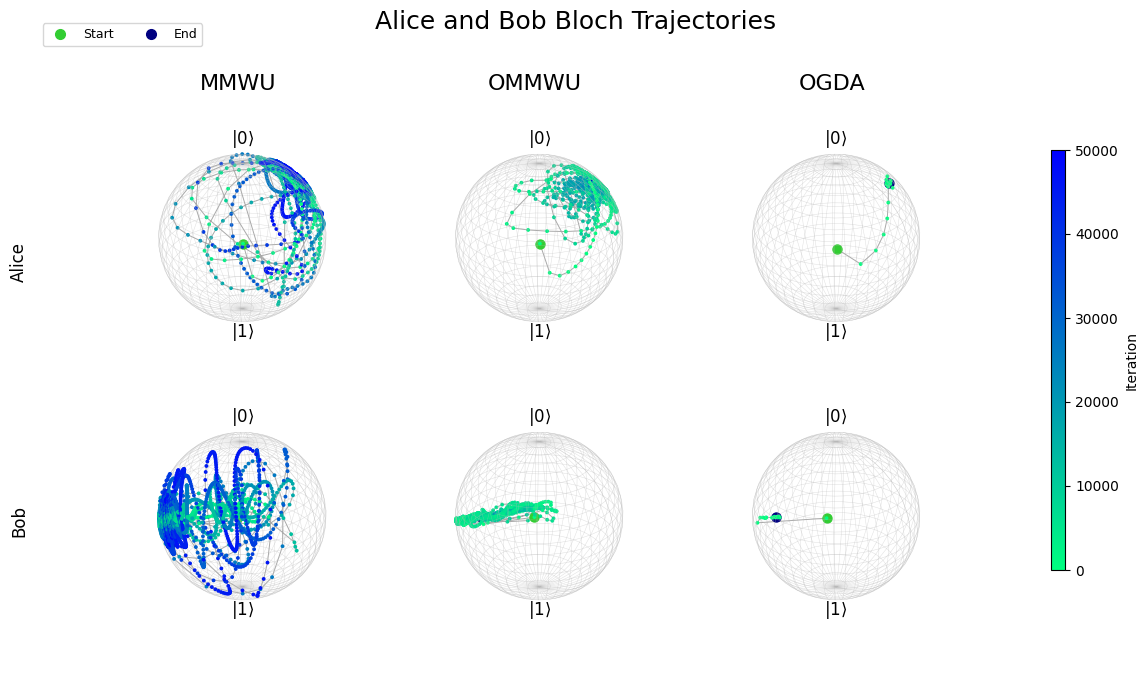}
    \caption{\textbf{Bloch-sphere trajectories of the last iterates for a single 2-qubit random quantum zero-sum game.}
    Rows correspond to Alice and Bob, and columns correspond to \mmwu, \ommwu, and \ogda. For visual clarity, the plots display every 10th saved iterate.
    }
    \label{fig:2qubit-single-instance-bloch}
\end{figure}

\subsection{Slow Last-Iterate Convergence of \ommwu\ on $U_{\delta}$}\label{appendix:exp:ommwu-Udelta}
Theorem~\ref{thm:name} gives a slow last-iterate lower bound for classical \ommwu\ on $A_{\delta}$, and Lemma~\ref{lem:diag-reduction} in Appendix \ref{appendix:slow-last-iter-ommwu}
transfers it to quantum \ommwu\ on the diagonal payoff observable $U_{\delta}$.
We complement this transfer with numerical experiments that track the last-iterate quantum duality gap along the \ommwu\ trajectory,
following and extending the experimental setup of \citet{cai2025fastlastiterateconvergencelearning}.

Concretely, we run \ommwu (Algorithm~\ref{alg:ommwu}) with the von Neumann entropy regularizer on the one-qubit-per-player quantum game with payoff observable \eqref{eqn:bad observable}
\begin{align*}
U_\delta \;=\; \diag\Big(\tfrac12+\delta,\tfrac12,0,1\Big), \qquad \payoff(\a,\b) \;=\; \tr{U_\delta^\dagger(\a\otimes\b)},
\end{align*}
and record the last-iterate quantum duality gap $\gap_Q(\a_t,\b_t)$ (defined in \eqref{eqn:duality-gap}) along the trajectory.
All plots record $t\mapsto \gap_Q(\a_t,\b_t)$ on a log-scaled $y$-axis.

In the first experiment, we use the diagonal initialization from Lemma~\ref{lem:diag-reduction},
\begin{align*}
\a_0 \;=\; \b_0 \;=\; \tfrac12\,\eye,
\end{align*}
fix the step size to $\eta=0.1$, and run for $T=10{,}000$ iterations, sweeping
$\delta\in\{0.005,\,0.01,\,0.05,\,0.1\}$.
In this setting, Lemma~\ref{lem:diag-reduction} implies that the induced trajectory coincides exactly with classical \ommwu\ on $A_\delta$
and preserves the duality gap:
\begin{align*}
\gap_Q(\a_t,\b_t) \;=\; \gap_C(x_t,y_t)
\qquad\text{for all }t.
\end{align*}
Thus, the experiment directly visualizes the transferred last-iterate phenomenon.

Figure~\ref{fig:hard-instance-Gdelta-sweep} shows a lack of last-iterate stabilization.
Across all $\delta$, the duality gap initially drops rapidly as the iterates approach equilibrium, but the trajectory does not
remain near the low-gap region. Instead, it eventually leaves and enters large oscillatory excursions.
This behavior matches the ``lack of forgetfulness'' identified by \citet{cai2025fastlastiterateconvergencelearning} for \oftrl-style
dynamics (of which \ommwu\ is a special case): historical gradient information accumulates, so the dynamics overshoot the equilibrium
and cycle rather than settling.
The dependence on $\delta$ is also visible: for smaller $\delta$, the trajectory can exhibit an extended ``flat'' interval in which
$\gap_Q(\a_t,\b_t)$ remains very small before the first large excursion, and the iteration of the first return to a large duality gap
moves later as $\delta$ decreases, consistent with the $1/\delta$ scaling in Theorem~\ref{thm:name}.

\begin{figure}[ht]
    \centering
    \includegraphics[width=0.8\textwidth]{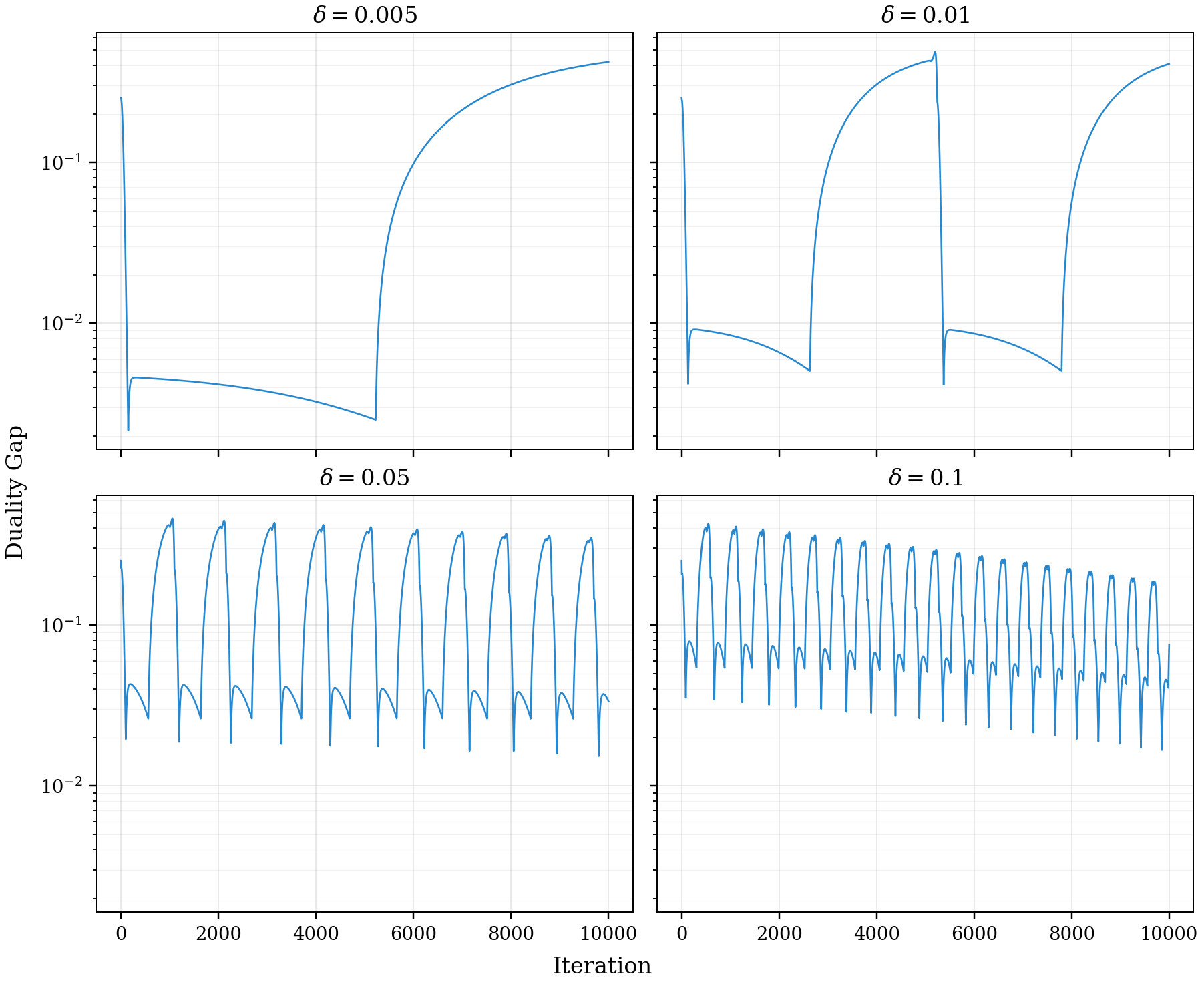}
    \caption{
    Last-iterate quantum duality gap trajectories $t\mapsto \gap_Q(\a_t,\b_t)$ for \ommwu\ on $U_\delta$
    under diagonal initialization $\a_0=\b_0=\tfrac12\eye$, with $\eta=0.1$.
    The trajectory can approach a small-gap region and later return to large excursions. Smaller $\delta$ yields longer delays before these excursions,
    consistent with slow last-iterate convergence on the hard instance.
    }
    \label{fig:hard-instance-Gdelta-sweep}
\end{figure}

In the second experiment, we test whether the last-iterate behavior in Figure~\ref{fig:hard-instance-Gdelta-sweep} is an artifact of the diagonal-invariant initialization from Lemma~\ref{lem:diag-reduction}, or whether it persists under genuinely quantum (non-diagonal) initial states.
We fix $\delta=0.05$ and initialize both players with an off-diagonal component,
\begin{align*}
\a_0 \;=\; \b_0 \;=\;
\begin{pmatrix}
\tfrac12 & \varepsilon\\
\varepsilon & \tfrac12
\end{pmatrix},
\end{align*}
sweeping $\varepsilon\in\{0,\,0.1,\,0.2,\,0.3,\,0.4,\,0.45\}$.
We run \ommwu\ with step size $\eta=0.1$ for $T=10{,}000$ iterations and record the last-iterate quantum duality gap trajectory
$t\mapsto \gap_Q(\a_t,\b_t)$ with log-scaled $y$-axis, as above.

Figure~\ref{fig:hard-instance-offdiag-sweep} shows that moving off the diagonal subspace does not eliminate the large-gap cycling observed in the first experiment.
Across the entire sweep, $\gap_Q(\a_t,\b_t)$ again drops rapidly at early times and later returns to recurring large excursions. Varying $\varepsilon$
primarily changes the phase and period of the oscillations rather than stabilizing the last iterate.
Thus, while the first experiment visualizes the exact classical reduction guaranteed by Lemma~\ref{lem:diag-reduction}, the second experiment indicates that the same
qualitative last-iterate phenomenon persists even when the initialization is non-diagonal and the trajectory explores off-diagonal quantum degrees of freedom.

\begin{figure}[t]
    \centering
    \includegraphics[width=\linewidth]{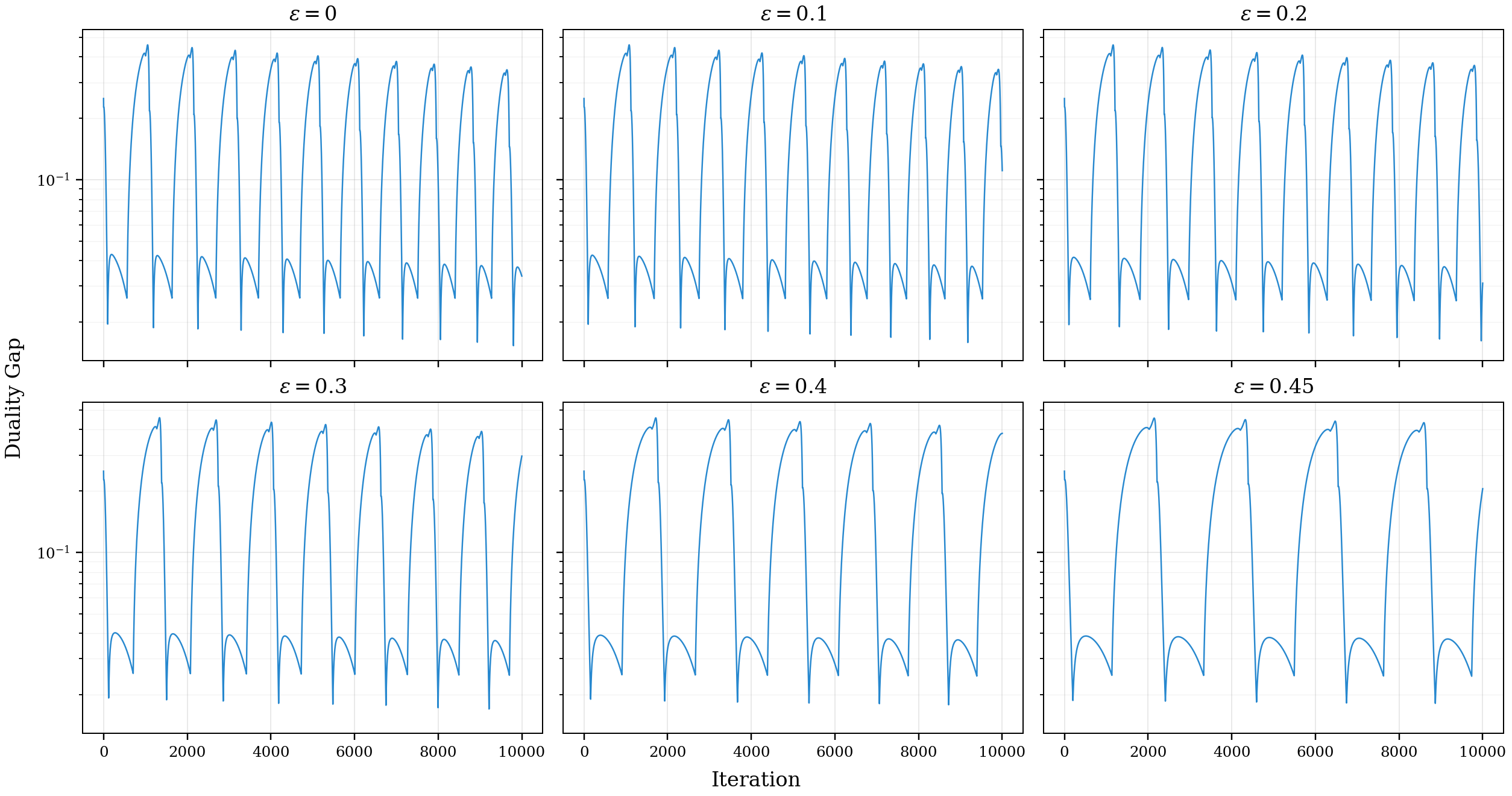}
    \caption{
    Last-iterate quantum duality gap trajectories $t\mapsto \gap_Q(\a_t,\b_t)$ for \ommwu\ on $U_\delta$
    with off-diagonal initialization, fixing $\delta=0.05$ and $\eta=0.1$ and sweeping $\varepsilon\in\{0,\,0.1,\,0.2,\,0.3,\,0.4,\,0.45\}$.
    Off-diagonal initial states do not suppress the recurring large-gap excursions.
    }
    \label{fig:hard-instance-offdiag-sweep}
\end{figure}

\newpage
\clearpage

\providecommand{\companion}{\citep{su2026lowerbounds}}
\interfootnotelinepenalty=10000 

\section{Epilogue: The Geometry of Acceleration}
\label{sec:epilogue}

In the space of a single year, von Neumann introduced the two objects that this paper has been playing against each other: the min--max point of a two-person game, presented at G\"ottingen in December 1926~\citep{v1928theorie}, and the density matrix, the mixed state of a quantum system~\citep{vonneumann1927wahrscheinlichkeit}. A century later, the first is computed routinely by first-order methods, at a rate that is \emph{linear} whenever the strategy sets are polyhedral; the second is the strategy set of a quantum game, and it is not polyhedral. Question~\hyperref[main-question]{$(\star)$} asked what the algorithmic theory of the first object looks like when it is made to live on the second. In this closing section we try to say, in plain words, what the answer is, why it takes the form it does, and what it leaves open; we also record what the lower bounds of the companion paper~\companion, obtained after the present results, add to the picture.

\begin{center}
\begin{minipage}{0.92\linewidth}\centering\itshape
The curvature of the spectraplex does not forbid linear convergence; it prices it.\\
With a full-rank equilibrium, the price is a condition number;\\
without one, it is an exponent.
\end{minipage}
\end{center}
\paragraph{Where linearity comes from: facets versus interior.}
The polyhedral argument for linear convergence is, at bottom, a counting argument. A simplex has finitely many vertices; an equilibrium supported on a face leaves every vertex outside that face \emph{uniformly} suboptimal, and this positive margin is the Hoffman constant that converts a small duality gap into a small distance~\citep{hoffman1963simple,tseng1995,Gilpin2012FirstOrderAlgorithm}. The spectraplex has no such margin to offer. Its extreme points form a manifold of rank-one projectors, and a sequence of them can have vanishing gap without ever approaching the equilibrium set (Section~\ref{sec:challenges}); the elementary example $f(x)=x^4$ is the whole difficulty in miniature, since an objective gap may vanish like a power of the distance strictly larger than one.

What takes the place of the vertex count is, somewhat paradoxically, the \emph{interior}. If some equilibrium $(\bar\astate,\bar\bstate)$ has both states positive definite, then complementary slackness, which in general only says that a positive semidefinite slack has zero trace against the equilibrium state, forces the slack itself to vanish: $\superop(\bar\bstate)=\gameval\,\eye_{\aspace}$ and $\adj{\superop}(\bar\astate)=\gameval\,\eye_{\bspace}$ (Claim~\ref{claim:spms-affine-optimality}). In words, a full-rank equilibrium is an \emph{equalizer}: against it, every pure quantum strategy of the opponent is exactly indifferent. The optimal sets then become affine slices of the spectraplex; affine slices obey a Hoffman bound; and positivity, the only non-affine constraint left, is repaired by mixing with the interior point (Claim~\ref{claim:spms-affine-repair}). This is the mechanism behind Theorem~\ref{thm:sp-ms-for-spectraplex}, and it inverts the classical intuition: on the simplex, linearity is a property of the \emph{boundary}; on the spectraplex, it is a property of the \emph{interior}. The price is visible in the constants. The constant $K$ produced by the proof grows with $1/\lambda_{\min}(\bar\astate)$ and $1/\lambda_{\min}(\bar\bstate)$, each at least the corresponding local dimension, and with the conditioning of $\superop$ on trace-zero matrices; this is the precise sense in which $\kappa(\superop)$ is a quantity of the game rather than of the accuracy.

Without an interior equilibrium, semialgebraic geometry still guarantees $\dist(\joint)^q\le K\gap(\joint)$ for some finite exponent $q$ (Claim~\ref{claim:spms-semialgebraic}), but the \L{}ojasiewicz argument is non-constructive: it tells us that an exponent exists, not what it is, and not that it is one. Nor is it one in general. Our companion paper~\companion\ exhibits one-qubit games, each with a unique pure equilibrium, on which the duality gap decays like the \emph{square} of the distance along an \ommwu trajectory and like its \emph{cube} along an \ogda trajectory; an immediate consequence is that no bound $\dist(\joint)^q\le K\gap(\joint)$ holds on these games with $q<2$, respectively $q<3$.\footnote{\label{fn:versions}%
A note on earlier versions. In versions~1 and~2 of this paper (arXiv:2509.21570, September 2025 and June 2026) we claimed Theorem~\ref{thm:sp-ms-for-spectraplex} with exponent one for every quantum zero-sum game. The proof adapted the normal-cone argument of \citet{wei2021linear} from the simplex to the spectraplex, and the step that produces a uniform Hoffman-type constant breaks down near a rank-deficient equilibrium. We did not see this at first. A comment from a referee at QTML~2026 sent us looking for non-commutative lower bounds, and once we had the one-qubit games of our companion paper~\companion, with exponents at least two and three, it was clear that the initial claim should be revised appropriately. In the current final corrected version, what survives of the old statement is the actual qualitative part: a finite exponent for every game. What replaces it is, we think, the more interesting theorem: exponent one exactly when there is an interior equilibrium, and lower bounds showing that this hypothesis cannot be dropped. Everything in the present version is stated under it, Conjecture~\ref{conj:ommwu-linear} included.}
The full-rank hypothesis is therefore not an artifact of our proof, and the mechanism is the curvature we started from: rotating the support of a rank-one equilibrium by an angle $\theta$ moves the state by $\Theta(\theta)$ in Frobenius norm, while the mass it places outside the equilibrium support, and with it, in those games, the gap, is only $\Theta(\theta^2)$.

\paragraph{Two ingredients, two geometries, two prices.}
It is worth separating two things that the phrase ``linear last-iterate convergence'' bundles together. \emph{Optimism buys convergence; the error bound buys the rate.} Plain \mmwu, like its classical counterpart, recurs rather than converges~\citep{mertikopoulos2017cyclesadversarialregularizedlearning,jain2022matrix}, and our experiments show exactly this (Figure~\ref{fig:2qubit-single-instance-eigenvalues}): the last iterate keeps visiting nearly pure product states. Optimism damps the recurrence~\citep{daskalakis2018last,mertikopoulos2018optimistic}; it is the error bound that turns damping into a geometric rate. Once the bound is in hand the algorithmic story is modular, since SP-MS separates the geometry of the game from the dynamics of the method, and any optimistic or smoothing scheme whose Lyapunov analysis runs on SP-MS inherits the linear rate without further quantum input~\citep{wei2021linear,hsieh2019convergence,malickrate}. This is how \itersmooth and \ogda obtain their $\bigoh(\kappa(\superop)\log(1/\eps))$ and $\bigoh(\log(1/\eps))$ iteration counts (Theorems~\ref{thm:q-iterated-convergence-rate-fin} and~\ref{thm:ogda-convergence-rate}), and it is also why those counts carry a condition number and, through the Frobenius Lipschitz constant, a dependence on dimension. The hypothesis is used, not merely convenient: on the one-qubit game of \companion\ whose unique equilibrium is pure, \ogda from the maximally mixed state has distance $\Theta(1/t)$ and gap $\Theta(1/t^{3})$ for every sufficiently small fixed step size, the quantum counterpart of the curved-domain example of \citet[Theorem~9]{wei2021linear}. Outside the full-rank regime, then, the H\"older exponent is not a technicality: it is the rate.

The entropic geometry pays differently. \ommwu is intrinsic to density matrices: it never leaves the spectraplex, it needs no projection, and its iterates are Gibbs states, the same objects that drive the parallel algorithm of \citet{jain2009parallel} and the quantum SDP solvers of \citet{brando2019}. Its last-iterate guarantee in this paper, however, goes through the quantal response equilibrium (Theorem~\ref{thm:ommwu-convergence}): linear convergence to a regularized target, and $\widetilde\bigoh(1/\eps)$ to Nash once the regularization is tuned to the accuracy. With the contraction rate of that theorem this is not an artifact that a better schedule would remove. The factor $(1+\delta\eta)^{-1}$ degrades with the regularization strength $\delta$, so an outer loop that sends $\delta\to0$ geometrically pays $\Theta(1/(\delta_i\eta))$ at stage $i$ and $\Theta(1/\eps)$ overall. A logarithmic rate for the entropic method, if it exists, must come from the \emph{unregularized} dynamics' own contraction near equilibrium, that is, from an error bound expressed in the Bregman geometry of the von Neumann entropy rather than in the Frobenius norm.

We do not have such a bound, and we now know that uniqueness of the equilibrium does not supply one, even together with strict complementarity. On the game $U_M$ of \companion, whose unique equilibrium is pure and strictly complementary, unregularized \ommwu from the maximally mixed state converges only polynomially for \emph{every} fixed step size: distance and relative entropy decay as $\Theta(1/t)$, the gap as $\Theta(1/t^2)$, and no geometric estimate holds even with game-dependent constants. Conjecture~\ref{conj:ommwu-linear} is thus false as originally stated, and we have restated it with a full-rank hypothesis. The mechanism deserves to be spelled out, because it has no classical counterpart. Unrolling the matrix exponential, Alice's iterate at time $t$ can be written as $\Lambda(\kappa t Z-p_tX)$ where $\kappa = \bigoh(\eta)$, $p_t\to p_\infty>0$ as $t \rightarrow \infty$, and $X,Z$ are Pauli matrices. The diagonal score of the iterate grows linearly, so her \emph{eigenvalues} concentrate exponentially fast. However, the exponentiation preserves \emph{eigenvectors}, and the off-diagonal score, accumulated when Bob was still moving, tilts her leading eigenvector away from the equilibrium support by an angle of order $p_\infty/(\kappa t)$. In the decomposition of \citet{lotidis2023learning}, the commutative part of the dynamics converges like classical \texttt{OMWU}; the non-commutative part does not. Support identification, the engine of the classical proof~\citep{wei2021linear}, therefore has two halves on the spectraplex, and optimism as currently formulated controls only the first. Section~\ref{subsec:ommwu-slowdown} adds that even the constants of a rate cannot be uniform across games, since the diagonal subspace is invariant under \ommwu (Lemma~\ref{lem:diag-reduction}) and every quantum zero-sum game inherits the classical hard instances of \citet{cai2025fastlastiterateconvergencelearning}.
\paragraph{When does acceleration pay?}
The two kinds of bound also answer a practical question. Against the average-iterate rate $\bigoh(\log d/\eps)$ of unregularized \ommwu~\citep{Vasconcelos2025quadraticspeedupin}, whose dependence on dimension is only logarithmic, a linear rate $\bigoh(\kappa(\superop)\log(1/\eps))$ wins only once $\eps$ falls below roughly $\log d/\kappa(\superop)$, up to logarithmic factors. Linear convergence is thus a high-accuracy phenomenon, and its value depends on how large $\kappa(\superop)$ is for the games one actually meets. The $1/\eps$ on the other side of this comparison is not an artifact either: on a diagonal one-qubit game, the uniform average that retains the maximally mixed initial state has gap at least $1/T$ no matter how the iterates are generated~\companion, so the rate of \citet{Vasconcelos2025quadraticspeedupin} is tight for that output rule, and acceleration has to come from the last iterate or from a different averaging rule. Our experiments (Section~\ref{sec:experiments_convergence}) are suggestive on both counts: \ogda is the fastest method in every regime we tested, \ommwu converges in the last iterate with damped oscillations, and the slow instances that the diagonal embedding predicts become harder to find at random as the number of qubits grows. None of this verifies the full-rank hypothesis on the instances tested, nor the step-size restrictions of the theorems; we offer it as evidence of the phenomenon, not of the theorems.

\paragraph{The vocabulary of degeneracy.}
Everything above turns on one equilibrium being interior, and Appendix~\ref{appendix:strict complementarity} develops the vocabulary for what happens when none is. Slack operators partition each player's directions into \emph{essential} ones, used by some equilibrium; \emph{non-essential} ones, strictly suboptimal against every equilibrium; and \emph{neutral} ones, payoff-tight against some equilibrium opponent yet carrying no mass at any equilibrium (Definition~\ref{def:three-direction-types}). Neutral directions exist classically only in degenerate games, and there they are harmless for error bounds, because polyhedrality supplies linearity regardless. The class vanishes when every equilibrium is strictly complementary (Theorem~\ref{thm:SC-collapse}); the full-rank hypothesis is the extreme case of strict complementarity in which the slack vanishes identically. Classically, Goldman--Tucker guarantees that a strictly complementary equilibrium always exists (Theorem~\ref{thm:classical-sc-existence}); in semidefinite games it need not. One might hope that strict complementarity is the right intermediate hypothesis between nothing and full rank. It is not: the game $U_M$ above is strictly complementary with a unique equilibrium and still has exponent at least two. What the slow games of \companion\ have in common is simpler and, we believe, decisive: in each of them every equilibrium has a singular component, and the spectraplex is curved exactly there.

\clearpage

\clearpage
\paragraph{Future work.}
Three directions seem to us the most pressing.
\begin{enumerate}[leftmargin=*,itemsep=3pt]
\item \textbf{Beyond worst-case analysis.}
Every statement in this paper is worst-case over games, and the worst case is now known to be bad: the exponent can reach three already on the Bloch ball, and $\kappa(\superop)$ carries the local dimensions through $1/\lambda_{\min}$ of the equilibrium states. Two questions remain on that side. Is a full-rank equilibrium \emph{necessary} for a linear bound, and is the exponent determined by the ranks of the equilibrium states and their slacks, as the rotation mechanism suggests? For linear matrix inequalities the answer runs through the singularity degree~\citep{sturm2000errorbounds,drusvyatskiy2017manyfaces,sremac2021errorbounds}, and a semidefinite-game analogue would make the exponent of Theorem~\ref{thm:sp-ms-for-spectraplex} computable. And is the dependence on $\kappa(\superop)$ necessary, in the sense of an $\Omega(\kappa(\superop)\log(1/\eps))$ lower bound for first-order methods? The more consequential question, however, is how often the worst case occurs. The narrowing confidence bands of Section~\ref{sec:experiments_convergence} hint that ill-conditioned instances become rare as the dimension grows; a smoothed or average-case analysis of $\kappa(\superop)$, or more basically of the probability that a random game has a full-rank equilibrium at all, would turn that hint into a statement, and would say whether the regime in which this paper's rates apply is the typical one.
\item \textbf{Conjecture~\ref{conj:ommwu-linear}.}
The restated conjecture asks for linear last-iterate convergence of unregularized \ommwu under a unique \emph{full-rank} equilibrium. The hypothesis is natural: near an interior point the relative entropy is locally quadratic, so the entropic and Frobenius geometries agree to first order and Corollary~\ref{prop:main:condition-number} applies. What is missing is an argument that optimism controls eigenvector motion when the limit is interior; when the limit is singular, the question becomes whether an optimistic scheme that also extrapolates in the eigenbasis can restore geometric alignment. More than a rate hangs on the answer: the method whose iterates are quantum-native Gibbs states is the one that provably stalls on eigenvector alignment, so the conjecture decides whether the hardware-native dynamics can be accelerated at all.
\item \textbf{Feedback.}
Our players see $\superop(\bstate)$ and $\adj{\superop}(\astate)$ exactly. In a learning model they see what the referee measures, that is, outcomes drawn from the Born rule, and the gradient is only available as a noisy estimate from finitely many shots. Error bounds are the standard tool for converting such noise into convergence to a neighborhood of equilibrium, and the full-rank hypothesis should then quantify the size of that neighborhood in terms of $\kappa(\superop)$ and the number of shots. The quantum learning dynamics of \citet{lotidis2023learning,Lin_2024} are the natural place to carry this out, and the companion lower bounds suggest that, without a full-rank equilibrium, noise will interact with eigenvector motion in ways that the classical theory does not anticipate.
\end{enumerate}
\paragraph{Coda.}
The question we began with asked whether gradient methods can converge linearly in quantum zero-sum games. The answer this paper supports, and the companion lower bounds confirm from the other side, is: yes, whenever the game is well conditioned, and the spectraplex says precisely what that means: an interior equilibrium and a small $\kappa(\superop)$. Without the interior equilibrium, the rate is polynomial and the exponent is geometric. How often the good case occurs is a question about probability; whether the geometry native to quantum states can be made to pay the same price is a question about entropy. That both are, in the end, questions about games would not have surprised the man who gave us all three.
\newpage

\newpage

\bibliographystyle{quantum.bst}
\bibliography{refs.bib}

\newpage
\appendix

\section{Omitted Preliminaries}
\label{appendix:OmittedPreliminaries}
\subsection{Notations}
\label{appendix:notations}
For clarity, we collect here the main notation used in the paper.

\begin{itemize}[leftmargin=*]
  \item \textbf{Hermitian matrices.}  
  $\hmat^d := \{ A \in \C^{2^d \times 2^d} : A = A^\dagger \}$,  
  where $A^\dagger$ is the conjugate transpose of $A$.  
  The cone of positive semidefinite Hermitian matrices is  
  $$\hmat^d_{+} := \{ A \in \hmat^d : A \succeq 0 \}.$$

  \item \textbf{Identity matrix.} Let $\aspace \defeq \hmat^d_{+}$. We write $\eye_{\aspace}$ for the $2^d\times 2^d$ identity matrix (acting on $\aspace$).

  \item \textbf{Trace.}  
  For $A \in \C^{d \times d}$, the trace is  
  $\tr{A} := \sum_{i=1}^d A_{ii}$.

  \item \textbf{Inner product.}  
  For $A,B \in \C^{d \times d}$, the Hilbert–Schmidt inner product is  
  $$\inner{A,B} := \tr{A^\dagger B}.$$

  \item \textbf{Frobenius norm.}  
  For $A \in \C^{d \times d}$,  
  $\|A\|_F := \sqrt{\inner{A,A}} = \sqrt{\tr{A^\dagger A}} 
  = \bigl(\sum_{i=1}^d \sum_{j=1}^d |A_{ij}|^2\bigr)^{1/2}$.

  \item \textbf{Dual norm.}  
  For a normed space $(\mathcal{X}, \|\cdot\|)$ with dual $\mathcal{X}_*$,  
  the dual norm of $y \in \mathcal{X}_*$ is  
  $\|y\|_* := \sup \{ \inner{y,x} : x \in \mathcal{X}, \|x\| \le 1 \}$.  
  The Frobenius norm is self-dual: $\|x\|_* = \|x\|_F$ for all $x$.

  \item \textbf{Adjoint operator.}  
  For a linear map $F : \mathcal{X} \to \mathcal{Y}_*$, the adjoint  
  $\adj{F} : \mathcal{Y} \to \mathcal{X}_*$ is defined by  
  \[\text{$\inner{F(x),y} = \inner{x,\adj{F}(y)}$ for all $x \in \mathcal{X}, y \in \mathcal{Y}$.}\]

  \item \textbf{Operator norm.}  
  For a linear map $F : \mathcal{X} \to \mathcal{Y}_*$, the operator norm is  
$$
\|F\|_{\mathrm{op}} := \sup_{\|x\|=1,\;\|y\|=1} \inner{F(x),y}.
$$
\end{itemize}

Finally, crucial to our analysis are the notions of strong convexity and smoothness, which govern the curvature properties of functions and play a central role in determining algorithmic convergence rates.
\begin{definition}[$\mu$-Strong Convexity]\label{def:strong_convex}
A differentiable function $f:\mathcal{X}\to\R$ is \emph{$\mu$-strongly convex} with respect to a norm $\|\cdot\|$ if, for all $X,Y \in \mathcal{X}$,
\[
f(X) \;\ge\; f(Y) + \langle \nabla f(Y), X-Y \rangle + \tfrac{\mu}{2}\|X-Y\|^2.
\]
\end{definition}

\begin{definition}[$\beta$-Smoothness]\label{def:smooth}
A differentiable function $f:\mathcal{X}\to\R$ is \emph{$\beta$-smooth} with respect to a norm $\|\cdot\|$ if, for all $X,Y \in \mathcal{X}$,
\[
f(X) \;\le\; f(Y) + \langle \nabla f(Y), X-Y \rangle + \tfrac{\beta}{2}\|X-Y\|^2.
\]
\end{definition}

\subsection{Smoothness Properties of the Smoothed Duality Gap}
\label{appendix:prop:smoothness}
\begin{restatable}{proposition}{smoothnessprop}\label{prop:smoothness}
Suppose the regularizer $\reg$ is $\sigma_{\reg}$-strongly convex.
Then $\gap_\smooth(\joint)$ is smooth with gradient $\grad_\joint \gap_\smooth(\joint)=\gradient(\maximizer(\joint)),$
and $\grad_\joint \gap_\smooth(\joint)$ is Lipschitz with constant $L_\smooth=\opnorm{\gradient}^2 / (\smooth\sigma_{\reg}).$
\end{restatable}
\begin{proof}
    We extend $\gradient$ linearly to the real Hilbert space $\mathcal H:=\hmat^{\adim}\times\hmat^{\bdim}$, with inner product $\inner{(A,B),(C,D)}=\tr{AC}+\tr{BD}$ and norm $\fnorm{(A,B)}^2=\fnorm{A}^2+\fnorm{B}^2$. Traces of products of pairs denote the sum of the componentwise traces.
    We know the joint payoff gradient is in Equation $\eqref{eqn:payoff-gradient}$ and its adjoint
    \begin{align*}
        & \gradient(\joint) = \parens{\gradient_\a(\b), \gradient_\b(\a)} = \parens{\superop(\b), -\adj{\superop}(\a)} \\
        & \adj{\gradient}(\joint) = \parens*{-\superop(\b),\,\adj{\superop}(\a)}.
    \end{align*}
    We can rewrite the $\gap_\smooth(\joint)$ in Equation \eqref{eqn:smoothed-gap} using the adjoint identity \(\adj{\gradient}=-\gradient\)
    \begin{align} \label{eqn:smoothed gap gradient}
        \gap_\smooth(\joint) = \max_{\jointalt \in \jointsp} \braces*{\tr{\gradient(\jointalt) \joint} - \smooth \reg(\jointalt)}.
    \end{align}
    Then, the gradient of $\gap_\smooth(\joint)$ in $\joint$ is
    \begin{align*}
        \grad_\joint \gap_\smooth (\joint) =  \gradient(\maximizer(\joint)).
    \end{align*}
    To establish Lipschitz continuity of the gradient, we adapt Nesterov’s technique from the proof of Theorem 1 in \citet{Nesterov2005}. Suppose we have two arbitrary joint states $\joint_1, \joint_2 \in \jointsp$. By the first-order optimality condition, we have
    \begin{align}
        & \tr{\bigl(\adj{\gradient}(\joint_1) - \smooth \grad \reg(\maximizer(\joint_1))\bigr)(\maximizer(\joint_2) - \maximizer(\joint_1))} \leq 0,  \\
        & \tr{\bigl(\adj{\gradient}(\joint_2) - \smooth \grad \reg(\maximizer(\joint_2))\bigr)(\maximizer(\joint_1) - \maximizer(\joint_2))} \leq 0 .
    \end{align}
    Adding them and applying the strong convexity of $\reg$ and the adjoint identity, we have
    \begin{align*}
        \tr{\adj{\gradient}(\joint_1 - \joint_2)(\maximizer(\joint_1) - \maximizer(\joint_2))} & \geq \smooth \mathop{\operatorname{Tr}} \bigl[\parens*{\grad \reg(\maximizer(\joint_1)) - \grad\reg(\maximizer(\joint_2))} \parens*{\maximizer(\joint_1) - \maximizer(\joint_2)}\bigr] \\
        & \geq \smooth \sigma_{\reg}\fnorm*{\maximizer(\joint_1) - \maximizer(\joint_2)}^2.
    \end{align*}
    We apply the definition of operator norm
    \begin{align}
        \opnorm{\gradient} = \max_{\joint \in \mathcal H, \jointalt \in \mathcal H} \setdef{\tr{\gradient(\joint)^{\dagger} \jointalt}}{\fnorm{\joint} = 1, \fnorm{\jointalt} = 1}.
    \end{align}
    Therefore, for any $\joint \in \mathcal H$,
    \begin{align*}
        \dnorm{\gradient(\joint)} \leq \opnorm{\gradient} \fnorm{\joint}
    \end{align*}
    Following the argument of Nesterov,
    \begin{align*}
        \dnorm{\gradient(\maximizer(\joint_1)) - \gradient(\maximizer(\joint_2))}^2 & \leq \opnorm{\gradient}^2 \cdot \fnorm{\maximizer(\joint_1) - \maximizer(\joint_2)}^2 \\
        & \leq \frac{1}{\smooth \sigma_{\reg}} \opnorm{\gradient}^2 \cdot \tr{\adj{\gradient}(\joint_1 - \joint_2)(\maximizer(\joint_1) - \maximizer(\joint_2))} \\
        & = \frac{1}{\smooth \sigma_{\reg}} \opnorm{\gradient}^2 \cdot \tr{\gradient(\maximizer(\joint_1) - \maximizer(\joint_2))(\joint_1 - \joint_2)} \\
        & \leq \frac{1}{\smooth \sigma_{\reg}} \opnorm{\gradient}^2  \cdot \dnorm{\gradient(\maximizer(\joint_1) - \maximizer(\joint_2))} \cdot \fnorm{\joint_1 - \joint_2},
    \end{align*}
    which proves the smoothness of $\gap_{\smooth}$. The Lipschitz constant is
    \begin{align}
        L_\smooth = \frac{1}{\smooth \sigma_{\reg}} \opnorm{\gradient}^2.
    \end{align}
\end{proof}

\section{Slack Geometry and Equilibrium Structure}
\label{appendix:strict complementarity}
This section develops a geometric description of Nash equilibria in semidefinite zero-sum games through slack operators and their associated tight subspaces. We first introduce best-response values and define positive semidefinite slack operators whose kernels characterize the directions that are optimal against a given opponent strategy. Using complementary slackness, we show that equilibrium strategies are supported entirely on these tight subspaces. We then lift this local characterization to a global one by aggregating across all equilibria, which leads to a classification of directions into essential, neutral, and non-essential. If every Nash equilibrium is strictly complementary, for example under the rank-nondegeneracy condition introduced below, this global trichotomy collapses to a sharp dichotomy between directions that can arise at equilibrium and those that are strictly suboptimal. We conclude by comparing this semidefinite picture with the classical matrix-game setting, where strict complementarity is guaranteed to exist for at least one equilibrium. This yields a selected-equilibrium essential/unsupported separation, although it does not in general imply equality of the global supported and tight sets defined by unions over all equilibria.

\paragraph{Best-response values and slack operators}
We begin by introducing the best-response values and their associated slack operators, which identify the directions that are tight against a given opponent strategy.

\begin{definition}[Best-response values and slack operators]
\label{def:slack-operators}
With Alice minimizing and Bob maximizing, for $\beta\in\bspace$ and $\alpha\in\aspace$, define the best-response values and slack operators:
\begin{align*}
    w(\beta) &:= \min_{\alpha\in\aspace}\inner{\alpha,\Xi(\beta)} = \lambda_{\min}(\Xi(\beta)), & W(\beta) &:= \Xi(\beta)-w(\beta)\,\eye_{\aspace}\in\hmat_+^{n}, \\
    v(\alpha) &:= \max_{\beta\in\bspace}\inner{\Xi^\dagger(\alpha),\beta} = \lambda_{\max}(\Xi^\dagger(\alpha)), & V(\alpha) &:= v(\alpha)\,\eye_{\bspace}-\Xi^\dagger(\alpha)\in\hmat_+^{m}.
\end{align*}
\end{definition}

The operator $W(\beta)\succeq 0$ measures Alice's slack against $\beta$. Its kernel is the eigenspace of $\Xi(\beta)$ corresponding to the bottom eigenvalue $w(\beta)$, representing the set of Alice directions that achieve the best-response value against $\beta$. We refer to $\ker(W(\beta))$ as Alice's \emph{tight subspace} against $\beta$. Bob's slack against $\alpha$ and his tight subspace $\ker(V(\alpha))$ are defined symmetrically, using the top eigenvalue of $\Xi^\dagger(\alpha)$.

\paragraph{Complementary slackness identities}
We connect these tight subspaces to Nash equilibria via semidefinite complementary slackness, specializing a result from \citet[Theorem~2.1]{ickstadt2025nashequilibriasemidefinitegames}.

\begin{lemma}[Complementary slackness characterization of equilibrium]
\label{lem:complementary-slackness}
A pair $(X^*,Y^*)\in\aspace\times\bspace$ is a Nash equilibrium with value $\gameval$ if and only if
\[
    W(Y^*)=\Xi(Y^*)-\gameval\,\eye_{\aspace}\succeq 0,
    \qquad
    V(X^*)=\gameval\,\eye_{\bspace}-\Xi^\dagger(X^*)\succeq 0,
\]
and
\begin{align}
\label{eqn:complementary slackness identities}
    \inner{X^*,W(Y^*)}=0,
    \qquad
    \inner{Y^*,V(X^*)}=0.
\end{align}
\end{lemma}

The PSD conditions ensure no pure direction outperforms the equilibrium value; e.g., $\bra{u}W(Y^*)\ket{u} = \bra{u}\Xi(Y^*)\ket{u} - \gameval \ge 0$. Equality holds exactly on the tight subspaces. Consequently, \eqref{eqn:complementary slackness identities} implies an equilibrium strategy places mass only on directions tight against the opponent:

\begin{corollary}
\label{cor:support-contained-in-tight}
Let $(X^*,Y^*)$ be a Nash equilibrium. Then $\range(X^*)\subseteq \ker(W(Y^*))$ and $\range(Y^*)\subseteq \ker(V(X^*))$.
\end{corollary}
\begin{proof}
We prove the first inclusion and the second is symmetric. Since $X^*\succeq 0$ and $W(Y^*)\succeq 0$, the identity $\inner{X^*,W(Y^*)}=\operatorname{Tr}(X^*W(Y^*))=0$
implies $X^{*1/2}W(Y^*)X^{*1/2}=0.$
Hence \(W(Y^*)X^*=0\), and therefore, $\range(X^*)\subseteq \ker(W(Y^*))$.
\end{proof}

To pass from local geometry to a global classification across the full equilibrium set $\NE$, we define the global supported ($\mathcal{S}$) and tight ($\mathcal{T}$) sets:
\begin{align*}
    \mathcal S_A &:= \bigcup_{(X^*,Y^*)\in\NE} \range(X^*), & \mathcal T_A &:= \bigcup_{(X^*,Y^*)\in\NE} \ker(W(Y^*)), \\
    \mathcal S_B &:= \bigcup_{(X^*,Y^*)\in\NE} \range(Y^*), & \mathcal T_B &:= \bigcup_{(X^*,Y^*)\in\NE} \ker(V(X^*)).
\end{align*}
By Corollary~\ref{cor:support-contained-in-tight}, $\mathcal S_A \subseteq \mathcal T_A$ and $\mathcal S_B \subseteq \mathcal T_B$. The potential gap between these sets motivates the following trichotomy.

\begin{definition}[Essential, neutral, and non-essential directions]
\label{def:three-direction-types}
A unit direction $u \in \C^{2^n}$ is \emph{essential} for Alice if $u\in \mathcal S_A$, \emph{neutral} if $u\in \mathcal T_A\setminus \mathcal S_A$, and \emph{non-essential} if $u\notin \mathcal T_A$. Bob's directions $v \in \C^{2^m}$ are classified symmetrically using $\mathcal S_B$ and $\mathcal T_B$.
\end{definition}

\paragraph{Strict complementarity and rank nondegeneracy}
We now introduce structural conditions that eliminate the neutral class by upgrading the inclusions in Corollary~\ref{cor:support-contained-in-tight} to equalities.

\begin{condition}[Strict complementarity]
\label{cond:strict-complementarity}
A Nash equilibrium $(X^*,Y^*)$ is \emph{strictly complementary} if 
\[
    \rank(X^*)+\rank(W(Y^*))=2^n, \qquad \rank(Y^*)+\rank(V(X^*))=2^m.
\]
\end{condition}

\begin{lemma}[Pointwise equality under strict complementarity]
\label{lem:SC-gives-equality}
If $(X^*,Y^*)$ is strictly complementary, then $\range(X^*)=\ker(W(Y^*))$ and $\range(Y^*)=\ker(V(X^*))$.
\end{lemma}

\begin{proof}
By Corollary~\ref{cor:support-contained-in-tight}, $\range(X^*)\subseteq \ker(W(Y^*))$. Strict complementarity and rank-nullity dictate $\dim(\ker(W(Y^*))) = 2^n - \rank(W(Y^*)) = \rank(X^*)$. Since dimensions match, the inclusion is an equality. The argument for Bob is symmetric.
\end{proof}

We also introduce a rank-based regularity condition, adapted from the nondegeneracy framework of \citet{ickstadt2025nashequilibriasemidefinitegames}, that is sufficient for strict complementarity in our setting.

\begin{definition}[Rank nondegeneracy]
\label{def:nondegeneracy}
The game is \emph{rank-nondegenerate} if for every $X\in\aspace$ and $Y\in\bspace$,
\[
\dim(\ker(W(Y))) \le \rank(Y), 
\qquad 
\dim(\ker(V(X))) \le \rank(X).
\]
\end{definition}

\begin{lemma}[Rank nondegeneracy implies strict complementarity]
\label{lem:nondegenerate-implies-SC}
If the game is rank-nondegenerate, then every Nash equilibrium is strictly complementary.
\end{lemma}

\begin{proof}
For any $(X^*,Y^*)\in\NE$, Corollary~\ref{cor:support-contained-in-tight} and rank nondegeneracy imply
\[
\rank(X^*) 
\le \dim(\ker(W(Y^*))) 
\le \rank(Y^*).
\]
By symmetry,
\[
\rank(Y^*) 
\le \dim(\ker(V(X^*))) 
\le \rank(X^*).
\]
Hence equality holds throughout. In particular,
\[
\dim(\ker(W(Y^*)))=\rank(X^*),
\qquad
\dim(\ker(V(X^*)))=\rank(Y^*).
\]
By rank-nullity,
\[
\rank(X^*)+\rank(W(Y^*))=2^n,
\qquad
\rank(Y^*)+\rank(V(X^*))=2^m,
\]
so $(X^*,Y^*)$ is strictly complementary.
\end{proof}

\paragraph{Essential / non-essential separation}
If strict complementarity holds at every Nash equilibrium, then the pointwise equalities from Lemma~\ref{lem:SC-gives-equality} imply
$\mathcal S_A=\mathcal T_A$ and $\mathcal S_B=\mathcal T_B$, collapsing the global trichotomy into a dichotomy characterized by equilibrium payoffs.

\begin{theorem}
\label{thm:SC-collapse}
If every Nash equilibrium is strictly complementary (e.g., if the game is rank-nondegenerate), then neutral directions do not occur ($\mathcal S_A=\mathcal T_A, \mathcal S_B=\mathcal T_B$). Every direction is either:
\begin{enumerate}
    \item \emph{Essential}: $\bra{u}\Xi(Y^*)\ket{u}=\gameval$ for some equilibrium $(X^*,Y^*)$; or
    \item \emph{Non-essential}: $\bra{u}\Xi(Y^*)\ket{u}>\gameval$ for all equilibria.
\end{enumerate}
The symmetric statement holds for Bob using $\Xi^\dagger(X^*)$ and $\gameval$, with the strict inequality reversed.
\end{theorem}

\begin{proof}
Global equality $\mathcal{S}=\mathcal{T}$ follows from taking unions over Lemma~\ref{lem:SC-gives-equality}. For the payoffs, $u \in \mathcal{S}_A = \mathcal{T}_A \iff \exists (X^*,Y^*) \text{ s.t. } u \in \ker(W(Y^*))$, giving $\bra{u}(\Xi(Y^*) - \gameval I)\ket{u} = 0$. Conversely, $u \notin \mathcal{T}_A \iff \bra{u}W(Y^*)\ket{u} > 0$ for all equilibria, yielding $\bra{u}\Xi(Y^*)\ket{u} > \gameval$.
\end{proof}

\begin{corollary}[Unique equilibrium]
\label{cor:unique-equiv-neutral-SC}
Under a unique Nash equilibrium $(X^*,Y^*)$, strict complementarity is equivalent to the absence of neutral directions.
\end{corollary}

\begin{proof}
If no neutral directions exist, $\range(X^*) = \ker(W(Y^*))$ and $\range(Y^*) = \ker(V(X^*))$. The claim follows from rank-nullity. The converse follows from Theorem~\ref{thm:SC-collapse}.
\end{proof}

\begin{theorem}
\label{thm:master-equilibrium}
There exists a equilibrium $(X^\dagger,Y^\dagger)\in\NE$ such that
\begin{align*}
     \range(X^\dagger) = \mathcal{S}_A,  \qquad
     \range(Y^\dagger) = \mathcal{S}_B.
\end{align*}
Every non-essential direction is strictly suboptimal against this equilibrium:
\[
    \bra{u}\Xi(Y^\dagger)\ket{u}>\gameval
    \quad\text{for all non-essential }u,
\]
and symmetrically for Bob, with the strict inequality reversed.
\end{theorem}
\begin{proof}
Recall that
\[
\mathcal S_A
=
\bigcup_{(X^*,Y^*)\in\NE}\range(X^*),
\qquad
\mathcal S_B
=
\bigcup_{(X^*,Y^*)\in\NE}\range(Y^*).
\]
Because $\NE$ is convex and
\[
\range(X_1+X_2)=\range(X_1)+\range(X_2)
\]
for any $X_1,X_2\succeq 0$, both $\mathcal S_A$ and $\mathcal S_B$ are subspaces. Since the ambient spaces are finite-dimensional, we may choose a finite collection of equilibria whose $X$-ranges span $\mathcal S_A$, and another finite collection whose $Y$-ranges span $\mathcal S_B$. Taking the union of these two finite collections, write the resulting equilibria as
\[
\{(X^{(i)},Y^{(i)})\}_{i=1}^N\subseteq\NE,
\]
so that
\[
\mathcal S_A=\sum_{i=1}^N\range(X^{(i)}),
\qquad
\mathcal S_B=\sum_{i=1}^N\range(Y^{(i)}).
\]
Let
\[
(X^\dagger,Y^\dagger)
=
\frac1N\sum_{i=1}^N(X^{(i)},Y^{(i)}).
\]
By convexity of the Nash equilibrium set,
$(X^\dagger,Y^\dagger)\in\NE$. Furthermore, for any finite sum of positive semidefinite matrices with positive weights, the range of the sum is the sum of the individual ranges. Hence
\[
\range(X^\dagger)=\mathcal S_A,
\qquad
\range(Y^\dagger)=\mathcal S_B.
\]
    
    Finally, if $u$ is non-essential, then by Definition~\ref{def:three-direction-types}, $u \notin \mathcal{T}_A$. By the definition of the global tight set, $u$ cannot be in the kernel of the slack operator against any equilibrium opponent strategy, in particular $W(Y^\dagger)$. Since $W(Y^\dagger) \succeq 0$, it follows that $\bra{u}W(Y^\dagger)\ket{u} > 0$, which yields the strict suboptimality $\bra{u}\Xi(Y^\dagger)\ket{u} > \gameval$.
\end{proof}

\paragraph{Comparison with classical matrix games}

In this part, we contrast the classical linear programming (LP) setting with our quantum SDP framework to highlight the role of strict complementarity. The classical setting is naturally embedded within our framework: if the strategies $X,Y$ and the payoff observable are restricted to be diagonal matrices, the SDP game reduces to a classical matrix game.

A key distinction is the existence of a strictly complementary equilibrium. In general semidefinite games, strict complementarity may fail. By contrast, for every finite classical zero-sum game, the Goldman--Tucker theorem guarantees the existence of at least one strictly complementary equilibrium. Fixing such an equilibrium yields a canonical essential/unsupported separation: every pure strategy that occurs in the support of some equilibrium is supported by the selected strictly complementary equilibrium, whereas every strategy that never occurs in any equilibrium is strictly suboptimal against the selected opponent equilibrium.

This selected-equilibrium separation should be distinguished from the global tightness classification introduced above. In particular, existence of one strictly complementary equilibrium does not imply $\mathcal S_A=\mathcal T_A$ or $\mathcal S_B=\mathcal T_B$, because a strategy that is unsupported at every equilibrium may still be payoff-tight against some other degenerate equilibrium.

\begin{definition}[Classical slacks and complementary slackness]
Let $A \in \mathbb{R}^{m\times n}$ be the payoff matrix of a zero-sum game with value $v$. For any Nash equilibrium $(x^*,y^*) \in \Delta_m \times \Delta_n$, the primal-dual LP formulation defines the row slacks $r \in \mathbb{R}^m_+$ and column slacks $c \in \mathbb{R}^n_+$ as
$$r := Ay^* - v\mathbf{1} \ge 0, \qquad c := v\mathbf{1} - A^\top x^* \ge 0.$$
Complementary slackness dictates that $x_i^* r_i = 0$ and $y_j^* c_j = 0$ for all $i, j$.
\end{definition}

In the classical setting, the tight subspaces of Definition~\ref{def:slack-operators} reduce to the zero coordinates of the slack vectors $r$ and $c$. For the selected-equilibrium comparison below, we use the support-based classification: a pure strategy is \emph{essential} if it belongs to the support of at least one Nash equilibrium, and \emph{unsupported} otherwise. In the global trichotomy of Definition~\ref{def:three-direction-types}, an unsupported pure strategy is either neutral or non-essential, depending on whether it is tight against some equilibrium opponent strategy.

\begin{theorem}[Goldman-Tucker strict complementarity]
\label{thm:classical-sc-existence}
Every finite zero-sum matrix game admits at least one \emph{strictly complementary} Nash equilibrium $(x^*,y^*)$, satisfying
$$x_i^* + r_i > 0 \quad \text{for all } i, \qquad y_j^* + c_j > 0 \quad \text{for all } j.$$
\end{theorem}
Because Goldman--Tucker guarantees the existence of a strictly complementary equilibrium, fix such an equilibrium $(x^*,y^*)$. For this particular equilibrium pair, support and tightness coincide pointwise:
\[
x_i^*>0 \iff r_i=0,
\qquad
y_j^*>0 \iff c_j=0.
\]
Moreover, as shown below, this selected equilibrium contains every essential pure strategy in its support and strictly separates every unsupported pure strategy. This is a statement relative to the selected strictly complementary equilibrium; it does not assert that the global tight set obtained by taking a union over all equilibria coincides with the global supported set.

\begin{theorem}[Classical essential / unsupported separation]
\label{thm:essential-separation-classical}

Let $(x^*,y^*)$ be a strictly complementary equilibrium of a classical zero-sum game. Then this equilibrium separates pure strategies according to whether they can occur at equilibrium:

\begin{enumerate}
    \item \textbf{Essential strategies:} If row $i$ is essential, then
    $x_i^*>0,$
    and hence it is tight against $y^*$:
    $(Ay^*)_i=v.$

    \item \textbf{Unsupported strategies:} If row $i$ is not contained in the support of any Nash equilibrium, then
    $x_i^*=0$,
    and it is strictly suboptimal against $y^*$:
    $(Ay^*)_i>v.$
\end{enumerate}

The symmetric statement holds for the column player's strategies against $x^*$, with the strict inequality reversed.
\end{theorem}

\begin{proof}
If row $i$ is essential, there exists an equilibrium strategy $\bar{x}$ with $\bar{x}_i>0$. Since $y^*$ is also an equilibrium strategy, the pair $(\bar{x},y^*)$ is an equilibrium of the zero-sum game, and hence
$\bar{x}^{\top}Ay^*=v.$
Because
\[
(Ay^*)_k\ge v
\qquad\text{for every }k,
\]
every row in the support of $\bar{x}$ must be tight. In particular,
$(Ay^*)_i=v,$
so $r_i=0$. Since $(x^*,y^*)$ is strictly complementary,
\[
r_i=0\implies x_i^*>0.
\]

Conversely, if row $i$ is unsupported, then it is not in the support of any equilibrium, so in particular $x_i^*=0$. Strict complementarity then implies $r_i>0$, and therefore
$(Ay^*)_i>v.$
The column-player argument is symmetric.
\end{proof}


\section{Omitted Proofs for Convergence Analysis}
\label{appendix:OmittedConvergence}
This appendix collects the proofs of our convergence guarantees. In Section \ref{subsec:q-smoothing}, we bound the iteration complexity of the inner \qsmoothing procedure of the \itersmooth algorithm \ref{alg:iterative smoothing}. Section \ref{appendix:spms} defines Saddle-Point Metric Subregularity (SP-MS) for spectraplexes and proves that SP-MS holds for every finite-dimensional quantum zero-sum game with a finite, game-dependent exponent. If the game admits a Nash equilibrium whose two states are positive definite, the proof gives a linear gap-to-distance error bound and SP-MS with exponent zero. Under this full-rank hypothesis, in Section {\ref{appendix:itersmooth-convergence}}, we use {the inner-loop complexity bound and the linear error bound} to prove Theorem \ref{thm:q-iterated-convergence-rate-fin}. Section \ref{appendix_subsec:ogda-convergence} then establishes the linear last-iterate convergence rate of \ogda under SP-MS with $\spmsparam=0$. Finally, Section~\ref{subsec:ommwu-convergence} proves relative-entropy contraction toward the QRE for the separately defined regularized entropic method and bounds its approximation error relative to the Nash equilibrium under the full-rank and uniqueness assumptions.

\subsection{Convergence Rates of \qsmoothing}
\label{subsec:q-smoothing}
\subsubsection{Road-map for Main Results}
\label{subsec:q-smoothing-statements}
This section transfers the progress guarantee for the smoothed gap $\gap_\smooth(\joint)$ to a guarantee for the original duality gap $\gap(\joint)$.
By the smoothing construction in Section~\ref{subsec:smoothing-prelim}, the smoothed gap defined in~\eqref{eqn:smoothed-gap} uniformly approximates the original gap:
\begin{align} \label{ineqn:bound-gap-by-smoothed-gap}
    \gap_\smooth(\joint)
    \;\le\;
    \gap(\joint)
    \;\le\;
    \gap_\smooth(\joint) + \smooth \smoothradius,
    \qquad \forall \joint \in \jointsp,
\end{align}
where $\smoothradius := \max_{\jointalt \in \jointsp} \reg(\jointalt)$.
Choosing $\smooth=\varepsilon/(2\smoothradius)$ gives $\smooth\smoothradius=\varepsilon/2$.
Hence it suffices to ensure $\gap_\smooth(\joint_k)\le \varepsilon/2$, since then $\gap(\joint_k)\le \varepsilon$.

It remains to bound the number of \qsmoothing iterations needed to drive the smoothed gap below $\varepsilon/2$.
The \qsmoothing\ analysis yields an explicit rate for $\gap_\smooth(\joint_k)-\gap_\smooth(\jointstar)$ with $\jointstar\in\Opt$. Combining this rate with $\smooth=\varepsilon/(2\smoothradius)$ gives the iteration complexity stated below.
\qsmoothingprop*
The proof of Proposition~\ref{prop:qsmoothing}, together with the auxiliary lemmas needed for the analysis, is deferred to Appendix~\ref{appendix:prop:qsmoothing}.

\subsubsection{Proofs of Auxiliary Lemmas and Proposition \ref{prop:qsmoothing}}
\label{appendix:prop:qsmoothing}
To analyze the number of iterations of \qsmoothing, we adapt the technique of \citet{Lan2011PrimalDual}, which bounds the original duality gap $\gap(\joint_k)$ using the smoothed gap $\gap_\smooth(\joint_k)$ introduced in Section~\ref{subsec:smoothing-prelim}.
Recall the definition in~\eqref{eqn:smoothed-gap}:
\begin{align*}
    \gap_\smooth(\joint) = \max_{\jointalt \in \jointsp} \braces*{\tr{\a \superop(\balt)} - \tr{\aalt \superop(\b)} - \smooth \reg(\jointalt)}.
\end{align*}
The next lemma obtained by applying \citet[Theorem 1]{Lan2011PrimalDual} gives a bound on the convergence of \qsmoothing.
To make this application explicit, fix an outer stage \(i\) and identify the three sequences in Lan--Lu--Monteiro with those in Algorithm~\ref{alg:iterative smoothing} as follows:
\begin{align*}
&
u_k^{\mathrm{sd}}\leftrightarrow\joint_i^{(k)},
\qquad
u_k^{\mathrm{ag}}\leftrightarrow\Phi_i^{(k)},
\qquad
u_k\leftrightarrow\jointbar_i^{(k)}, \qquad
f\leftrightarrow\gap_{\smooth_{i+1}},
\qquad
L\leftrightarrow L_{i+1}.
\end{align*}
For the accelerated prox geometry, take
\begin{align*}
&
h_U(\joint)=\frac12\fnorm{\joint}^2,
\qquad
\sigma_U=1,
\qquad
d_U(\joint;\joint_i)=\frac12\fnorm{\joint-\joint_i}^2.
\end{align*}
The acceleration weights satisfy
\begin{align*}
&
\alpha_t=\frac{t+1}{2},
\qquad
A_k=\sum_{t=0}^k\alpha_t=\frac{(k+1)(k+2)}{4},
\qquad
\frac{\alpha_k}{A_k}=\frac{2}{k+2}.
\end{align*}
Thus the convex-combination weights, the aggregate-gradient weights, and the Euclidean prox objective agree exactly with the cited method. The function-value terms in its affine linearizations are constant with respect to the decision variable and may therefore be omitted from both \(\arg\min\) subproblems.
\begin{lemma} \label{lemma:smoothed-gap-bound}
    Fix an outer stage \(i\). Suppose that the sequence $\braces{\joint_i^{(k)}}$ is generated by the \qsmoothing algorithm initialized at {\(\joint_i\)}, with target accuracy {\(\varepsilon_{i+1}\)} and smoothing parameter $\smooth_{i+1} = \frac{\varepsilon_{i+1}}{2 \smoothradius}$, and let $\jointstar \in \Opt$. Then, for all $k \geq 1$,
    \begin{align} \label{eqn:gap-between-kth-interation-and-opt}
        \gap_{\smooth_{i+1}} (\joint_i^{(k)}) - \gap_{\smooth_{i+1}}(\jointstar) \leq \frac{4 L_{i+1}\, d_U(\jointstar;\joint_i) }{\sigma_U k (k+1)} = \frac{2L_{i+1}\fnorm{\joint_i-\jointstar}^2}{k(k+1)}.
    \end{align}
\end{lemma}
\begin{proof}
Fix an outer stage \(i\) and ignore the stopping test, which only truncates the accelerated sequence. Under the identification above, the resulting sequence coincides with Nesterov's sequence in \citet[Theorem~1]{Lan2011PrimalDual}. Applying that theorem with the Euclidean prox gives the first inequality in \eqref{eqn:gap-between-kth-interation-and-opt}. The equality follows from \(\sigma_U=1\) and \(d_U(\jointstar;\joint_i)=\frac12\fnorm{\joint_i-\jointstar}^2\).
\end{proof}
\noindent Given Lemma \ref{lemma:smoothed-gap-bound} bounding $\gap_{\smooth_{i+1}}(\joint_i^{(k)})$, we can now derive a bound on the original gap $\gap(\joint_i^{(k)})$.
\begin{lemma} 
    \label{lemma:gap-bound}
   Under the same assumptions as lemma \ref{lemma:smoothed-gap-bound}, for all $k \geq 1$,
    \begin{align*}
        \gap(\joint_i^{(k)}) \leq \smooth_{i+1} \smoothradius + \frac{4 L_{i+1}\, d_U(\jointstar;\joint_i) }{\sigma_U k^2}.
    \end{align*}    
\end{lemma}
\begin{proof}
    By the min–max theorem we have $\min_{\joint\in \jointsp} \gap(\joint) = 0$. Observe that
    \begin{align}
        &\gap(\joint_i^{(k)}) - \min_{\joint\in \jointsp} \gap(\joint)
        = \underbrace{\gap(\joint_i^{(k)}) - \gap_{\smooth_{i+1}}(\joint_i^{(k)})}_{(I)} \notag\\
        &\quad + \underbrace{\gap_{\smooth_{i+1}}(\joint_i^{(k)}) - \gap_{\smooth_{i+1}}(\jointstar)}_{\text{Lemma } \ref{lemma:smoothed-gap-bound}}
        + \underbrace{\gap_{\smooth_{i+1}}(\jointstar) - \min_{\joint\in \jointsp} \gap(\joint)}_{(II)}.
    \end{align}
    From Inequality \eqref{ineqn:bound-gap-by-smoothed-gap}, Equation $(I) \leq \smooth_{i+1} \smoothradius$ and Equation $(II) \leq 0$. Applying Lemma \ref{lemma:smoothed-gap-bound} to the middle term yields
        \begin{align*}
            \gap(\joint_i^{(k)}) \leq \smooth_{i+1} \smoothradius + \frac{4 L_{i+1}\, d_U(\jointstar;\joint_i) }{\sigma_U k^2}.
        \end{align*}
\end{proof}
We are interested in the gap between $\gap(\joint_i^{(k)})$ and $\min_{\joint\in \jointsp} \gap(\joint)$. We want to investigate how many first-order iterations the algorithm requires to reach $ \gap(\joint_i^{(k)}) -  \min_{\joint\in \jointsp} \gap(\joint) \mathrel{\leq} \varepsilon_{i+1}$. We know $\min_{\joint\in \jointsp} \gap(\joint) = 0$ by the minimax theorem. Observe that
\begin{align}
    &\gap(\joint_i^{(k)})
    = \underbrace{\gap(\joint_i^{(k)}) - \gap_{\smooth_{i+1}}(\joint_i^{(k)})}_{(I)} \notag\\
    &\quad + \underbrace{\gap_{\smooth_{i+1}}(\joint_i^{(k)}) - \gap_{\smooth_{i+1}}(\jointstar)}_{\text{Inequality } \eqref{eqn:gap-between-kth-interation-and-opt}}
    + \underbrace{\gap_{\smooth_{i+1}}(\jointstar) - \min_{\joint\in \jointsp} \gap(\joint)}_{(II)}.
\end{align}
According to Inequality \eqref{ineqn:bound-gap-by-smoothed-gap}, Equation $(I) \leq \smooth_{i+1} \smoothradius$ and Equation $(II) \leq 0$. Therefore,
\begin{align*}
    \gap(\joint_i^{(k)}) \leq \smooth_{i+1} \smoothradius + \frac{4 L_{i+1}\, d_U(\jointstar;\joint_i) }{\sigma_U k^2}, \quad \forall k \geq 1.
\end{align*}
With the bound on the duality gap after $k$ iterations in hand, we can derive the iteration complexity of \qsmoothing.
We restate Proposition~\ref{prop:qsmoothing} for convenience below.
\qsmoothingprop*
\begin{proof}
    Substituting {\(\sigma_U = 1\)}, {\(L_{i+1} = \opnorm{\gradient}^2/\smooth_{i+1}\)}, and {\(\smooth_{i+1}=\varepsilon_{i+1}/(2\smoothradius)\)} into Lemma~\ref{lemma:gap-bound} gives
    \begin{align} \label{ineqn:bound-kth-gap}
        \gap(\joint_i^{(k)}) \leq \frac{\varepsilon_{i+1}}{2} + \frac{4 \cdot \opnorm{\gradient}^2 \cdot \smoothradius \cdot \fnorm{\joint_i - \jointstar}^2}{\varepsilon_{i+1} \cdot k^2} .
    \end{align}
    Choose $\jointstar \in \argmin_{\jointalt \in \Opt} \braces{\fnorm{\joint_i - \jointalt}}$, so that $\fnorm{\joint_i - \jointstar} = \dist(\joint_i,\cond{\jointsp})$. The right-hand side of \eqref{ineqn:bound-kth-gap} is at most {\(\varepsilon_{i+1}\)} whenever
    \begin{align*}
        k
        \geq
        \frac{2\sqrt{2}\cdot\opnorm{\gradient}\cdot\sqrt{\smoothradius}\cdot\dist(\joint_i,\cond{\jointsp})}{\varepsilon_{i+1}}.
    \end{align*}
    Taking the ceiling yields the claimed integer iteration bound.
\end{proof}
\subsection{Saddle-Point Metric Subregularity for Spectraplexes (SP-MS)}
\label{appendix:spms}

\subsubsection{Road-map for Main Result}
\label{appendix:spms-roadmap}

We first define SP-MS and state the error bounds that we will prove.
Throughout this subsection, write $\aspaceopt=\Opt_\a$,
$\bspaceopt=\Opt_\b$, and
$\cond{\jointsp}=\aspaceopt\times\bspaceopt=\Opt$.
Let $d_A$ and $d_B$ denote the matrix dimensions of the two players.
We use the product Frobenius norm, so that
$\fnorm{\joint}^2=\fnorm{\astate}^2+\fnorm{\bstate}^2$
for $\joint=(\astate,\bstate)$.
The adjoint $\superop^\dagger$ is taken with respect to the trace
inner product. Thus
$\tr{\astate\superop(\bstate)}
=\tr{\superop^\dagger(\astate)\bstate}$, and
$\gradpayoff(\joint)
=(\superop(\bstate),-\superop^\dagger(\astate))$.
Trace pairings of joint states are understood componentwise.

\begin{restatable}{definition}{spmsdefspectraplexes}
(SP-MS for spectraplexes)
\label{def:spms}
For $\joint\in\jointsp$, let
$\sol{\joint}=\projz(\joint)$ be its Frobenius projection onto
$\cond{\jointsp}$, and write
$\dist(\joint)=\fnorm{\joint-\sol{\joint}}$.
The game satisfies SP-MS with parameter $\spmsparam\geq0$ if there
exists a constant $C>0$ such that
\begin{equation}
R(\joint):=\sup_{\substack{\jointalt\in\jointsp\\\jointalt\neq\joint}}
\frac{\tr{\gradpayoff(\joint)(\joint-\jointalt)}}
     {\fnorm{\joint-\jointalt}}
\geq C\dist(\joint)^{1+\spmsparam}
\qquad
\text{for every }\joint\in\jointsp\setminus\cond{\jointsp}.
\label{eqn:spms-corrected-definition}
\end{equation}
The constant may depend on the game, but not on $\joint$ or on the
desired accuracy.
\end{restatable}

The conclusion with $\spmsparam=0$ is proved under the existence of
a Nash equilibrium whose two states are positive definite. Without
this assumption, we obtain an error bound with a finite H\"older
exponent.

\begin{restatable}{theorem}{thmspmsforspectraplexappendix}
\label{thm:sp-ms-for-spectraplex-appendix}
Consider a finite-dimensional quantum zero-sum game with bilinear
payoff $\payoff(\astate,\bstate)
=\tr{\astate\superop(\bstate)}$. The following properties hold.
\begin{enumerate}
\item There exist an integer $q\geq1$ and a constant $K>0$ such that
\begin{equation}
\dist(\joint)^q\leq K\gap(\joint)
\qquad\text{for every }\joint\in\jointsp.
\label{eqn:spms-holder-error-bound}
\end{equation}
Consequently, the game satisfies SP-MS with
$\spmsparam=q-1$ and $C=1/(2K)$.

\item If there exists
$(\bar{\astate},\bar{\bstate})\in\cond{\jointsp}$ with
$\bar{\astate}\succ0$ and $\bar{\bstate}\succ0$, then one may
take $q=1$. In particular, for some $K>0$,
\begin{equation}
\dist(\joint)\leq K\gap(\joint)
\qquad\text{for every }\joint\in\jointsp,
\label{eqn:spms-linear-error-bound}
\end{equation}
and the game satisfies SP-MS with $\spmsparam=0$.
\end{enumerate}
In both cases, the exponent and constants are independent of
$\joint$ and of the desired accuracy.
\end{restatable}

As in the simplex analysis of \citet[Theorem~5]{wei2021linear},
we obtain SP-MS by comparing the duality gap with the distance to
the equilibrium set. In the spectraplex setting, the proof of a
linear error bound uses the positive definite equilibrium directly.
The unconditional argument instead uses the semialgebraic structure
of the strategy spaces and the gap. Our proof consists of five claims.

\medskip
\noindent\textbf{Goal reduction.}
For convenience, denote the left-hand side of
\eqref{eqn:spms-corrected-definition} by $R(\joint)$.
Define the two nonnegative one-sided gaps by
\begin{align}
g_A(\astate)
&=\max_{\balt\in\bspace}\tr{\astate\superop(\balt)}-\gameval
 =\lambda_{\max}(\superop^\dagger(\astate))-\gameval,
\label{eqn:spms-gap-alpha}\\
g_B(\bstate)
&=\gameval-\min_{\aalt\in\aspace}\tr{\aalt\superop(\bstate)}
 =\gameval-\lambda_{\min}(\superop(\bstate)).
\label{eqn:spms-gap-beta}
\end{align}
Then $\gap(\joint)=g_A(\astate)+g_B(\bstate)$.
Bilinearity and the diameter bound
$\fnorm{\joint-\jointalt}\leq2$ give
$R(\joint)\geq\gap(\joint)/2$ outside the equilibrium set.
Thus an error bound of the form
$\dist(\joint)^q\leq K\gap(\joint)$ implies SP-MS with
$\spmsparam=q-1$. Under the full-rank assumption, we prove the
stronger one-sided bounds: for every $\astate\in\aspace$ and $\bstate\in\bspace$,
\begin{align}
g_A(\astate)
&\geq c\fnorm{\astate-\proj_{\aspaceopt}(\astate)},
\label{eqn:bound-deviation-alpha}\\
g_B(\bstate)
&\geq c\fnorm{\bstate-\proj_{\bspaceopt}(\bstate)}.
\label{eqn:bound-deviation-beta}
\end{align}
with a fixed constant $c>0$.

\paragraph{Claim~\ref{claim:compact}: Compactness of equilibria sets.}
The optimal strategy sets are nonempty, compact, and convex.
Moreover, the duality gap vanishes exactly on their Cartesian
product. This ensures that projection onto the equilibrium set is
well-defined and that the gap and the distance have the same zero set.

\paragraph{Claim~\ref{claim:spms-affine-optimality}: Affine equations for optimality.}
Fix a full-rank equilibrium
$(\bar{\astate},\bar{\bstate})$.
Complementarity and positive definiteness imply
$\superop(\bar{\bstate})=\gameval \eye_{\aspace}$ and
$\superop^\dagger(\bar{\astate})=\gameval \eye_{\bspace}$.
Pairing any optimal strategy with these fixed equilibrium states
then shows that the entire optimal strategy sets are described by
the affine equations
$\superop^\dagger(\astate)=\gameval \eye_{\bspace}$ and
$\superop(\bstate)=\gameval \eye_{\aspace}$, respectively, within the
spectraplexes. Only the fixed equilibrium needs to have full rank.

\paragraph{Claim~\ref{claim:spms-residual-gap}: The gaps control the equations' residuals.}
The residuals of these affine equations have zero trace pairing
with the fixed positive definite equilibrium states.
An elementary eigenvalue estimate therefore bounds their Frobenius
norms by constant multiples of $g_A(\astate)$ and $g_B(\bstate)$.
This step controls both positive and negative eigenvalues of the
residuals using the one-sided gaps.

\paragraph{Claim~\ref{claim:spms-affine-repair}: Distance to an affine slice of a spectraplex.}
First project a density matrix onto the affine equations, temporarily
omitting positive semidefiniteness. Linear algebra bounds this
projection error by the equation residual. The projected matrix may
have negative eigenvalues; mixing it with the fixed positive
definite feasible state restores positive semidefiniteness while
preserving the equations. The required mixing changes the matrix
by at most a constant multiple of the projection error.
Together with Claim~\ref{claim:spms-residual-gap}, this gives
\eqref{eqn:bound-deviation-alpha} and
\eqref{eqn:bound-deviation-beta}.

\paragraph{Claim~\ref{claim:spms-semialgebraic}: An unconditional H\"older error bound.}
A set is \emph{semialgebraic} if it is a finite Boolean combination
of polynomial equalities and inequalities with real coefficients;
a function is semialgebraic if its graph is semialgebraic.
In real coordinates, the joint spectraplex is compact and
semialgebraic, and both the gap and the distance to equilibrium are
continuous semialgebraic functions.
The semialgebraic \L{}ojasiewicz inequality then compares these
functions by a finite power because they have the same zero set. The error bound then comes out as a direct consequence \cite[Theorem 1.1 \& Theorem 2.9]{BasuMohammadNezhad2024}.
This proves \eqref{eqn:spms-holder-error-bound} without a full-rank
assumption, but does not in general allow us to choose $q=1$.

\paragraph{How the claims yield SP-MS.}
Under the full-rank assumption, the one-sided bounds and the product
structure of the equilibrium set give
$\gap(\joint)\geq c\dist(\joint)$, and hence
$R(\joint)\geq(c/2)\dist(\joint)$.
For an arbitrary game, Claim~\ref{claim:spms-semialgebraic} gives
$R(\joint)\geq(1/(2K))\dist(\joint)^q$.
These are the two conclusions of
Theorem~\ref{thm:sp-ms-for-spectraplex-appendix}.

\subsubsection{Proof of Theorem \ref{thm:sp-ms-for-spectraplex-appendix}}
\label{appendix:thm:sp-ms-for-spectraplex-appendix}

We now prove the five claims outlined in
Appendix~\ref{appendix:spms-roadmap}. For convenience, we restate
the definition and theorem.
\spmsdefspectraplexes*
\thmspmsforspectraplexappendix*

\begin{proof}
We first record the properties that hold without a rank assumption.

\begin{claim}\label{claim:compact}
The sets $\aspaceopt$ and $\bspaceopt$ are nonempty, compact, and
convex. The Nash equilibrium set is
$\cond{\jointsp}=\aspaceopt\times\bspaceopt$, and
$\gap(\joint)=0$ if and only if $\joint\in\cond{\jointsp}$.
\end{claim}

\begin{proof}
The minimax theorem applies because the two spectraplexes are
nonempty, compact, and convex, and the payoff is continuous and
bilinear. In particular, the value $\gameval$ is attained by both
players. The functions
$\astate\mapsto\lambda_{\max}(\superop^\dagger(\astate))$ and
$\bstate\mapsto-\lambda_{\min}(\superop(\bstate))$ are continuous
and convex: each is the maximum of a compact family of linear
functions. Consequently,
\begin{align*}
\aspaceopt
&=\{\astate\in\aspace:
       \lambda_{\max}(\superop^\dagger(\astate))\leq\gameval\},\\
\bspaceopt
&=\{\bstate\in\bspace:
       \lambda_{\min}(\superop(\bstate))\geq\gameval\}
\end{align*}
are nonempty closed convex subsets of compact sets, and hence are
compact and convex. Any pair of optimal strategies satisfies the
saddle-point inequalities and is a Nash equilibrium. Conversely,
the components of a Nash equilibrium are optimal strategies.
Finally, minimax gives $g_A\geq0$ and $g_B\geq0$, with
$g_A^{-1}(0)=\aspaceopt$ and $g_B^{-1}(0)=\bspaceopt$.
Since $\gap(\joint)=g_A(\astate)+g_B(\bstate)$, its zero set is
exactly $\cond{\jointsp}$.
\end{proof}

Claim~\ref{claim:compact} also justifies the unique Frobenius
projections used above. For density matrices $\rho$ and $\sigma$,
$\tr{\rho^2}\leq1$, $\tr{\sigma^2}\leq1$, and
$\tr{\rho\sigma}\geq0$. Therefore
\begin{equation}
\fnorm{\rho-\sigma}^2
=\tr{\rho^2}+\tr{\sigma^2}-2\tr{\rho\sigma}\leq2,
\label{eqn:spms-state-diameter}
\end{equation}
so the diameter of the joint strategy space is at most $2$.
The adjoint identity gives, for
$\jointalt=(\aalt,\balt)$,
\begin{align*}
\tr{\gradpayoff(\joint)(\joint-\jointalt)}
&=\tr{\superop(\bstate)(\astate-\aalt)}
  -\tr{\superop^\dagger(\astate)(\bstate-\balt)}\\
&=\tr{\astate\superop(\balt)}
  -\tr{\aalt\superop(\bstate)}.
\end{align*}
Maximizing over $\jointalt\in\jointsp$ yields $\gap(\joint)$.
If $\joint\notin\cond{\jointsp}$, this maximum is positive, and
compactness gives a maximizer $\widehat{\joint}\neq\joint$.
Evaluating the normalized supremum at this maximizer gives
\begin{equation}
R(\joint)
\geq\frac{\tr{\gradpayoff(\joint)(\joint-\widehat{\joint})}}
             {\fnorm{\joint-\widehat{\joint}}}
\geq\frac12\gap(\joint).
\label{ineqn:condition-measure-ineqn}
\end{equation}

\medskip
\noindent\textbf{The full-rank case.}
Fix an $(\bar{\astate},\bar{\bstate})\in\cond{\jointsp}$ with
$\bar{\astate}\succ0$ and $\bar{\bstate}\succ0$, and set
$a=\lambda_{\min}(\bar{\astate})>0$ and
$b=\lambda_{\min}(\bar{\bstate})>0$.
The next three claims use this assumption.

\begin{claim}\label{claim:spms-affine-optimality}
The fixed equilibrium satisfies
\begin{equation}
\superop(\bar{\bstate})=\gameval \eye_{\aspace},
\qquad
\superop^\dagger(\bar{\astate})=\gameval \eye_{\bspace}.
\label{eqn:spms-full-rank-indifference}
\end{equation}
Moreover, the optimal strategy sets are
\begin{align}
\aspaceopt
&=\{\astate\in\aspace:
          \superop^\dagger(\astate)=\gameval \eye_{\bspace}\},
\label{eqn:spms-affine-alpha}\\
\bspaceopt
&=\{\bstate\in\bspace:
          \superop(\bstate)=\gameval \eye_{\aspace}\}.
\label{eqn:spms-affine-beta}
\end{align}
\end{claim}

\begin{proof}
We repeatedly use the following elementary fact. If $W\succeq
\gamma I$ with $\gamma>0$, and $S\succeq0$, then
$\tr{WS}\geq\gamma\tr{S}$. Indeed,
$\tr{(W-\gamma I)S}
=\tr{S^{1/2}(W-\gamma I)S^{1/2}}\geq0$.
Thus $\tr{WS}=0$ implies $S=0$.

Since $\bar{\astate}$ minimizes against $\bar{\bstate}$ with
payoff $\gameval$, the matrix
$\superop(\bar{\bstate})-\gameval \eye_{\aspace}$ is positive semidefinite
and has zero trace pairing with $\bar{\astate}$.
The preceding fact proves the first identity in
\eqref{eqn:spms-full-rank-indifference}. Similarly,
$\gameval \eye_{\bspace}-\superop^\dagger(\bar{\astate})\succeq0$ has
zero trace pairing with $\bar{\bstate}$, proving the second identity.

Now take any $\aopt\in\aspaceopt$.
Its optimality gives
$S_A=\gameval \eye_{\bspace}-\superop^\dagger(\aopt)\succeq0$.
Using the adjoint identity and
\eqref{eqn:spms-full-rank-indifference},
\begin{align*}
\tr{\bar{\bstate}S_A}
&=\gameval-\tr{\bar{\bstate}\superop^\dagger(\aopt)}\\
&=\gameval-\tr{\aopt\superop(\bar{\bstate})}\\
 &=\gameval-\gameval\tr{\aopt}=0.
\end{align*}
Because $\bar{\bstate}\succ0$, this forces $S_A=0$.
Conversely, a feasible $\astate$ with
$\superop^\dagger(\astate)=\gameval \eye_{\bspace}$ has worst-case payoff
$\gameval$ and is optimal. This proves
\eqref{eqn:spms-affine-alpha}.

For any $\bstate^*\in\bspaceopt$, optimality gives
$S_B=\superop(\bstate^*)-\gameval \eye_{\aspace}\succeq0$, and
\begin{align*}
\tr{\bar{\astate}S_B}=\tr{\superop^\dagger(\bar{\astate})\bstate^*}-\gameval
=\gameval\tr{\bstate^*}-\gameval=0.
\end{align*}
Positive definiteness of $\bar{\astate}$ implies $S_B=0$.
The converse follows because
$\superop(\bstate)=\gameval \eye_{\aspace}$ gives minimum payoff
$\gameval$. This proves \eqref{eqn:spms-affine-beta}.
\end{proof}

\begin{claim}\label{claim:spms-residual-gap}
For every $\astate\in\aspace$ and $\bstate\in\bspace$,
\begin{align}
\fnorm{\superop^\dagger(\astate)-\gameval \eye_{\bspace}}
&\leq\frac{\sqrt{d_B}}{b}\,g_A(\astate),
\label{eqn:spms-residual-alpha}\\
\fnorm{\superop(\bstate)-\gameval \eye_{\aspace}}
&\leq\frac{\sqrt{d_A}}{a}\,g_B(\bstate).
\label{eqn:spms-residual-beta}
\end{align}
\end{claim}

\begin{proof}
First let $W$ be a $k\times k$ density matrix with
$W\succeq\gamma I$, where $\gamma>0$, and let $H$ be Hermitian
with $\tr{WH}=0$. Put $g=\lambda_{\max}(H)$.
Since $H\preceq gI$ and $\tr{W}=1$, we have
$0=\tr{WH}\leq g$, so $g\geq0$.
The matrix $M=gI-H$ is positive semidefinite, and
\begin{equation*}
g=\tr{WM}\geq\gamma\tr{M}
\geq\gamma\|M\|_{\mathrm{op}}.
\end{equation*}
Every eigenvalue of $H$ therefore lies in $[g-g/\gamma,g]$.
Also $\gamma\leq1$, because $W$ has trace one. Hence every
eigenvalue has absolute value at most $g/\gamma$, which implies
\begin{equation}
\fnorm{H}\leq\frac{\sqrt{k}}{\gamma}\lambda_{\max}(H).
\label{eqn:spms-weighted-spectral-bound}
\end{equation}

Apply this estimate to
$H_A=\superop^\dagger(\astate)-\gameval \eye_{\bspace}$ and
$H_B=\gameval \eye_{\aspace}-\superop(\bstate)$.
Equation~\eqref{eqn:spms-full-rank-indifference} gives the explicit
identities
\begin{align*}
\tr{\bar{\bstate}H_A}
&=\tr{\astate\superop(\bar{\bstate})}-\gameval
 =\gameval\tr{\astate}-\gameval=0,\\
\tr{\bar{\astate}H_B}
&=\gameval-\tr{\superop^\dagger(\bar{\astate})\bstate}
 =\gameval-\gameval\tr{\bstate}=0.
\end{align*}
Since $\lambda_{\max}(H_A)=g_A(\astate)$ and
$\lambda_{\max}(H_B)=g_B(\bstate)$,
\eqref{eqn:spms-weighted-spectral-bound}, with
$(W,\gamma)=(\bar{\bstate},b)$ and
$(W,\gamma)=(\bar{\astate},a)$, respectively, proves both bounds.
\end{proof}

\begin{claim}\label{claim:spms-affine-repair}
Let $T$ be a real-linear map from $d\times d$ Hermitian matrices
to $k\times k$ Hermitian matrices. Suppose that a density matrix
$\bar\rho\succ0$ satisfies $T(\bar\rho)=B$, and put
$\gamma=\lambda_{\min}(\bar\rho)>0$.
For
$S=\{\rho\succeq0:\tr{\rho}=1,\ T(\rho)=B\}$, there exists
$\kappa_T>0$, depending only on the restriction of $T$ to the
trace-zero subspace, such that every density matrix $\rho$ satisfies
\begin{equation}
\inf_{\sigma\in S}\fnorm{\rho-\sigma}
\leq\kappa_T\left(1+\frac{\sqrt2}{\gamma}\right)
       \fnorm{T(\rho)-B}.
\label{eqn:spms-affine-repair}
\end{equation}
\end{claim}

\begin{proof}
Let $V=\{H=H^\dagger:\tr{H}=0\}$ and
$N=\{H\in V:T(H)=0\}$. Denote the orthogonal complement of $N$ as $N^{\perp}$.
The affine set obtained by omitting positive semidefiniteness is
$E=\{Y=Y^\dagger:\tr{Y}=1,\ T(Y)=B\}=\bar\rho+N$.
Let $Y$ be the orthogonal projection of $\rho$ onto $E$, and set
$\delta=\fnorm{\rho-Y}$.
Then $\rho-Y\in V\cap N^\perp$ and
$T(\rho-Y)=T(\rho)-B$.

The restriction of $T$ to $V\cap N^\perp$ is injective.
If $V\cap N^\perp\neq\{0\}$, the unit sphere \[
S_T=\{H\in V\cap N^\perp:\|H\|_F=1\}
\] is nonempty and compact.
The continuous function $H\mapsto\fnorm{T(H)}$ therefore attains
its minimum at some $H_0\in V\cap N^\perp$ with $\fnorm{H_0}=1$.
If this minimum is zero, then $T(H_0)=0$. Since $H_0\in V$,
the definition of $N$ would imply $H_0\in N$. However, we have $H_0\in N^\perp$ as well, so $H_0$ would be orthogonal to itself, giving
$\fnorm{H_0}^2=\tr{H_0^2}=0$. This contradicts
$\fnorm{H_0}=1$. Thus the minimum is strictly positive. Let $\kappa_T$ be the reciprocal of this
minimum. If $V\cap N^\perp=\{0\}$, take $\kappa_T=1$;
then $\rho-Y=0$. In either case,
\begin{equation}
\delta\leq\kappa_T\fnorm{T(\rho)-B}.
\label{eqn:spms-affine-projection}
\end{equation}

Although $Y$ need not be positive semidefinite,
$\rho\succeq0$ and $\|Y-\rho\|_{\mathrm{op}}\leq\delta$
imply $Y\succeq-\delta I$.
Define $\widetilde\rho=(\gamma Y+\delta\bar\rho)/(\gamma+\delta)$.
Both $Y$ and $\bar\rho$ belong to $E$, so
$\widetilde\rho\in E$. Furthermore,
$\gamma Y+\delta\bar\rho
\succeq-\gamma\delta I+\delta\gamma I=0$.
Thus $\widetilde\rho\in S$.
The diameter estimate \eqref{eqn:spms-state-diameter} gives
\begin{align*}
\inf_{\sigma\in S}\fnorm{\rho-\sigma}
&\leq\fnorm{\rho-\widetilde\rho}\\
&\leq\frac{\gamma\fnorm{\rho-Y}
                 +\delta\fnorm{\rho-\bar\rho}}{\gamma+\delta}\\
&\leq\frac{\gamma\delta+\sqrt2\delta}{\gamma+\delta}
 \leq\left(1+\frac{\sqrt2}{\gamma}\right)\delta.
\end{align*}
Combining this with \eqref{eqn:spms-affine-projection} proves
\eqref{eqn:spms-affine-repair}.
\end{proof}

We can now combine Claims~\ref{claim:spms-affine-optimality}--\ref{claim:spms-affine-repair}.
Apply Claim~\ref{claim:spms-affine-repair} first with
$(T,B,\bar\rho)=(\superop^\dagger,\gameval \eye_{\bspace},\bar{\astate})$
and then with
$(T,B,\bar\rho)=(\superop,\gameval \eye_{\aspace},\bar{\bstate})$.
The descriptions of the optimal sets and the residual bounds give
\begin{align*}
\fnorm{\astate-\proj_{\aspaceopt}(\astate)}
&\leq K_Ag_A(\astate),\\
\fnorm{\bstate-\proj_{\bspaceopt}(\bstate)}
&\leq K_Bg_B(\bstate),
\end{align*}
where
\begin{equation}
K_A=\kappa_A\left(1+\frac{\sqrt2}{a}\right)\frac{\sqrt{d_B}}{b},
\qquad
K_B=\kappa_B\left(1+\frac{\sqrt2}{b}\right)\frac{\sqrt{d_A}}{a}.
\label{eqn:spms-full-rank-constants}
\end{equation}
Here $\kappa_A$ and $\kappa_B$ are the constants from
Claim~\ref{claim:spms-affine-repair} for
$\superop^\dagger$ and $\superop$, respectively. Take $K=\max\{K_A,K_B\}$ and $c=1/K$. The two bounds above prove
\eqref{eqn:bound-deviation-alpha} and
\eqref{eqn:bound-deviation-beta}.

Projection onto $\aspaceopt\times\bspaceopt$ separates into the
two component projections. Consequently,
\begin{align*}
\dist(\joint)
&=\left(\fnorm{\astate-\proj_{\aspaceopt}(\astate)}^2
        +\fnorm{\bstate-\proj_{\bspaceopt}(\bstate)}^2\right)^{1/2}\\
&\leq\fnorm{\astate-\proj_{\aspaceopt}(\astate)}
      +\fnorm{\bstate-\proj_{\bspaceopt}(\bstate)}\\
&\leq K_Ag_A(\astate)+K_Bg_B(\bstate)
 \leq K\gap(\joint).
\end{align*}
This proves \eqref{eqn:spms-linear-error-bound}.
Combining it with \eqref{ineqn:condition-measure-ineqn} gives
$R(\joint)\geq\dist(\joint)/(2K)$ outside the equilibrium set,
establishing SP-MS with $\spmsparam=0$.
All constants use the same fixed full-rank equilibrium and are
therefore uniform over $\joint$. The projected equilibrium
$\sol{\joint}$ itself may be rank deficient.

\medskip
\noindent\textbf{The unconditional case.}
We now drop the full-rank assumption. We use the following standard
form of the semialgebraic \L{}ojasiewicz inequality
\cite[Theorem~1.1]{BasuMohammadNezhad2024}: if $D\subseteq\mathbb R^N$
is compact and semialgebraic, and $u,w:D\to\mathbb R$ are
continuous semialgebraic functions with
$u^{-1}(0)\subseteq w^{-1}(0)$, then there exist an integer
$q\geq1$ and $K>0$ such that $|w(x)|^q\leq K|u(x)|$ for all
$x\in D$. To prove the error bound, we verify every hypothesis for $D=\jointsp$, $u=\gap$, and
$w=\dist$. The general arguments in \cite[Theorem 2.9]{BasuMohammadNezhad2024} already cover our case. Here, we give a self-contained proof for completeness.

\begin{claim}\label{claim:spms-semialgebraic}
In real coordinates, $\jointsp$ is compact and semialgebraic,
and $\gap$ and $\dist$ are continuous semialgebraic functions
on $\jointsp$ with common zero set $\cond{\jointsp}$.
\end{claim}

\begin{proof}
Choose a real orthonormal basis of each Hermitian matrix space.
This identifies the product space with
$\mathbb R^{d_A^2+d_B^2}$ and the product Frobenius norm with the
Euclidean norm. Matrix entries are linear expressions in these
real coordinates, and principal minors of Hermitian matrices are
real polynomials.
A Hermitian matrix is positive semidefinite if and only if all
its principal minors are nonnegative. Together with the trace-one
equation, these finitely many polynomial inequalities describe
each spectraplex. Thus $\jointsp$ is semialgebraic.
It is closed and bounded, because the eigenvalues of a density
matrix are nonnegative and sum to one, so it is also compact.

For a $k\times k$ Hermitian matrix $H$ and a real number $t$,
\begin{align*}
t=\lambda_{\max}(H)
&\quad\Longleftrightarrow\quad
tI-H\succeq0\ \text{ and }\ \det(tI-H)=0,\\
t=\lambda_{\min}(H)
&\quad\Longleftrightarrow\quad
H-tI\succeq0\ \text{ and }\ \det(H-tI)=0.
\end{align*}
For example, the first positive semidefinite condition says that
$t$ is at least every eigenvalue, while the determinant equation
says that $t$ is an eigenvalue. These descriptions are polynomial
conditions, so the extreme-eigenvalue functions are semialgebraic.
The Tarski--Seidenberg theorem states that eliminating real
quantifiers from formulas involving polynomial equalities and
inequalities preserves semialgebraicity
\cite[Chapter~2]{BochnakCosteRoy1998}.
In particular, sums, differences, and compositions of
semialgebraic functions are semialgebraic. Therefore
\begin{equation*}
\gap(\joint)
=\lambda_{\max}(\superop^\dagger(\astate))
 -\lambda_{\min}(\superop(\bstate))
\end{equation*}
is semialgebraic.
The extreme eigenvalues change by at most the operator norm of
a matrix perturbation, so $\gap$ is continuous as well.
Claim~\ref{claim:compact} gives
$\gap^{-1}(0)=\cond{\jointsp}$.
Consequently, $\cond{\jointsp}$ is a nonempty compact
semialgebraic set.

Distance to a nonempty set is $1$-Lipschitz, so
$|\dist(\joint_1)-\dist(\joint_2)|
\leq\fnorm{\joint_1-\joint_2}$.
To check semialgebraicity, describe its graph by
\begin{equation}
\begin{split}
r=\dist(\joint)\quad\Longleftrightarrow\quad
&r\geq0,\\
&\exists U\in\cond{\jointsp}:
           \fnorm{\joint-U}^2=r^2,\\
&\forall V\in\cond{\jointsp}:
           \fnorm{\joint-V}^2\geq r^2
\quad (\joint\in\jointsp).
\end{split}
\label{eqn:spms-distance-graph}
\end{equation}
The minimizing point $U$ exists by compactness of
$\cond{\jointsp}$. All squared norms in this formula are
polynomials in real coordinates. Since the equilibrium set is
semialgebraic, the Tarski--Seidenberg theorem shows that this graph
is semialgebraic. Finally, closedness of $\cond{\jointsp}$ gives
$\dist^{-1}(0)=\cond{\jointsp}=\gap^{-1}(0)$.
\end{proof}

Apply the stated \L{}ojasiewicz inequality using
Claim~\ref{claim:spms-semialgebraic}.
It provides an integer $q\geq1$ and $K>0$ such that
$\dist(\joint)^q\leq K\gap(\joint)$ throughout $\jointsp$,
which proves \eqref{eqn:spms-holder-error-bound}.
For every $\joint\notin\cond{\jointsp}$,
\eqref{ineqn:condition-measure-ineqn} then yields
\begin{equation*}
R(\joint)\geq\frac12\gap(\joint)
\geq\frac{1}{2K}\dist(\joint)^q.
\end{equation*}
This is SP-MS with $\spmsparam=q-1$, completing the proof.
\end{proof}

The unconditional result establishes the existence of a finite
exponent for each fixed game. It does not identify the smallest
admissible exponent or establish $\spmsparam=0$ without an
additional assumption. The existence of a full-rank equilibrium
is a sufficient condition for the stronger linear error bound.

\subsection{Convergence Rates of \itersmooth}
\label{appendix:itersmooth-convergence}

\subsubsection{Road-map for Main Results}
\label{appendix:roadmap-itersmooth-convergence}
In this section, we aim to prove the following theorem.
\qiteratedconvergenceratefin*
Specifically, we {use the full-rank case of Theorem~\ref{thm:sp-ms-for-spectraplex-appendix} to obtain} the linear gap-to-distance error bound for spectraplexes. Using the error bound, we then prove the total number of iterations required for the \itersmooth algorithm to compute an $\varepsilon$-equilibrium. The outer-loop invariant is {\(\gap(\joint_i)\leq\varepsilon_i\)}. It places {\(\joint_i\)} within {\(\varepsilon_i/\delta(\superop)\)} of the equilibrium set, while stage {\(i\)} targets {\(\varepsilon_{i+1}=\varepsilon_i/\gamma\)}. Proposition~\ref{prop:qsmoothing} therefore gives a stage cost independent of {\(i\)}, and the geometric decrease of the targets yields a logarithmic number of stages.
\begin{restatable}
{proposition}{propconditionnumber} \label{prop:condition-number}
   Let $\superop: \bspace \to \aspace$ be the superoperator defined by POVM and $\mathcal{U}$ as in \eqref{eqn:superoperator}, and let the duality gap be defined as in \eqref{eqn:duality-gap}. Suppose the quantum game admits an equilibrium $(\bar{\astate},\bar{\bstate})$ with
$\bar{\astate}\succ0$ and $\bar{\bstate}\succ0$. Then there exists $\delta > 0$ such that, for all $\joint \in \jointsp$,
    \begin{align}
        \dist(\joint) \leq \frac{\gap(\joint)}{\delta}.
    \end{align}
\end{restatable}
Having established the error bound, we now define the condition measure of a superoperator $\superop$ as
\begin{align}
    \delta(\superop) = \sup_\delta \setdef*{\delta}{\dist(\joint) \leq \frac{\gap(\joint)}{\delta}, \, \forall \joint \in \jointsp}.
\end{align}
With this definition, we can state the main theorem on the iteration complexity of the \itersmooth.
\begin{restatable}{theorem}{thmqiteratedconvergencerate} \label{thm:q-iterated-convergence-rate}
    Let {\(\varepsilon_0=\gap(\joint_0)\)}. If {\(\varepsilon_0\leq\varepsilon\)}, then \itersmooth terminates immediately. Otherwise, the total number {\(T\)} of first-order iterations satisfies
    \begin{align*}
        T
        \leq
        \ceil*{
            \frac{2\sqrt{2}\cdot\gamma\cdot\opnorm{\gradient}\cdot\sqrt{\smoothradius}}
            {\delta(\superop)}
        }
        \ceil*{
            \frac{\ln(\varepsilon_0/\varepsilon)}{\ln(\gamma)}
        }.
    \end{align*}
    Consequently,
    \begin{align*}
        T
        =
        \bigoh\!\left(
            \frac{\gamma\cdot\opnorm{\gradient}\cdot\sqrt{\smoothradius}}
            {\delta(\superop)\ln(\gamma)}
            \ln\!\left(\frac{\varepsilon_0}{\varepsilon}\right)
        \right).
    \end{align*}
\end{restatable}
The proof of Proposition \ref{prop:condition-number} is deferred to Appendix \ref{appendix:thm:q-iterated-convergence-rate}.
The proof of Theorem~\ref{thm:q-iterated-convergence-rate} proceeds in two steps. The \itersmooth algorithm consists of an inner loop, which invokes \qsmoothing at a fixed smoothing level, and an outer loop, which progressively decreases the smoothing parameter. In Step 1, we bound the number of iterations required for \qsmoothing to terminate within each inner loop. In Step 2, we bound the number of outer-loop stages before the algorithm terminates. We state these bounds as two claims.
Their proofs, as well as the proof of Theorem~\ref{thm:q-iterated-convergence-rate},
appear in Appendix~\ref{appendix:thm:q-iterated-convergence-rate}.

Theorem \ref{thm:q-iterated-convergence-rate-fin} now follows naturally from Theorem \ref{thm:q-iterated-convergence-rate}.
\begin{proof}[Proof of Theorem \ref{thm:q-iterated-convergence-rate-fin}]
    Let {\(\varepsilon_0=\gap(\joint_0)\)}. By the definition of the duality gap and the Cauchy--Schwarz inequality,
    \begin{align*}
        \varepsilon_0
        =
        \max_{\jointalt\in\jointsp}
        \inner{\gradient(\joint_0),\joint_0-\jointalt}
        \leq        \fnorm{\gradient(\joint_0)}\operatorname{diam}(\jointsp)
        \leq   \opnorm{\gradient}\fnorm{\joint_0}\operatorname{diam}(\jointsp).
    \end{align*}
    Each component of {\(\joint_0\)} is a density matrix, so {\(\fnorm{\joint_0}\leq\sqrt{2}\)}. Together with {\(\operatorname{diam}(\jointsp)\leq2\)}, this gives
    \begin{align*}
        \varepsilon_0\leq2\sqrt{2}\,\opnorm{\gradient}.
    \end{align*}
    Moreover, since {\(\jointtilde\in\jointsp\)} and {\(\operatorname{diam}(\jointsp)\leq2\)}, the definition~\eqref{eqn:smoothing-prox-radius} gives {\(\smoothradius\leq2\)}. Fixing {\(\gamma=e\)} in Theorem~\ref{thm:q-iterated-convergence-rate} and writing {\(\kappa(\superop)=\opnorm{\gradient}/\delta(\superop)\)}, we obtain
    \begin{align*}
        T
        &=
        \bigoh\!\left(
            \kappa(\superop)
            \ln\!\left(\frac{2\sqrt{2}\,\opnorm{\gradient}}{\varepsilon}\right)
        \right)
        =
        \bigoh\!\left(
            \kappa(\superop)
            \ln\!\left(\frac{\opnorm{\gradient}}{\varepsilon}\right)
        \right)
    \end{align*}
    in the nontrivial accuracy regime. This is the claimed asymptotic bound.
\end{proof}

\subsubsection{Proofs of Proposition \ref{prop:condition-number} and Theorem \ref{thm:q-iterated-convergence-rate}}
We restate and prove Proposition \ref{prop:condition-number}.
\propconditionnumber*

\begin{proof}
    The adjoint calculation preceding Inequality \eqref{ineqn:condition-measure-ineqn} gives the first equality below. Moreover, the full-rank error bound \eqref{eqn:spms-linear-error-bound} gives the inequality with \(c=1/K\).
    \begin{align*}
        & \frac{1}{2} \sup_{\jointalt \in \jointsp} \inner{\gradpayoff(\joint),\joint - \jointalt}
        = \frac{1}{2} \parens*{\max_{\balt \in \bspace} \tr{\astate \superop(\balt)} -  \min_{\aalt \in \aspace} \tr{\aalt \superop(\b)} } \\
        & = \frac{1}{2} \gap(\joint) \geq \frac{c}{2} \fnorm*{\joint - \projz(\joint)} = \frac{c}{2} \dist(\joint), \quad \forall \joint \in \jointsp.
    \end{align*}
    Therefore, we have $\dist(\joint) \leq \gap(\joint) / c,$ so one may take $\delta=c$, completing the proof.
\end{proof}
We now restate and prove Theorem \ref{thm:q-iterated-convergence-rate}.
\thmqiteratedconvergencerate*
\begin{proof}
If \(\gap(\joint_0)\leq\varepsilon\), Algorithm~\ref{alg:iterative smoothing} terminates immediately. Henceforth assume \(\varepsilon_0=\gap(\joint_0)>\varepsilon\); in particular, every smoothing parameter formed before termination is positive.

By construction, the outer loop maintains
\begin{align} \label{eqn:itersmooth-invariant}
    \gap(\joint_i)\leq\varepsilon_i
\end{align}
at the beginning of every stage {\(i\)}. The base case holds because {\(\varepsilon_0=\gap(\joint_0)\)}. For the induction step, stage {\(i\)} sets {\(\varepsilon_{i+1}=\varepsilon_i/\gamma\)} and calls \qsmoothing from {\(\joint_i\)} with target {\(\varepsilon_{i+1}\)}. The returned point therefore satisfies {\(\gap(\joint_{i+1})\leq\varepsilon_{i+1}\)}, which establishes the invariant at stage {\(i+1\)}.

We now use this invariant to bound the cost of each stage and the number of stages.
\label{appendix:thm:q-iterated-convergence-rate}
\begin{claim} \label{claim:claim1-q-iterated-convergence-rate}
    Each outer-stage call to \qsmoothing requires at most
    \begin{align*}
        K_{\mathrm{stage}}
        \defeq
        \ceil*{
            \frac{2\sqrt{2}\cdot\gamma\cdot\opnorm{\gradient}\cdot\sqrt{\smoothradius}}
            {\delta(\superop)}
        }
    \end{align*}
    inner iterations.
\end{claim}
\begin{claim} \label{claim:claim2-q-iterated-convergence-rate}
   The \itersmooth algorithm invokes \qsmoothing at most
   \begin{align*}
       N
       \defeq
       \ceil*{
           \frac{\ln(\varepsilon_0/\varepsilon)}{\ln(\gamma)}
       }
   \end{align*}
   times.
\end{claim}
\begin{proof} \textbf{of Claim \ref{claim:claim1-q-iterated-convergence-rate}}
    At the beginning of stage {\(i\)}, the invariant~\eqref{eqn:itersmooth-invariant} and the global error bound give
    \begin{align*}
        \dist(\joint_i,\cond{\jointsp})
        \leq
        \frac{\gap(\joint_i)}{\delta(\superop)}
        \leq
        \frac{\varepsilon_i}{\delta(\superop)}.
    \end{align*}
    Stage {\(i\)} calls \qsmoothing from {\(\joint_i\)} with target {\(\varepsilon_{i+1}=\varepsilon_i/\gamma\)}. Applying Proposition~\ref{prop:qsmoothing} at this stage yields
    \begin{align*}
        k_i
        &\leq
        \ceil*{
            \frac{2\sqrt{2}\cdot\opnorm{\gradient}\cdot\sqrt{\smoothradius}
            \cdot\dist(\joint_i,\cond{\jointsp})}
            {\varepsilon_{i+1}}
        }\\
        &\leq
        \ceil*{
            \frac{2\sqrt{2}\cdot\opnorm{\gradient}\cdot\sqrt{\smoothradius}}
            {\varepsilon_{i+1}}
            \frac{\varepsilon_i}{\delta(\superop)}
        }\\
        &=
        \ceil*{
            \frac{2\sqrt{2}\cdot\gamma\cdot\opnorm{\gradient}\cdot\sqrt{\smoothradius}}
            {\delta(\superop)}
        }
        =K_{\mathrm{stage}}.
    \end{align*}
\end{proof}       

\begin{proof} \textbf{of Claim \ref{claim:claim2-q-iterated-convergence-rate}}
    The geometric schedule gives {\(\varepsilon_i=\varepsilon_0/\gamma^i\)}. For
    \begin{align*}
        N
        =
        \ceil*{
            \frac{\ln(\varepsilon_0/\varepsilon)}{\ln(\gamma)}
        },
    \end{align*}
    we have {\(\varepsilon_N\leq\varepsilon\)}. After the {\(N\)}-th \qsmoothing call, the invariant gives
    \begin{align*}
        \gap(\joint_N)\leq\varepsilon_N\leq\varepsilon.
    \end{align*}
    Hence the algorithm terminates after at most {\(N\)} outer stages, equivalently, after at most {\(N\)} calls to \qsmoothing.
\end{proof}
    Combining Claims~\ref{claim:claim1-q-iterated-convergence-rate} and~\ref{claim:claim2-q-iterated-convergence-rate}, the total number of first-order iterations satisfies
    \begin{align*}
        T
        \leq
        K_{\mathrm{stage}}N
        =
        \ceil*{
            \frac{2\sqrt{2}\cdot\gamma\cdot\opnorm{\gradient}\cdot\sqrt{\smoothradius}}
            {\delta(\superop)}
        }
        \ceil*{
            \frac{\ln(\varepsilon_0/\varepsilon)}{\ln(\gamma)}
        }.
    \end{align*}
    Since {\(\gamma\)} and the game are fixed, the ceiling operations affect only constant factors, which gives the stated asymptotic bound.
\end{proof}

\subsection{Convergence Rates of \ogda}
\label{appendix_subsec:ogda-convergence}
In \cite{wei2021linear}, Wei et al. proves the following convergence rate for $\ogda$.
\begin{theorem}[Theorem 8 in \citet{wei2021linear}]\label{thm:wei2021-theorem8}
    For smoothness parameter $L$ (of $F$) and step size $\eta \leq \frac{1}{8L}$, if SP-MS holds with $\b=0$ , then $\ogda$ guarantees linear last-iterate convergence:
    \begin{align*}
        \mathrm{dist}^2(z_t,\mathcal{Z^*}) \leq 64\mathrm{dist}^2(\hat z_1,\mathcal{Z^*})(1+C_5)^{-t}.
    \end{align*}
    Here $C_5 := \min\{\,16\eta^2 C^2/81,\; 1/2\,\}$, where $C$ is the constant from the SP-MS condition. 
\end{theorem}
The proof of this convergence rate in \cite{wei2021linear} is for feasible sets $\aspace, \bspace \subset \mathbb{R}^d$. However, the same proof also works for feasible sets inside any Hilbert space that supports the SP-MS condition over the feasible sets. We already have SP-MS with exponent zero under the full-rank equilibrium hypothesis for spectraplexes $\aspace$ and $\bspace$ in Theorem \ref{thm:sp-ms-for-spectraplex-appendix}. Therefore, we are ready to prove the following theorem, which is the number of iterations for $\ogda$ to compute a $\varepsilon$-equilibrium for quantum zero sum games.

\begin{theorem} (Convergence for \ogda)\label{thm:ogda-convergence-rate-appendix}
Suppose the quantum game admits an equilibrium
$(\bar{\astate},\bar{\bstate})$ with
$\bar{\astate}\succ0$ and $\bar{\bstate}\succ0$. Let $L>0$ be a Lipschitz constant for $\gradpayoff$ in the product Frobenius norm. Let $\eta$ be a fixed step size satisfying $0<\eta\leq 1/(8L)$. Then \ogda computes an $\varepsilon$-equilibrium in at most $\bigoh_d(\log(1/\varepsilon))$
    first-order iterations.
\end{theorem}

\begin{proof}
    First, by Lemma 9 in \cite{Vasconcelos2025quadraticspeedupin}, we have $L = \bigoh_d(1)$ for $\gradpayoff$ in quantum zero sum games, which allows us to choose the step size $0<\eta\leq 1/(8L)$. Since under the assumption, we know that SP-MS (for spectraplexes) holds for quantum zero sum games with $\spmsparam=0$  from Theorem \ref{thm:sp-ms-for-spectraplex-appendix}, for payoff $\payoff$ of quantum zero sum games, the $\ogda$ has the following convergence by Theorem \ref{thm:wei2021-theorem8}:
    \begin{align*}
        \dist^2(\joint_t) \leq 64\mathrm{dist}^2(\hat\joint_1)(1+C_5)^{-t}.
    \end{align*}
    The gap bound in Section~\ref{subsec:equilibrium-measures} therefore gives
    \begin{align*}
        \gap(\joint_t) &\leq \sqrt{2}\,\opnorm{\gradient}\dist(\joint_t)\\
        &\leq 8\sqrt{2}\,\opnorm{\gradient}\dist(\hat\joint_1)(1+C_5)^{-t/2}.
    \end{align*}
    In terms of the number of iterations, for a fixed game and fixed step size, we can conclude that the $\ogda$ computes $\varepsilon$-equilibrium for quantum zero sum games in at most $\bigoh_d(\log(1/\varepsilon))$ iterations.
\end{proof}

\subsection{Convergence Rates of Regularized \ommwu}
\label{subsec:ommwu-convergence}
\subsubsection{Road-map for Main Results}
This section gives a proof for the quantum last-iterate convergence of {regularized }\ommwu in quantum zero-sum games. We study two-player zero-sum games where each player’s action is a density matrix and the payoff is bilinear through the linear map $\superop$ (and its adjoint $\superop^*$).
The {regularized }\ommwu for quantum games iterates $(\joint_t,\hat \joint_t)$ evolve by matrix exponentials, so the natural potential is the quantum relative entropy. Our goal is a \emph{last-iterate} guarantee toward the QRE, i.e., a bound directly on $S({\joint^{(\delta)}}\|{\hat\joint_{t+1}})$ rather than on an average iterate.

\medskip
\noindent As in the classical case, the starting point is the one-step regret inequality for optimistic mirror descent-ascent, which applies to updates of the form
\begin{align*}
 \joint_t &= \argmin_{\joint\in \mathcal{Z}}\Bigl\{\eta\langle \joint, \gradient(\hat \joint_t)\rangle + D_\psi(\joint,\hat \joint_t)\Bigr\},\\
\hat \joint_{t+1} &= \argmin_{\joint\in \mathcal{Z}}\Bigl\{\eta\langle \joint, \gradient(\joint_t)\rangle + D_\psi(\joint,\hat \joint_t)\Bigr\}.
\end{align*}
In our quantum setting, we apply Hilbert-Schmidt inner product. The $\jointsp = \aspace \times \bspace$, $\psi$ is the negative von-Neumann entropy, and $D_\psi(\cdot,\cdot)$ coincides with the quantum relative entropy $S(\cdot\|\cdot)$. In order to connect the \ommwu to a quantal-response formulation, we introduce an explicit entropic regularization term into each optimistic prox subproblem:
\begin{align*}
 \joint_t &= \argmin_{\joint\in \mathcal{Z}}\Bigl\{\eta({\delta}\psi(\joint)+\langle \joint, \gradient(\hat \joint_t)\rangle) + D_\psi(\joint,\hat \joint_t)\Bigr\},\\
\hat \joint_{t+1} &= \argmin_{\joint\in \mathcal{Z}}\Bigl\{\eta({\delta}\psi(\joint)+\langle \joint, \gradient(\joint_t)\rangle) + D_\psi(\joint,\hat \joint_t)\Bigr\}.
\end{align*}

 This modification is natural because, in finite-dimensional zero-sum games, adding a negative-entropy penalty changes the equilibrium notion from a Nash equilibrium to a quantal response equilibrium (QRE), namely a regularized saddle point in which each player responds according to a softened best response rather than an exact best response.

In the present quantum setting, the same principle applies over the spectraplex: replacing the unregularized saddle-point problem by its entropy-regularized counterpart leads to the variational inequality associated with the monotone operator $\gradient +\delta \nabla \psi$, where $\psi$ is the negative von Neumann entropy and $\delta >0$ is the regularization parameter. Accordingly, the fixed point of the regularized \ommwu dynamics is no longer the unregularized Nash equilibrium, but rather the quantum analogue of a QRE, i.e., the solution $\joint^{(\delta)}\in \operatorname{int}(\mathcal Z)$ of the min-max problem
\begin{align}
\label{problem:regularized-min-max}
\min_{\a \in \aspace} \max_{\b \in \bspace} u(\a,\b)+\delta g_1(\a)-\delta g_2(\b)    
\end{align}
 where $g_1,g_2$ are entropic regularizers for both player A and B. To prove linear convergence to this regularized solution, we apply the non-Euclidean prox lemma \cite[Proposition. 2.3]{Bauschke07ProxTheorem} to the second optimistic update with comparison point $\joint^{(\delta)}$, obtaining a descent inequality for $D_\psi(\joint^{(\delta)},\hat{\joint}_{t+1})$, and we apply the same lemma to the first optimistic update with comparison point $\hat{\joint}_{t+1}$ in order to control the intermediate term $D_\psi(\hat{\joint}_{t+1},\hat{\joint}_t)$. After substituting the latter inequality into the former, we can obtain
 \begin{align}
(1+\eta\delta)\,D_\psi(\joint^{(\delta)},\hat{\joint}_{t+1})
\le\;&
D_\psi(\joint^{(\delta)},\hat{\joint}_t)
-
D_\psi(\joint_t,\hat{\joint}_t)
-
(1+\eta\delta)\,D_\psi(\hat{\joint}_{t+1},\joint_t)
\notag\\
&\quad
+
\eta\,\Big\langle
\gradient(\joint_t)+\delta\nabla\psi(\joint^{(\delta)}),
\,\joint^{(\delta)}-\joint_t
\Big\rangle
\notag\\
&\quad
+
\eta\,\Big\langle
\gradient(\hat{\joint}_t)-\gradient(\joint_t),
\,\hat{\joint}_{t+1}-\joint_t
\Big\rangle
-
\eta\delta D_\psi(\joint_t,\joint^{(\delta)}).
\label{eq:main-recursion-star}
\end{align}
 The Bregman terms combine so that a factor $(1+\eta{\delta})$ appears on the left-hand side, while the remaining terms consist of an inner product involving $\gradient(\joint_t)+\delta\nabla\psi(\joint^{(\delta)})$ and an cross term involving $\gradient(\hat{\joint}_t)-\gradient(\joint_t)$. The first term is nonpositive because $\joint^{(\delta)}$ is an QRE and $\gradient$ is monotone, whereas the cross term is controlled by the $L$-smoothness of $F$ and the strong convexity of $\psi$, allowing it to be absorbed into the negative Bregman terms whenever $\eta\le 1/L$. This yields the contraction
\[
D_\psi(\joint^{(\delta)},\hat{\joint}_{t+1})
\le
({1+\eta\delta})^{-1}\,D_\psi(\joint^{(\delta)},\hat{\joint}_t),
\]
and hence
\[
D_\psi(\joint^{(\delta)},\hat{\joint}_{t+1})
\le
(1+\eta\delta)^{-t}D_\psi(\joint^{(\delta)},\hat{\joint}_1).
\]
In particular, when $\psi$ is chosen to be the negative von Neumann entropy, $D_\psi$ is exactly the corresponding quantum relative entropy, so the above estimate gives a linear convergence guarantee for the regularized \ommwu iterates to the quantum QRE. Finally, a QRE $\joint^{(\delta)}$ satisfies that $\gap(\joint^{(\delta)}) \le \mathcal{O}(\delta)$ where $\gap$ is the duality gap. By Corollary~\ref{prop:main:condition-number}, under the full-rank Nash equilibrium hypothesis, we can turn the upper bound of $\gap$ into an upper bound for geometric distance between $\joint^{(\delta)}$ and the Nash equilibrium of the game. Finally, we can conclude the following convergence result:
\ommwuconvergence*

\subsubsection{Auxiliary Lemmas}
In this section, we present the necessary lemmas for the proof of the convergence rates of \ommwu.
\begin{proposition}[{\citealp[Proposition 2.3]{Bauschke07ProxTheorem}}]
\label{prop:QRE-D1}
Let $\{x,y\}\subseteq \mathrm{dom}(\psi)$ and $\{u,v\}\subseteq \operatorname{int}\mathrm{dom}(\psi)$. Then
\begin{align}
&D_\psi(u,v)+D_\psi(v,u)
=
\langle \nabla \psi(u)-\nabla \psi(v),\,u-v\rangle, \\
&D_\psi(x,u)
=
D_\psi(x,v)+D_\psi(v,u)+\langle \nabla \psi(v)-\nabla \psi(u),\,x-v\rangle, \\
&D_\psi(x,v)+D_\psi(y,u)
=
D_\psi(x,u)+D_\psi(y,v)+\langle \nabla \psi(u)-\nabla \psi(v),\,x-y\rangle.
\end{align}
\end{proposition}

\begin{proposition}[{\citealp[Proposition D.2]{Sokota2023QRE}}]
\label{prop:QRE-D2}
Assume $\mathcal{Z}$ closed convex and both $f$ and $\psi$ are differentiable at $\bar z$ (defined below).
Then the following statements are equivalent:
\begin{enumerate}
    \item $\bar z
    =
    \arg\min_{z\in \mathcal Z}
    \eta \langle g ,z\rangle + f(z) + D_\psi(z,y).$
    \item For all $z\in \mathcal Z$, $\langle \eta g + \nabla f(\bar z),\,\bar z-z\rangle
    \le
    D_\psi(z,y)-D_\psi(z,\bar z)-D_\psi(\bar z,y).$
\end{enumerate}
\end{proposition} 
For \ommwu, we can write each proximal step in the form of (1) in Proposition \ref{prop:QRE-D2}. The assumptions for $f,\psi$ are satisfied, given $f  = \delta\eta \cdot  \psi$ where $\psi$ is the negative von Neumann entropy and $g = \gradient$ where $\gradient$ is the gradient operator. Therefore, as we assume that the iterations of \ommwu stays in the interior of player's strategy space, the Proposition \ref{prop:QRE-D2} can be applied to each step of \ommwu.
\begin{lemma}[QRE implies approximate Nash equilibrium]
\label{lemm:qre-eps-nash-gap-vn}
Let \((\a^{(\delta)},\b^{(\delta)})\) be a \(\delta\)-QRE, namely a solution of the
entropy-regularized quantum zero-sum game
\[
\min_{\a\in\aspace}\max_{\b\in\bspace}
\Bigl\{
u(\a,\b)+\delta g_A(\a)-\delta g_B(\b)
\Bigr\},
\]
where \(g_A\) and \(g_B\) are the negative von-Neumann entropy regularizers. Then it satisfies that $\gap(\a^{(\delta)},\b^{(\delta)})
\le\delta(\log d_A+\log d_B)$ for $d_A,d_B$ as the dimension of $\aspace,\bspace$ and $\gap$ as the \ref{eqn:duality-gap}. Therefore \((\a^{(\delta)},\b^{(\delta)})\) is an \(\varepsilon\)-Nash equilibrium with respect
to the duality gap for $\varepsilon=\delta(\log d_A+\log d_B).$
\end{lemma}

\begin{proof}
Since \((\a^{(\delta)},\b^{(\delta)})\) is a saddle point of the regularized problem, it satisfies
\begin{align}
&u(\a^{(\delta)},\b)+\delta g_A(\a^{(\delta)})-\delta g_B(\b)
\le
u(\a^{(\delta)},\b^{(\delta)})+\delta g_A(\a^{(\delta)})-\delta g_B(\b^{(\delta)})
\qquad
\forall \b\in\bspace,
\label{eq:qre-br-b-vn}
\\
&u(\a^{(\delta)},\b^{(\delta)})+\delta g_A(\a^{(\delta)})-\delta g_B(\b^{(\delta)})
\le
u(\a,\b^{(\delta)})+\delta g_A(\a)-\delta g_B(\b^{(\delta)})
\qquad
\forall \a\in\aspace.
\label{eq:qre-br-a-vn}
\end{align}
Rearranging \eqref{eq:qre-br-b-vn} yields
\[
u(\a^{(\delta)},\b)-u(\a^{(\delta)},\b^{(\delta)})
\le
\delta\bigl(g_B(\b)-g_B(\b^{(\delta)})\bigr)
\qquad
\forall \b\in\bspace,
\]
and therefore
\begin{equation}
\max_{\b\in\bspace}u(\a^{(\delta)},\b)-u(\a^{(\delta)},\b^{(\delta)})
\le
\delta\Bigl(\max_{\b\in\bspace} g_B(\b)-g_B(\b^{(\delta)})\Bigr).
\label{eq:qre-gap-b-vn-pre}
\end{equation}
{\emergencystretch=1em
Since \(g_B\) is the negative von-Neumann entropy, we can write this in terms of its maximum rather than the minimum of the corresponding entropy. Indeed, if
\(H_B(\b):=-g_B(\b)\), then\par}
\[
\min_{\b\in\bspace} H_B(\b)-H_B(\b^{(\delta)})
=
g_B(\b^{(\delta)})-\max_{\b\in\bspace} g_B(\b).
\]
Thus \eqref{eq:qre-gap-b-vn-pre} is equivalent to
\begin{equation}
\max_{\b\in\bspace}u(\a^{(\delta)},\b)-u(\a^{(\delta)},\b^{(\delta)})
\le
\delta\Bigl(-g_B(\b^{(\delta)})+\max_{\b\in\bspace} g_B(\b)\Bigr).
\label{eq:qre-gap-b-vn}
\end{equation}
Similarly, rearranging \eqref{eq:qre-br-a-vn} gives
\[
u(\a^{(\delta)},\b^{(\delta)})-u(\a,\b^{(\delta)})
\le
\delta\bigl(g_A(\a)-g_A(\a^{(\delta)})\bigr)
\qquad
\forall \a\in\aspace,
\]
and hence we get
\begin{equation}
u(\a^{(\delta)},\b^{(\delta)})-\min_{\a\in\aspace}u(\a,\b^{(\delta)})
\le
\delta\Bigl(-g_A(\a^{(\delta)})+\max_{\a\in\aspace} g_A(\a)\Bigr).
\label{eq:qre-gap-a-vn}
\end{equation}
Adding \eqref{eq:qre-gap-b-vn} and \eqref{eq:qre-gap-a-vn}, we obtain
\begin{align*}
\gap(\a^{(\delta)},\b^{(\delta)})
&=
\max_{\b\in\bspace}u(\a^{(\delta)},\b)
-
\min_{\a\in\aspace}u(\a,\b^{(\delta)})
\\
&=
\Bigl(
\max_{\b\in\bspace}u(\a^{(\delta)},\b)-u(\a^{(\delta)},\b^{(\delta)})
\Bigr)
+
\Bigl(
u(\a^{(\delta)},\b^{(\delta)})-\min_{\a\in\aspace}u(\a,\b^{(\delta)})
\Bigr)
\\
&\le
\delta\Bigl(-g_B(\b^{(\delta)})+\max_{\b\in\bspace} g_B(\b)\Bigr)
+
\delta\Bigl(-g_A(\a^{(\delta)})+\max_{\a\in\aspace} g_A(\a)\Bigr).
\end{align*}
It remains to make the bound explicit using the definition of negative von-Neumann entropy.
For any density matrix \(\rho\) on \(d\) dimensions with eigenvalues
\(\lambda_1,\dots,\lambda_d\), we have
\[
g(\rho)=\tr{\rho\log\rho}=\sum_{i=1}^d \lambda_i \log \lambda_i.
\]
Since \(x\log x\le 0\) for \(x\in[0,1]\), it follows that $g(\rho)\le 0.$ Moreover, by concavity of the von-Neumann entropy, we have $g(\rho)\ge -\log d$, with equality at the maximally mixed state \(\rho=\eye/d\). Therefore
\[
\min_{\a\in\aspace} g_A(\a)=-\log d_A,
\qquad
\min_{\b\in\bspace} g_B(\b)=-\log d_B,
\]
Also, $\max_{\a\in\aspace}g_A(\a)=\max_{\b\in\bspace}g_B(\b)=0$, attained at pure states,
and
\[
-g_A(\a^{(\delta)})+\max_{\a\in\aspace} g_A(\a)
=
-g_A(\a^{(\delta)})
\le
\log d_A,
\]
\[
-g_B(\b^{(\delta)})+\max_{\b\in\bspace} g_B(\b)
=
-g_B(\b^{(\delta)})
\le
\log d_B.
\]
Substituting these inequalities into the previous estimate gives
\[
\gap(\a^{(\delta)},\b^{(\delta)})
\le
\delta(\log d_A+\log d_B).
\]
Hence \((\a^{(\delta)},\b^{(\delta)})\) is an \(\varepsilon\)-Nash equilibrium with respect to
the duality gap for \(\varepsilon=\delta(\log d_A+\log d_B)\).
\end{proof}
In our context, we assume $d_A = 2^n$ and $d_B = 2^m$. Therefore we can further simplify the result of Lemma \ref{lemm:qre-eps-nash-gap-vn} as $\gap(\a^{(\delta)},\b^{(\delta)})
\le
\delta(n+m).$
\subsubsection{Proof of Theorem \ref{thm:ommwu-convergence}}
Here we restate and proof Theorem \ref{thm:ommwu-convergence}.
\ommwuconvergence*
\begin{proof}
We prove convergence of the \emph{regularized} \ommwu dynamics to the unique solution of \eqref{problem:regularized-min-max}. The proof follows the same framework as the proof of Theorem~3.4 of \cite{Sokota2023QRE}. In our notation, the regularized \ommwu iterates are
\begin{align}
\joint_t
&=
\arg\min_{\joint\in Z}
\Bigl\{
\eta\langle \gradient(\hat{\joint}_t),\joint\rangle+\eta\delta \psi(\joint)+D_\psi(\joint,\hat{\joint}_t)
\Bigr\},
\label{eq:omwmwu-proof-first-step}
\\
\hat{\joint}_{t+1}
&=
\arg\min_{\joint\in Z}
\Bigl\{
\eta\langle \gradient(\joint_t),\joint\rangle+\eta\delta \psi(\joint)+D_\psi(\joint,\hat{\joint}_t)
\Bigr\}.
\label{eq:omwmwu-proof-second-step}
\end{align}
Applying Proposition \ref{prop:QRE-D2} to \eqref{eq:omwmwu-proof-second-step} with
\(z=\joint^{(\delta)}\), \(y=\hat{\joint}_t\), \(g=\gradient(\joint_t)\), and \(f=\eta\delta\psi\), we obtain
\begin{align}
D_\psi(\joint^{(\delta)},\hat{\joint}_{t+1})
\le\,&
D_\psi(\joint^{(\delta)},\hat{\joint}_t)
-
D_\psi(\hat{\joint}_{t+1},\hat{\joint}_t)
\notag\\
&\quad
+
\eta\bigl\langle \gradient(\joint_t),\,\joint^{(\delta)}-\hat{\joint}_{t+1}\bigr\rangle
+
\eta\delta\bigl\langle \nabla\psi(\hat{\joint}_{t+1}),\,\joint^{(\delta)}-\hat{\joint}_{t+1}\bigr\rangle.
\label{eq:omwmwu-proof-prox-second}
\end{align}
Likewise, applying Proposition \ref{prop:QRE-D2} to \eqref{eq:omwmwu-proof-first-step} with
\(z=\hat{\joint}_{t+1}\), \(y=\hat{\joint}_t\), \(g=\gradient(\hat{\joint}_t)\), and \(f=\eta\delta\psi\), we can obtain
\begin{align}
D_\psi(\hat{\joint}_{t+1},\joint_t)
\le\,&
D_\psi(\hat{\joint}_{t+1},\hat{\joint}_t)
-
D_\psi(\joint_t,\hat{\joint}_t)
\notag\\
&\quad
+
\eta\bigl\langle \gradient(\hat{\joint}_t),\,\hat{\joint}_{t+1}-\joint_t\bigr\rangle
+
\eta\delta\bigl\langle \nabla\psi(\joint_t),\,\hat{\joint}_{t+1}-\joint_t\bigr\rangle.
\label{eq:omwmwu-proof-prox-first}
\end{align}
Rearranging \eqref{eq:omwmwu-proof-prox-first}, we get
\begin{align}
-D_\psi(\hat{\joint}_{t+1},\hat{\joint}_t)
\le\,&
-D_\psi(\hat{\joint}_{t+1},\joint_t)
-D_\psi(\joint_t,\hat{\joint}_t)
\notag\\
&\quad
+
\eta\bigl\langle \gradient(\hat{\joint}_t),\,\hat{\joint}_{t+1}-\joint_t\bigr\rangle
+
\eta\delta\bigl\langle \nabla\psi(\joint_t),\,\hat{\joint}_{t+1}-\joint_t\bigr\rangle.
\label{eq:omwmwu-proof-prox-first-rearranged}
\end{align}
Substituting \eqref{eq:omwmwu-proof-prox-first-rearranged} into
\eqref{eq:omwmwu-proof-prox-second}, we arrive at
\begin{align}
D_\psi(\joint^{(\delta)},\hat{\joint}_{t+1})
\le\,&
D_\psi(\joint^{(\delta)},\hat{\joint}_t)
-
D_\psi(\hat{\joint}_{t+1},\joint_t)
-
D_\psi(\joint_t,\hat{\joint}_t)
\notag\\
&\quad
+
\eta\bigl\langle \gradient(\hat{\joint}_t),\,\hat{\joint}_{t+1}-\joint_t\bigr\rangle
+
\eta\bigl\langle \gradient(\joint_t),\,\joint^{(\delta)}-\hat{\joint}_{t+1}\bigr\rangle
\notag\\
&\quad
+
\eta\delta\bigl\langle \nabla\psi(\joint_t),\,\hat{\joint}_{t+1}-\joint_t\bigr\rangle
+
\eta\delta\bigl\langle \nabla\psi(\hat{\joint}_{t+1}),\,\joint^{(\delta)}-\hat{\joint}_{t+1}\bigr\rangle.
\label{eq:omwmwu-proof-after-substitution}
\end{align}
We now simplify the operator and entropy terms. First,
\begin{align}
\bigl\langle \gradient(\hat{\joint}_t),\,\hat{\joint}_{t+1}-\joint_t\bigr\rangle
&+
\bigl\langle \gradient(\joint_t),\,\joint^{(\delta)}-\hat{\joint}_{t+1}\bigr\rangle
\notag\\
&=\bigl\langle \gradient(\joint_t),\,\joint^{(\delta)}-\joint_t\bigr\rangle
+
\bigl\langle \gradient(\hat{\joint}_t)-\gradient(\joint_t),\,\hat{\joint}_{t+1}-\joint_t\bigr\rangle.
\label{eq:omwmwu-proof-F-rewrite}
\end{align}
Second, using Proposition \ref{prop:QRE-D1}, we have
\begin{align}
&\bigl\langle \nabla\psi(\joint_t),\,\hat{\joint}_{t+1}-\joint_t\bigr\rangle
+
\bigl\langle \nabla\psi(\hat{\joint}_{t+1}),\,\joint^{(\delta)}-\hat{\joint}_{t+1}\bigr\rangle
\notag\\
&\qquad=
\bigl\langle \nabla\psi(\joint^{(\delta)}),\,\joint^{(\delta)}-\joint_t\bigr\rangle
-
D_\psi(\joint_t,\joint^{(\delta)})
-
D_\psi(\joint^{(\delta)},\hat{\joint}_{t+1})
-
D_\psi(\hat{\joint}_{t+1},\joint_t).
\label{eq:omwmwu-proof-psi-rewrite}
\end{align}
Substituting \eqref{eq:omwmwu-proof-F-rewrite} and \eqref{eq:omwmwu-proof-psi-rewrite}
into \eqref{eq:omwmwu-proof-after-substitution} gives
\begin{align}
D_\psi(\joint^{(\delta)},\hat{\joint}_{t+1})
\le\,&
D_\psi(\joint^{(\delta)},\hat{\joint}_t)
-
D_\psi(\hat{\joint}_{t+1},\joint_t)
-
D_\psi(\joint_t,\hat{\joint}_t)
\notag\\
&\quad
+
\eta\bigl\langle \gradient(\joint_t),\,\joint^{(\delta)}-\joint_t\bigr\rangle
+
\eta\bigl\langle \gradient(\hat{\joint}_t)-\gradient(\joint_t),\,\hat{\joint}_{t+1}-\joint_t\bigr\rangle
\notag\\
&\quad
+
\eta\delta\bigl\langle \nabla\psi(\joint^{(\delta)}),\,\joint^{(\delta)}-\joint_t\bigr\rangle
-
\eta\delta D_\psi(\joint_t,\joint^{(\delta)})
\notag\\
&\quad
-
\eta\delta D_\psi(\joint^{(\delta)},\hat{\joint}_{t+1})
-
\eta\delta D_\psi(\hat{\joint}_{t+1},\joint_t).
\label{eq:omwmwu-proof-main-recursion-pre}
\end{align}
Moving the term
\(\eta\delta D_\psi(\joint^{(\delta)},\hat{\joint}_{t+1})\) to the left-hand side yields
\begin{align}
(1+\eta\delta)\,D_\psi(\joint^{(\delta)},\hat{\joint}_{t+1})
\le\,&
D_\psi(\joint^{(\delta)},\hat{\joint}_t)
-
D_\psi(\joint_t,\hat{\joint}_t)
-
(1+\eta\delta)\,D_\psi(\hat{\joint}_{t+1},\joint_t)
\notag\\
&\quad
+
\eta\bigl\langle \gradient(\joint_t)+\delta\nabla\psi(\joint^{(\delta)}),\,\joint^{(\delta)}-\joint_t\bigr\rangle
\notag\\
&\quad
+
\eta\bigl\langle \gradient(\hat{\joint}_t)-\gradient(\joint_t),\,\hat{\joint}_{t+1}-\joint_t\bigr\rangle
-
\eta\delta D_\psi(\joint_t,\joint^{(\delta)}).
\label{eq:omwmwu-proof-main-recursion}
\end{align}
\noindent Since \(\joint^{(\delta)}\) solves the variational inequality \(\operatorname{VI}(\mathcal{Z},\gradient+\delta\nabla\psi)\) by Proposition 2.4 in \cite{Sokota2023QRE}, we have
\[
\bigl\langle \gradient(\joint^{(\delta)})+\delta\nabla\psi(\joint^{(\delta)}),\,\joint_t-\joint^{(\delta)}\bigr\rangle
\ge 0.
\]
Since \(\gradient\) is monotone, we have
\[
\bigl\langle \gradient(\joint_t)-\gradient(\joint^{(\delta)}),\,\joint_t-\joint^{(\delta)}\bigr\rangle \ge 0.
\]
Adding previous two inequalities gives
\[
\bigl\langle \gradient(\joint_t)+\delta\nabla\psi(\joint^{(\delta)}),\,\joint^{(\delta)}-\joint_t\bigr\rangle
\le 0.
\]
Therefore, the corresponding term in \eqref{eq:omwmwu-proof-main-recursion} is nonpositive and may be
discarded. We thus obtain
\begin{align}
(1+\eta\delta)\,D_\psi(\joint^{(\delta)},\hat{\joint}_{t+1})
\le\,&
D_\psi(\joint^{(\delta)},\hat{\joint}_t)
-
D_\psi(\joint_t,\hat{\joint}_t)
-
(1+\eta\delta)\,D_\psi(\hat{\joint}_{t+1},\joint_t)
\notag\\
&\quad
+
\eta\bigl\langle \gradient(\hat{\joint}_t)-\gradient(\joint_t),\,\hat{\joint}_{t+1}-\joint_t\bigr\rangle.
\label{eq:omwmwu-proof-main-recursion-simplified}
\end{align}
It remains to deal with the cross term. We use a sharper absorption than the Young inequality in order to improve the dependence on $\delta$. By \(L\)-smoothness of \(\gradient\), generalized Cauchy--Schwarz, and \(1\)-strong convexity of \(\psi\), we have
\begin{align}
\eta\bigl\langle \gradient(\hat{\joint}_t)-\gradient(\joint_t),\,
\hat{\joint}_{t+1}-\joint_t\bigr\rangle
&\le
\eta\|\gradient(\hat{\joint}_t)-\gradient(\joint_t)\|\,
\|\hat{\joint}_{t+1}-\joint_t\|
\notag\\
&\le
\eta L\|\hat{\joint}_t-\joint_t\|\,
\|\hat{\joint}_{t+1}-\joint_t\|
\notag\\
&\le
2\eta L
\sqrt{
D_\psi(\joint_t,\hat{\joint}_t)
D_\psi(\hat{\joint}_{t+1},\joint_t)
}
\notag\\
&\le
\eta L \big(D_\psi(\joint_t,\hat{\joint}_t)
+
 D_\psi(\hat{\joint}_{t+1},\joint_t) \big).
\label{eq:omwmwu-proof-cross-bound-improved}
\end{align}
The third inequality follows from
\(D_\psi(X,Y)\ge \frac12\|X-Y\|^2\), and the final inequality is
\(2\sqrt{ab}\le a+b\).
Substituting \eqref{eq:omwmwu-proof-cross-bound-improved} into
\eqref{eq:omwmwu-proof-main-recursion-simplified}, we obtain
\begin{align}
(1+\eta\delta)\,D_\psi(\joint^{(\delta)},\hat{\joint}_{t+1})
\le\,&
D_\psi(\joint^{(\delta)},\hat{\joint}_t)
-
(1-\eta L)D_\psi(\joint_t,\hat{\joint}_t)
\notag\\
&\quad
-
(1+\eta\delta-\eta L)D_\psi(\hat{\joint}_{t+1},\joint_t).
\label{eq:omwmwu-proof-final-ineq-improved}
\end{align}
Consequently, if $\eta \le \frac{1}{L}$, then $1-\eta L\ge 0$ and $1+\eta\delta-\eta L\ge \eta\delta\ge 0$.
Hence the last two terms on the right-hand side of
\eqref{eq:omwmwu-proof-final-ineq-improved} are nonpositive, and we conclude that
\[
(1+\eta\delta)\,D_\psi(\joint^{(\delta)},\hat{\joint}_{t+1})
\le
D_\psi(\joint^{(\delta)},\hat{\joint}_t).
\]
Equivalently,
\[
D_\psi(\joint^{(\delta)},\hat{\joint}_{t+1})
\le
\frac{1}{1+\eta\delta}\,
D_\psi(\joint^{(\delta)},\hat{\joint}_t).
\]
Iterating this estimate gives the linear convergence rate
\[
D_\psi(\joint^{(\delta)},\hat{\joint}_{t+1})
\le
(1+\eta\delta)^{-t}\,
D_\psi(\joint^{(\delta)},\hat{\joint}_1).
\]
This establishes the improved \(\mathcal O(1/\delta)\) dependence for convergence of the regularized
\ommwu dynamics to the QRE \(\joint^{(\delta)}\). Finally, by Lemma
\ref{lemm:qre-eps-nash-gap-vn} and Corollary~\ref{prop:main:condition-number}, using the full-rank hypothesis and uniqueness of the Nash equilibrium, we have
\[
\|\joint^{(\delta)}-\jointstar\|_F \le \frac{\delta(\log d_A+\log d_B)}{\delta(\superop)}\leq\mathcal{O}_d(\delta).
\]
Write $B_d:=\log d_A+\log d_B$. With maximally mixed initialization, $S(\joint^{(\delta)}\|\hat\joint_1)\leq B_d$. Since $S(X\|Y)\geq\frac12\fnorm{X-Y}^2$ on the product space, the contraction and the triangle inequality give
\begin{align*}
\fnorm{\hat\joint_{t+1}-\jointstar}
&\leq\fnorm{\hat\joint_{t+1}-\joint^{(\delta)}}+\fnorm{\joint^{(\delta)}-\jointstar}\\
&\leq\sqrt{2B_d}(1+\eta\delta)^{-t/2}+\frac{B_d\delta}{\delta(\superop)}.
\end{align*}
For $0<\varepsilon<1$, choose
\[
\delta=\frac{\delta(\superop)\varepsilon}{2B_d},
\qquad
t\geq\left\lceil\frac{\log(8B_d/\varepsilon^2)}{\log(1+\eta\delta)}\right\rceil.
\]
Each term is then at most $\varepsilon/2$, so returning $\hat\joint_{t+1}$ guarantees $\fnorm{\hat\joint_{t+1}-\jointstar}\leq\varepsilon$. For a fixed game and a fixed admissible step size $\eta>0$, this requires $\widetilde{\mathcal O}(1/\varepsilon)$ iterations, with constants depending on the game through $\delta(\superop)$ as well as on the dimensions and step size.
\end{proof}





\section{Slow Last-Iterate Convergence of \ommwu}
\label{appendix:slow-last-iter-ommwu}
\subsection{Conjecture on Linear Last-Iterate Convergence of \ommwu}
\label{appendix:conjecture-ommwu}
\providecommand{\companion}{\citep{su2026lowerbounds}}

Going back to the last-iterate behavior of \ommwu, our analysis establishes linear-rate convergence of {the regularized method} to a {quantal response equilibrium} for each fixed regularization parameter $\delta$ {(Theorem~\ref{thm:ommwu-convergence}), but not a linear rate for approximating a Nash equilibrium, since the passage from the QRE to Nash requires $\delta$ to vanish with the target accuracy}. {In the classical simplex setting, optimistic multiplicative weights converge linearly in the last iterate whenever the equilibrium is unique~\citep{wei2021linear}, and one might expect the same of unregularized \ommwu. Uniqueness alone, however, is \emph{not} sufficient on the spectraplex: the companion paper~\companion\ constructs a one-qubit game with a unique, strictly complementary, pure equilibrium on which unregularized \ommwu from the maximally mixed state has Frobenius distance and relative entropy to equilibrium of order $\Theta(1/t)$ and duality gap of order $\Theta(1/t^2)$ for every fixed step size, so that no geometric rate holds even with instance-dependent constants. The obstruction is eigenvector motion: the eigenvalues of the iterate concentrate exponentially fast, but the accumulated off-diagonal score tilts its leading eigenvector away from the equilibrium support by an angle that closes only like $1/t$. All slow games in~\companion\ share the feature that every equilibrium has a singular component. This motivates the following conjecture, in which the full-rank hypothesis plays the role that polyhedrality plays classically: near an interior equilibrium the quantum relative entropy is locally quadratic, so the entropic and Frobenius geometries agree to first order, and the linear error bound of Corollary~\ref{prop:main:condition-number} is available.}

\begin{conjecture}[Linear last-iterate convergence of \ommwu { under a full-rank equilibrium}]
\label{conj:ommwu-linear}
Consider a quantum zero-sum game over the spectraplex $\jointsp$ {whose Nash equilibrium $\joint^\star=(\aopt,\bopt)$ is unique and satisfies $\aopt\succ0$ and $\bopt\succ0$}.
Then there exists a constant step size $\eta_0>0$ such that, for every $\eta\in(0,\eta_0]$, the iterates generated by {unregularized} \ommwu\ {(Algorithm~\ref{alg:ommwu})} converge linearly in the last iterate:
\[
\dist(\joint_t,\joint^*)\le C(1-\gamma)^t
\]
for some instance-dependent constants $C>0$ and $\gamma\in(0,1)$.
Equivalently, \ommwu\ computes an $\varepsilon$-approximate Nash equilibrium in
\(
T=\bigoh\!\big(\mathfrak{C}\log(1/\varepsilon)\big)
\)
iterations, where $\mathfrak{C}$ is a geometric condition number of the game.
\end{conjecture}
\noindent {The full-rank hypothesis cannot be dropped, by the counterexample above, and by Section~\ref{subsec:ommwu-slowdown} the constants $C,\gamma$ cannot be uniform over games even at fixed dimension. The remaining obstacle is the one identified in~\companion: an argument that optimism controls the non-commutative (eigenvector) component of the dynamics, for which the classical coordinatewise support identification of \citet{wei2021linear} has no direct analogue.}

\subsection{Diagonal reduction and a last-iterate obstruction}

Although Conjecture~\ref{conj:ommwu-linear} remains open, one can still ask whether a last-iterate guarantee can hold uniformly over all quantum zero-sum games. Our next result gives an obstruction to such a uniform guarantee: even if \ommwu\ does admit linear last-iterate convergence, such a guarantee cannot be instance-independent over this class.
We lift the classical last-iterate lower bounds for \ommwu from \citet{cai2025fastlastiterateconvergencelearning} to quantum \ommwu\unskip, adapting their proof to the initialization in Algorithm~\ref{alg:ommwu}. We work throughout in the quantum zero-sum framework of Section~\ref{sec:preliminaries}. Our reduction specializes to the one-qubit case $\adim = \bdim = 1$ and selects a \emph{diagonal} payoff observable on the joint space.

Throughout this section, we study the unregularized updates with a fixed positive step size and continue the trajectory without the stopping test in Algorithm~\ref{alg:ommwu}. The negative von Neumann entropy determines the update geometry; no entropy term is added to the game's payoff.

Under the maximally mixed diagonal initialization, this diagonal structure yields three consequences: (i) the expected payoff depends only on the diagonal entries of $(\a, \b)$; (ii) the payoff gradients $\gradient_\a(\b)$ and $\gradient_\b(\a)$ are diagonal and coincide with the classical feedback vectors of an embedded $2 \times 2$ matrix game; and (iii) when \ommwu{} is instantiated with the von Neumann entropy regularizer, the updates preserve diagonality, so all iterates remain diagonal. Consequently, quantum \ommwu{} coincides exactly with classical \ommwu{} with the same initialization and update order on the embedded $2 \times 2$ game, and the quantum duality gap along the trajectory matches the classical duality gap. This is formally stated in Lemma \ref{lem:diag-reduction}. Consequently, any slow last-iterate convergence lower bound for classical $\ommwu$ under this initialization transfers to quantum $\ommwu$. The following lemma shows that this reduction holds for any diagonal payoff observable and the stated maximally mixed initialization.
\begin{lemma}\label{lem:diag-reduction}
Fix $A\in\R^{2\times 2}$ and define the diagonal payoff observable
\[
U := \begin{pmatrix}
A_{11} & 0 & 0 & 0 \\
0 & A_{12} & 0 & 0 \\
0 & 0 & A_{21} & 0 \\
0 & 0 & 0 & A_{22}
\end{pmatrix} = \diag(A_{11},A_{12},A_{21},A_{22}).
\]
Run \ommwu Algorithm \ref{alg:ommwu} without the stopping test, with von Neumann entropy regularizer on the quantum game
$\payoff(\a,\b)=\tr{U^\dagger(\a\otimes\b)}$ over $\aspace\oplus\bspace$ with $\adim=\bdim=1$,
and initialize diagonally as
\[
\a_0=\b_0=\hat\a_0=\hat\b_0=\frac{\eye}{2} \in \hmat^2_{+}.
\]
Then for every $t\ge 0$ the following hold.

\begin{enumerate}[leftmargin=*]
\item
All iterates remain diagonal: there exist $x_t,y_t,\hat x_t,\hat y_t\in\Delta^2$ such that
\[
\a_t=\diag(x_t),\quad \b_t=\diag(y_t),\quad
\hat\a_t=\diag(\hat x_t),\quad \hat\b_t=\diag(\hat y_t).
\]
\item
The induced simplex trajectory $(x_t,y_t)$ coincides exactly with the classical \ommwu using the same initialization and update order on the $2\times 2$ game $A$ with vectors
$\ell_x^t = Ay_t$ and $\ell_y^t = -A^\top x_t$.
In particular, writing $\operatorname{softmax}(v)_i=e^{v_i}/\sum_j e^{v_j}$, the complete simplex updates are
\begin{align}
x_{t+1}&=\operatorname{softmax}(\log\hat x_t-\eta Ay_t),
&\hat x_{t+1}&=\operatorname{softmax}(\log\hat x_t-\eta Ay_{t+1}),\notag\\
y_{t+1}&=\operatorname{softmax}(\log\hat y_t+\eta A^\top x_t),
&\hat y_{t+1}&=\operatorname{softmax}(\log\hat y_t+\eta A^\top x_{t+1}).
\label{eq:diag-simplex-updates}
\end{align}
All four initial vectors equal $(1/2,1/2)$.

\item
Recall the quantum \eqref{eqn:duality-gap} and classical duality gap
\begin{align*}
        \gap_Q(\a, \b) & = \max_{\balt \in \bspace} \tr{U^\dagger (\a \otimes \balt)} - \min_{\aalt \in \aspace} \tr{U^\dagger (\aalt \otimes \b)}. \\
        \gap_C(x, y) &= \max_{y' \in \Delta^2} x^\top A y' - \min_{x' \in \Delta^2} {x'}^\top A y.
\end{align*}
For every $t$,
\[
\gap_Q(\a_t,\b_t)=\gap_C(x_t,y_t), \quad \text{where } \gap_C(x,y) := \max_{y'\in\Delta^2} x^\top A y' - \min_{x'\in\Delta^2} {x'}^\top A y.
\]
\end{enumerate}
\end{lemma}
For readability, we defer the proof of Lemma~\ref{lem:diag-reduction} to Appendix~\ref{appendix:lem:diag-reduction}.
We apply Lemma~\ref{lem:diag-reduction} to the hard $2\times 2$ instance from~\citet{cai2025fastlastiterateconvergencelearning}.
Fix $\delta\in(0,\tfrac12)$ and let
\[
A_\delta :=
\begin{pmatrix}
\frac12+\delta & \frac12\\
0 & 1
\end{pmatrix}.
\]
Let $U_\delta$ denote the payoff observable on the joint space
\begin{align} \label{eqn:bad observable}
    U_\delta \defeq
\diag\Big(\tfrac12+\delta,\tfrac12,0,1\Big).
\end{align}
By Lemma~\ref{lem:diag-reduction}, quantum \ommwu on $U_\delta$ with the stated maximally mixed initialization is equivalent to classical \ommwu on $A_\delta$ with the same initialization and update order, with matching duality gaps along the trajectory. We first restate the classical lower bound of \citet{cai2025fastlastiterateconvergencelearning}, with the uniform first played iterate used in their paper.
\caiLastIterateTheorem*

The diagonal reduction identifies the matrix updates with simplex updates, but their initialization differs from that in \citet{cai2025fastlastiterateconvergencelearning}. The following lemma verifies that their proof still yields a lower bound of the same order.

\begin{lemma}[Lower bound for the initialization in Algorithm~\ref{alg:ommwu}]
\label{lem:ommwu-startup-lower-bound}
There exist absolute constants $c_{\mathrm{lb}},c_{\mathrm{gap}},\delta_0>0$ such that, for every $0<\delta<\delta_0$ and fixed $0<\eta\leq1/2$, the simplex updates \eqref{eq:diag-simplex-updates} on $A_\delta$, with $x_0=y_0=\hat x_0=\hat y_0=(1/2,1/2)$, have an iterate satisfying
\[
t\geq\frac{c_{\mathrm{lb}}}{\eta\delta},
\qquad \gap_C(x_t,y_t)\geq c_{\mathrm{gap}}.
\]
\end{lemma}
The proof of Lemma~\ref{lem:ommwu-startup-lower-bound} is given in Appendix~\ref{appendix:lem:ommwu-startup-lower-bound}.

\begin{remark}
By Lemma \ref{lem:diag-reduction} and Lemma~\ref{lem:ommwu-startup-lower-bound}, there exist absolute constants $c_{\mathrm{lb}},c_{\mathrm{gap}},\delta_0>0$ such that for every $\delta \in (0, \delta_0)$, the quantum \ommwu dynamics on $U_\delta$ with any fixed step size $0\mathrel{<}\eta \leq \frac{1}{4L}$ satisfies the following: there exists an iteration $t \geq \frac{c_{\mathrm{lb}}}{\eta\delta}$ with $\gap_Q(\a_t, \b_t) \geq c_{\mathrm{gap}}$.
The constants need not equal those in the classical theorem. The parameter $\delta$ indexes this family; no identification with the error-bound modulus $\delta(\superop)$ is used.
\end{remark}

\begin{corollary}[No uniform last-iterate convergence rate]
Fix \(0<\eta\leq 1/2\). There exists an absolute constant \(c_{\mathrm{gap}}>0\) with the following property: for every iteration threshold \(T\geq 1\), there is a quantum zero-sum game with one qubit per player and \(0\preceq U\preceq I\) such that Algorithm 3, initialized at the maximally mixed state and run without the stopping test, has an iterate satisfying
\[
t\geq T, \qquad \mathcal G_Q(\alpha_t,\beta_t)\geq c_{\mathrm{gap}}.
\]

Consequently, no error bound depending only on the step size and iteration count can both tend to zero and hold for all games in this class.
\end{corollary}

\begin{proof}
Given \(T\), choose \(0<\delta<\delta_0\) small enough that
\[
\frac{c_{\mathrm{lb}}}{\eta\delta}\geq T.
\]

Lemmas 8 and 9 imply that the iterates on \(U_\delta\) satisfy
\[
\mathcal G_Q(\alpha_t,\beta_t)\geq c_{\mathrm{gap}} \quad\text{at some}\quad t\geq \frac{c_{\mathrm{lb}}}{\eta\delta}\geq T.
\]

Since \(T\) is arbitrary and \(c_{\mathrm{gap}}\) is independent of \(\delta\), a bound that tends to zero cannot hold uniformly over these games.
\end{proof}

Thus an instance-independent last-iterate rate is impossible, even at fixed dimension and bounded payoff. This conclusion does not preclude a rate with instance-dependent constants. For any $0<\varepsilon<c_{\mathrm{gap}}$, the time after which all subsequent gaps remain at most $\varepsilon$ is at least $c_{\mathrm{lb}}/(\eta\delta)$ on $U_\delta$; the result does not bound the first time the gap falls below $\varepsilon$.

The lower bound does not resolve Conjecture~\ref{conj:ommwu-linear}, which assumes a unique quantum equilibrium. The classical equilibrium of $A_\delta$ has full support, and distinct density matrices can share its equilibrium diagonals. Since a diagonal payoff ignores off-diagonal entries, these states yield multiple quantum equilibria.

\subsection{Proof for Lemma \ref{lem:diag-reduction}}
\label{appendix:lem:diag-reduction}
\begin{proof}
We prove the three statements one by one.

Firstly, we want to show that all iterates remain diagonal.
For any $\a\in\aspace$ and $\b\in\bspace$, let $x=(\a_{11},\a_{22})$ and $y=(\b_{11},\b_{22})$.
Since $U$ is diagonal, we have
\[
\payoff(\a,\b)=\tr{U^\dagger(\a\otimes\b)}
=\sum_{i,j\in\{1,2\}} A_{ij}\, \a_{ii}\,\b_{jj}
= x^\top A y.
\]
Moreover, the payoff gradients are given by partial traces, and the diagonal structure implies they are diagonal and depend only on
$(\a_{11},\a_{22})$ and $(\b_{11},\b_{22})$:
\[
\gradient_\a(\b)=\superop(\b)=\diag(Ay),
\qquad
\gradient_\b(\a)=-\adj{\superop}(\a)=-\diag(A^\top x).
\]
Now assume $\hat\a_t$ and $\b_t$ are diagonal and $\hat\a_t\succ0$. Then $\log(\hat\a_t)$ is diagonal and $\gradient_\a(\b_t)$ is diagonal, so
$\log(\hat\a_t)\mathbin{-}\eta\,\gradient_\a(\b_t)$ is diagonal. The matrix softmax
$\Lambda(Z)=\exp(Z)/\tr{\exp(Z)}$ preserves diagonality, hence $\a_{t+1}$ is diagonal.
It also yields a positive definite matrix, so every subsequent matrix logarithm is well-defined.
The same argument applies symmetrically to $\b_{t+1}$ and to $\hat\a_{t+1},\hat\b_{t+1}$.
Since $\a_0=\b_0=\hat\a_0=\hat\b_0=\eye/2$ are diagonal, an induction yields that all iterates remain diagonal.
Thus there exist $x_t ,y_t,\hat x_t,\hat y_t\in\Delta^2$ such that
\[
\a_t=\diag(x_t),\quad \b_t=\diag(y_t),\quad
\hat\a_t=\diag(\hat x_t),\quad \hat\b_t=\diag(\hat y_t),
\qquad \forall t\ge 0.
\]

Secondly, we show that the induced simplex trajectory $(x_t,y_t)$ coincides exactly with the classical \ommwu under the same initialization and update order on the $2\times 2$ game $A$.
Fix $t$ and write $\hat\a_t=\diag(\hat x_t)$, $\b_t=\diag(y_t)$.
By the previous step, $\gradient_\a(\b_t)=\diag(Ay_t)$, so the \ommwu update has the form
\[
\a_{t+1}
=\Lambda \bigl(\log(\hat\a_t)\mathbin{-}\eta\,\gradient_\a(\b_t)\bigr)
=\Lambda \Bigl(\diag\bigl(\log(\hat x_t)\mathbin{-}\eta\,Ay_t\bigr)\Bigr).
\]
For any $u\in\R^2$, $\exp(\diag(u))=\diag(\exp(u))$ and $\tr{\exp(\diag(u))}=\sum_i e^{u_i}$, hence
$\Lambda(\diag(u))=\diag(\mathrm{softmax}(u))$. Therefore,
\[
x_{t+1}[i]
=\frac{\hat x_t[i]\exp\bigl(-\eta\,(Ay_t)[i]\bigr)}
{\sum_{j=1}^2 \hat x_t[j]\exp\bigl(-\eta\,(Ay_t)[j]\bigr)},
\qquad (i\in\{1,2\}).
\]
The same derivation for Bob gives
\[
y_{t+1}[i]
=\frac{\hat y_t[i]\exp \bigl(-\eta\,(-A^\top x_t)[i]\bigr)}
{\sum_{j=1}^2 \hat y_t[j]\exp \bigl(-\eta\,(-A^\top x_t)[j]\bigr)},
\qquad (i\in\{1,2\}).
\]
These are exactly the classical \ommwu played-strategy updates
on the feedback vectors $\ell_x^t=Ay_t$ and $\ell_y^t=-A^\top x_t$.
Applying the same calculation to the auxiliary matrices, whose updates use the newly played strategies, gives the two auxiliary updates in \eqref{eq:diag-simplex-updates}. The initial vectors agree, so induction identifies the complete trajectories.

Lastly, we show that the quantum duality gap coincides with the classical duality gap along the trajectory.
Recall the quantum duality gap
\[
\gap_Q(\a,\b)
=\max_{\balt\in\bspace}\tr{U^\dagger(\a\otimes \balt)}
-\min_{\aalt\in\aspace}\tr{U^\dagger(\aalt\otimes \b)}.
\]
For any $\balt\in\bspace$, letting $y'=(\balt_{11},\balt_{22})\in\Delta^2$, the diagonal-payoff identity gives\linebreak
$\tr{U^\dagger(\a\otimes \balt)}=x^\top A y'$; conversely, for any $y'\in\Delta^2$ the diagonal state
$\balt=\diag(y')$ attains the same value. Hence
\[
\max_{\balt\in\bspace}\tr{U^\dagger(\a\otimes \balt)}
=\max_{y'\in\Delta^2} x^\top A y'.
\]
Similarly, for any $\aalt\in\aspace$, letting $x'=(\aalt_{11},\aalt_{22})\in\Delta^2$, we have
$\tr{U^\dagger(\aalt\otimes \b)}=x'^\top A y$, and every $x'\in\Delta^2$ is attained by the diagonal state $\aalt=\diag(x')$. Thus
\[
\min_{\aalt\in\aspace}\tr{U^\dagger(\aalt\otimes \b)}
=\min_{x'\in\Delta^2} x'^\top A y.
\]
Subtracting yields $\gap_Q(\a,\b)=\gap_C(x,y)$.
Applying this with $\a=\a_t=\diag(x_t)$ and $\b=\b_t=\diag(y_t)$ gives $\gap_Q(\a_t,\b_t)=\gap_C(x_t,y_t)$ for all $t$.
\end{proof}

\subsection{Proof of Lemma~\ref{lem:ommwu-startup-lower-bound}}
\label{appendix:lem:ommwu-startup-lower-bound}

\begin{proof}
We adapt the proof of \citet[Theorem~1 and Appendix~B.2]{cai2025fastlastiterateconvergencelearning}. Its mechanism is to build up cumulative loss that takes order $1/(\eta\delta)$ iterations to undo. We verify the initial separation estimate for our first played iterate; the subsequent stages use the same cumulative-loss identities.

Let $e_x^t$ and $e_y^t$ denote the differences between the two coordinates of the loss vectors, and let $E_x^t,E_y^t$ be their cumulative sums starting at $t=1$. On $A_\delta$,
\begin{align*}
e_x^t&=-\tfrac12+(1+\delta)y_t[1],
&E_x^t&=\sum_{s=1}^t e_x^s,\\
e_y^t&=1-(1+\delta)x_t[1],
&E_y^t&=\sum_{s=1}^t e_y^s.
\end{align*}
For the negative entropy regularizer, the mirror map in Theorem~\ref{thm:name} is $F_{\eta,\reg}(s)=(1+e^{\eta s})^{-1}$. The initialization in Algorithm~\ref{alg:ommwu} gives
\[
x_1[1]=F_{\eta,\reg}(\delta/2)\leq\tfrac12,
\qquad y_1[1]=F_{\eta,\reg}((1-\delta)/2)<\tfrac12.
\]
Expanding the auxiliary updates in \eqref{eq:diag-simplex-updates} yields, for every $t\geq2$,
\begin{equation}
x_t[1]=F_{\eta,\reg}(E_x^{t-1}+e_x^{t-1}),
\qquad y_t[1]=F_{\eta,\reg}(E_y^{t-1}+e_y^{t-1}).
\label{eq:startup-cumulative-loss}
\end{equation}
This is the same update formula used by \citet{cai2025fastlastiterateconvergencelearning}, but with a different first played iterate.

For the initial separation estimate, use $L=1/2$ and $c_1=1/2-F_{1,\reg}(1/(20L))>0$, as in Theorem~\ref{thm:name}. Take $\delta_0$ small enough that $\delta<\min\{1/15,c_1/2\}$. Let $T_s$ be the first time $x_t[1]\geq3/4$, and $T_1$ the first time $x_t[1]\geq1/(1+\delta)$. Since $x_1[1]\leq1/2$, both hitting times, when finite, occur at $t\geq2$. The bound $e_x^t\geq-1/2$ and the $\eta L$-Lipschitz property give
\[
\frac14
\leq x_{T_s}[1]-F_{\eta,\reg}(0)
\leq\frac{\eta L T_s}{2},
\qquad T_s\geq\frac{1}{2\eta L}.
\]
Before $T_s$, we have $e_y^t\geq1/5$; before $T_1$, we have $e_y^t\geq0$. Therefore, for $T_s\leq t\leq T_1$,
\begin{align*}
F_{\eta,\reg}(E_y^{t-1})
\leq F_{\eta,\reg}\bigl((T_s-1)/5\bigr) \leq F_{1,\reg}\bigl(1/(20L)\bigr)
=\tfrac12-c_1,
\end{align*}
where $T_s-1\geq1/(4\eta L)$ follows from $\eta\leq1/(4L)$. Equation~\eqref{eq:startup-cumulative-loss} and $e_y^{t-1}\geq0$ also give $y_t[1]\leq1/2-c_1$.

The hitting times are finite. If $T_s$ did not exist, then $e_y^t\geq1/5$ for all $t$, so \eqref{eq:startup-cumulative-loss} would force $y_t[1]\to0$ and then $x_t[1]\to1$, a contradiction. If $T_1$ did not exist, the preceding separation estimate would hold for all $t\geq T_s$ and imply $e_x^t\leq-c_1/2$, again forcing $x_t[1]\to1$. We have therefore established the entry conditions for Stage II of \citet{cai2025fastlastiterateconvergencelearning}:
\[
x_{T_1}[1]\geq\frac{1}{1+\delta},
\qquad F_{\eta,\reg}(E_y^{T_1-1})\leq\frac12-c_1.
\]

For completeness, the remaining stages require only these entry estimates, \eqref{eq:startup-cumulative-loss}, the bounds $-1/2\leq e_x^t\leq1/2+\delta$ and $-\delta\leq e_y^t\leq1$, and the properties of the negative entropy mirror map. All hold here. In Stage II, $e_y^t\geq-\delta$ keeps $y_t[1]\leq1/2-c_1/2$ throughout $T_1\leq t<T_1+H$, where $H=\lfloor c_1/(2\eta L\delta)\rfloor$. During those rounds, $e_x^t\leq-c_1/4$, accumulating a negative loss of order $1/(\eta\delta)$. Let $T_2$ be the first subsequent time $y_t[1]\geq1/(2(1+\delta))$. The Stage II argument gives $T_2\geq T_1+H$ and proves that $T_2$ is finite. Stage III then uses the accumulated negative loss to keep Alice near the boundary while Bob moves past equilibrium. Applying those estimates, after possibly shrinking $\delta_0$, yields a later time $t\geq T_2$ and an absolute constant $c_{\mathrm{gap}}>0$ such that
\[
x_t[1]\geq\frac{1}{1+\delta},
\qquad y_t[1]\geq\frac12+c_{\mathrm{gap}}.
\]
These are the same estimates as in Stages II and III of \citet[Appendix~B.2]{cai2025fastlastiterateconvergencelearning}; their proof no longer uses a uniform first played iterate. Since $H\geq c_{\mathrm{lb}}/(\eta\delta)$ for an absolute $c_{\mathrm{lb}}>0$ when $\delta_0$ is sufficiently small, the time bound follows. At the resulting iterate,
\begin{align*}
\gap_C(x_t,y_t)
=(\tfrac12+\delta)x_t[1]-(1-y_t[1]) \geq y_t[1]-\frac{1}{2(1+\delta)}
\geq c_{\mathrm{gap}},
\end{align*}
which completes the adaptation.
\end{proof}



\newpage
\section{Application: Parallel Approximation of Strictly
Positive Semidefinite Programs}
\label{appendix:approximation}
From the seminal work of \citet{jain2009parallel}, a fundamental connection was established between equilibrium states of quantum zero-sum games and positive semidefinite programs (SDPs). In particular, their parallel algorithm computes $\varepsilon$-approximate equilibria in quantum zero-sum games with iteration complexity $\bigoh_d(1/\varepsilon^2)$, and these equilibria can in turn be leveraged to approximate solutions of certain SDPs. Building on this connection, we show that the $\bigoh_d(\log(1/\varepsilon))$ convergence-rate algorithms developed in Section~\ref{sec:algorithms} yield an exponential improvement: $(1+\varepsilon)$-approximate solutions to SDPs can be obtained with logarithmic iteration complexity whenever the associated game admits a full-rank equilibrium, thereby exponentially speeding up the bound of prior parallel approaches. To set the stage for the main result of the section, we first collect the key definitions required for its formulation.
\paragraph{Superoperator Semidefinite Programs.}
Let $\a \in \C^{n\times n}$ and $\b \in \C^{m\times m}$ be Hermitian matrices, let $\superop \from \C^{m\times m} \to \C^{n\times n}$ be a linear superoperator that preserves Hermiticity. Given an instance tuple $\parens{\a, \b, \superop}$, the \emph{superoperator form} of the primal/dual semidefinite program is
\begin{align}
    \textbf{Primal} & \quad \max_{y \in \C^{n\times n}} \tr{\b y} \qquad \text{subject to } \superop(y) \mleq \a, \quad y \mgeq 0 \label{eqn:sdp-primal} \\
    \textbf{Dual} & \quad \min_{x \in \C^{m\times m}} \tr{\a x} \qquad \text{subject to } \adj{\superop}(x) \mgeq \b, \quad x \mgeq 0 \label{eqn:sdp-dual}.
\end{align}
This form is equivalent to the \emph{standard form} of semidefinite programs. We refer to it as the \emph{superoperator semidefinite program}. We use superoperator form to connects quantum zero-sum game with semidefinite programs.

\paragraph{Strong Duality via Slater’s Condition.} In contrast with LP, strong duality always holds whenever the primal (or dual) problem is feasible for the case of SDP, it does not always hold automatically. \emph{Slater’s condition} is a standard constraint qualification in convex optimization which ensures \emph{strong duality}. It requires the existence of a strictly feasible solution that satisfies all inequality constraints with strict inequalities. Applied to our setting, if either the primal or the dual program admits a strictly feasible state, then by Slater's condition, strong duality holds: the optimal values of \eqref{eqn:sdp-primal} and \eqref{eqn:sdp-dual} coincide and are attained. In particular, for the primal–dual pair \eqref{eqn:sdp-primal} and \eqref{eqn:sdp-dual}, strong duality holds if at least one of the following conditions is satisfied:
\begin{enumerate}
  \item[(i)] there exists $Y \mg 0$ with $\superop(Y) \ml A$,
  \item[(ii)] there exists $X \mg 0$ with $\adj{\superop}(X) \mg B$.
\end{enumerate}

\paragraph{Positive and Strictly Positive Instances.} We call an SDP instance $(\a, \b, \superop)$ positive if $\a \mgeq 0$, $\b \mgeq 0$, and $\superop$ is a positive superoperator, meaning that whenever $x \mgeq 0$, we also have $\superop(x) \mgeq 0$. We call the instance strictly positive if $\a \mg 0$, $\b \mg 0$, and $\superop$ is strictly positive, which is defined by the condition $\superop(\eye) \mg 0$.

Observe that every strictly positive instance satisfies Slater's condition. Indeed, for the primal program, pick $t > 0$ small enough that $\superop(t \eye) = t\superop(\eye) \ml \a$ yielding strictly feasibility for \eqref{eqn:sdp-primal}. Similarly, for the dual program, pick $s > 0$ large enough that $\adj{\superop}(s \eye) = s\adj{\superop}(\eye) \mg \b$ yielding strictly feasibility for \eqref{eqn:sdp-dual}. Therefore, strong duality holds for all strictly positive superoperator semidefinite program.

With strong duality for strictly positive SDPs in place, we next present the reduction to \emph{trace form} and show how the methods of Section~\ref{sec:algorithms} can be deployed.

\paragraph{From Strictly Positive Instances to Trace Form.} The \itersmooth algorithm \ref{alg:iterative smoothing} applies once any strictly positive superoperator semidefinite programs with instance $(\a, \b, \superopalt)$ has been normalized into the following \emph{trace form}
\begin{align}
    \textbf{Primal} & \quad \max_{y \in \C^{n\times n}} \tr{y} \qquad \text{subject to } \superop(y) \mleq \eye, \quad y \mgeq 0 \label{eqn:sdp-primal-trace} \\
    \textbf{Dual} & \quad \min_{x \in \C^{m\times m}} \tr{x} \qquad \text{subject to } \adj{\superop}(x) \mgeq \eye, \quad x \mgeq 0 \label{eqn:sdp-dual-trace} \\
    \text{with } & \superop(y) = \a^{-\frac{1}{2}} \superopalt\parens{\b^{-\frac{1}{2}} y \b^{-\frac{1}{2}}} \a^{-\frac{1}{2}}. \notag
\end{align}

Since $\superopalt$ is strictly positive, we have $\superop(\eye) \mg 0$, and hence $\superop$ is strictly positive as well. By Slater’s condition, strong duality holds: the common optimal value $\opt(\superop)$ of the primal and dual programs is attained and satisfies $\opt(\superop)>0$. 

\paragraph{Connection between Quantum Zero-Sum Games and Superoperator SDPs.}
The following proposition, adapted from \citet{jain2009parallel} establishes the connection between quantum zero-sum games and strictly positive superoperator semidefinite programs. In particular, it shows how an $\varepsilon$-approximate Nash equilibrium can be transformed into an approximate solution for this program. The connection is formalized in the following proposition.
\begin{proposition}[Proposition 5 in \citet{jain2009parallel}]\label{prop:connection-sdp-gameval}
    For any strictly positive superoperator $\superop$ whose game admits a full-rank Nash equilibrium, we have $\val(\superop) = 1 / \opt(\superop)$.
\end{proposition}
\begin{proof}
Write $v=\val(\superop)$. The saddle-point inequalities give $\superop(\bopt)-v\eye\mgeq0$ and $v\eye-\adj{\superop}(\aopt)\mgeq0$, together with $\tr{\aopt(\superop(\bopt)-v\eye)}=0$ and $\tr{\bopt(v\eye-\adj{\superop}(\aopt))}=0$. A positive semidefinite matrix whose trace pairing with a positive definite matrix is zero must vanish. Hence
\[
    \superop(\bopt)=v\eye,\qquad \adj{\superop}(\aopt)=v\eye.
\]
Taking the trace of the second identity gives $mv=\tr{\aopt\superop(\eye)}>0$, since $\superop(\eye)\mg0$; hence $v>0$. Define $y^\star=\bopt/v$ and $x^\star=\aopt/v$. Then $\superop(y^\star)=\eye$ and $\adj{\superop}(x^\star)=\eye$, so $y^\star$ and $x^\star$ are feasible for \eqref{eqn:sdp-primal-trace} and \eqref{eqn:sdp-dual-trace} with equal objective values $\tr{y^\star}=\tr{x^\star}=1/v$. By weak duality, both are optimal, and therefore $\opt(\superop)=1/v$. The two identities also give $\payoff(\a,\bopt)=\payoff(\aopt,\b)=v$ for all $\a$ and $\b$, so the value does not depend on which player minimizes.
\end{proof}
\noindent At the end of \itersmooth algorithm, we have $\gap(\joint) < \varepsilon$, which implies that
\[    \primal(\a)  \leq \val(\superop) + \varepsilon,    \dual(\b)  \geq \val(\superop) - \varepsilon.
\]
Suppose the approximate value of the game $\wilde{\val}(\superop) = \frac{\primal(\a) + \dual(\b)}{ 2 }$. We have
\begin{align*}
    \abs{\wilde{\val}(\superop) - \val(\superop)} \leq \frac{\primal(\a) - \dual(\b)}{ 2} \leq \frac{\varepsilon}{2}.
\end{align*}
Then, the approximated value of the game is between
\begin{align*}
    \wilde{\val}(\superop) \in \bracks*{\parens*{1 - \frac{\varepsilon}{2 \val(\superop)}} \val(\superop), \; \parens*{1 + \frac{\varepsilon}{2 \val(\superop)}}\val(\superop)}
\end{align*}
Using Proposition \ref{prop:connection-sdp-gameval}, let $\xi = \frac{\varepsilon}{2 \val(\superop)}$ the approximated optimal value falls between
\begin{align*}
    \wilde{\opt}(\superop) \in \bracks*{\parens*{1 + \xi}^{-1} \opt(\superop), \; \parens*{1 - \xi}^{-1} \opt(\superop)}
\end{align*}
\paragraph{Complexity Comparison.} 
The Optimistic Matrix Mirror Prox \texttt{OMMP} algorithm of \citet{Vasconcelos2025quadraticspeedupin} achieves a convergence rate of $\bigoh_d(1/\varepsilon)$ for computing an $\varepsilon$-approximate Nash equilibrium, leading to a total complexity of
\begin{align*}
    \bigoh_d(\frac{1}{2\gameval} \cdot \frac{1}{\xi}),
\end{align*}
for obtaining an approximation of strictly positive semidefinite programs

The \itersmooth algorithm attains a convergence rate of $\bigoh(\kappa(\superop)\ln(1/\varepsilon))$, which yields a total complexity of
\begin{align*}
    \bigoh\parens*{\kappa(\superop) \parens*{\ln(\frac{1}{2\gameval}) + \ln(1/\xi)}},
\end{align*}
and the \ogda algorithm attains a convergence rate of $\bigoh_d(\log(1/\varepsilon))$ yielding a total complexity of 
\begin{align*}
    \bigoh_d \parens*{\log(\frac{1}{2\gameval}) + \log(1/\xi)}
\end{align*}
which establishes an exponential speedup in iteration complexity.


\end{document}